%% file: main.tex
\pdfoutput=1
\documentclass[11pt,a4paper]{article}

\usepackage{dqinotation}

\usepackage[T1]{fontenc}
\usepackage{microtype}

\usepackage{booktabs}
\usepackage{array}
\newcolumntype{L}[1]{>{\raggedright\arraybackslash}p{#1}}
\usepackage{enumitem}
\setlist[enumerate]{leftmargin=*,itemsep=3pt,topsep=4pt}
\setlist[itemize]{leftmargin=*,itemsep=3pt,topsep=4pt}
\usepackage{tikz}
\usetikzlibrary{arrows.meta,calc}
\usepackage{pgfplots}
\pgfplotsset{compat=1.17}

\usepackage[margin=1in]{geometry}
\usepackage{xcolor}   %
\usepackage[colorlinks=true,linkcolor=black!70!blue,citecolor=black!60!green,
            urlcolor=black!70!blue,linktoc=all,
            pdftitle={Discovering New Problems for Decoded Quantum Interferometry}]{hyperref}

\title{\textbf{Discovering New Problems for Decoded Quantum Interferometry}}
\usepackage{authblk}

\author[1]{Rhik Mazumder}
\author[1]{Pradeep Niroula\thanks{Correspondence:
  \href{mailto:pradeep@magicstatelabs.com}
       {pradeep@magicstatelabs.com}}}
\affil[1]{Magic State Labs, New York, NY}

\date{}

\begin{document}
\maketitle
\begin{abstract}
Decoded quantum interferometry (DQI) uses classical error correction to solve optimization problems, but it is hard to tell which problems could give it a quantum advantage. We develop soundness rules that screen candidate problems by their code parameters, decoding guarantees and classical baselines. For the standard objective of counting satisfied constraints, we prove that under worst-case Hamming decoding DQI's guarantee can exceed Prange completion by at most $(\sqrt2-1)/2$. Optimal polynomial intersection (OPI) attains this margin cap asymptotically, so no candidate in this setting can exceed OPI's margin. We also derive a performance law for real-valued objectives: the objective fixes both the solution quality and which errors the decoder must correct. An AI-agent search guided by these rules found two candidate applications. Alternant max-agreement on McEliece-type public keys tests what a secret decoder is worth to DQI against classical solvers that see only the public instance. Multiplicative polynomial intersection, a knapsack-type problem built from discrete logarithms, pairs a cosine objective with a public decoder for signed errors. For both we prove decoding guarantees and finite-length comparisons with Prange restarts, and on the larger instances we tested, DQI's expected quality exceeds what our classical solvers reach within fixed budgets.
\end{abstract}
\bigskip

\input{sections/01-introduction}
\input{sections/02-setup}

\input{sections/03-rules}
\input{sections/04-problems}
\input{sections/05-discussion}
\input{supplement/code-availability-snippet}
\bibliographystyle{plain}
\bibliography{refs,refs-extra}

\appendix
\clearpage
\part*{Supplementary material}
\small
\renewcommand{\topfraction}{0.9}
\renewcommand{\textfraction}{0.08}
\setlength{\textfloatsep}{12pt plus 2pt minus 4pt}
\input{supplement/S0-notation}
\input{supplement/S1-proofs}
\input{supplement/S2-candidates}
\input{supplement/S3-search-checks}
\input{supplement/S4-scaling}

\end{document}

%% file: sections/01-introduction.tex
\section{Introduction}\label{sec:intro}

Jordan et al.\ introduced decoded quantum interferometry (DQI)~\cite{jordan24dqi}, a quantum
algorithm whose quality on an optimization problem is set by how well a related code can be
decoded. Their strongest example is optimal polynomial intersection (OPI): find a
low-degree polynomial whose values fall in prescribed sets at as many points as possible. Every other
candidate has been studied on its own terms, with a code, a decoder and a classical
comparison built for that one problem, so it is hard to tell in advance which problems deserve a closer look. We
use the structure of DQI to write down explicit screening criteria,
and these make a computer-aided search for promising
optimization problems possible.

DQI is usually applied to max-LINSAT~\cite{jordan24dqi}.
An instance is a matrix $B\in\F_q^{m\times n}$ together with an
\emph{accepted set} $F_i\subseteq\F_q$ for each row $i$. An assignment
$x\in\F_q^n$ satisfies row $i$ when $(Bx)_i\in F_i$, and the goal is
to satisfy as many rows as possible:
\[
\max_{x\in\F_q^n} S(x),
\qquad
S(x):=\sum_{i=1}^{m}\mathbf{1}\!\left\{(Bx)_i\in F_i\right\},
\]
where the indicator is one when the condition holds and zero
otherwise. Fixing a family for $B$ gives a structured
problem. OPI, for example, uses Vandermonde matrices
$B_{ij}=a_i^{j-1}$ at distinct points $a_i$, so
$x$ holds the coefficients of a polynomial and $(Bx)_i$ is its
value at $a_i$.

The matrix $B$ defines two dual linear codes: the \emph{optimizer code}
$\operatorname{im} B$ and the \emph{decoder code}
$\ker B^{\mathsf T}$. DQI prepares a superposition over sparse error
patterns, computes their syndromes under $B^{\mathsf T}$, and then erases
the patterns by running a classical decoder for $\ker B^{\mathsf T}$ coherently
(Section~\ref{sec:dqidecoder}). The minimum distance of this code caps how
many errors can be corrected in the worst case, and with it
the solution quality DQI can guarantee.

The classical baseline instead builds a word in
$\operatorname{im} B$ directly. Suppose $B$ has $m$ rows and $n$
linearly independent columns. A \emph{Prange completion}~\cite{prange62}
picks $n$ linearly independent rows, called an information set,
and an accepted value for each of them. Solving
these $n$ linear equations gives $x$, so the picked rows are
satisfied by construction. Computing $Bx$ then fixes
the remaining $m-n$ coordinates.
This step is the ``completion,'' and it needs no decoding.
If the information set is fixed in advance and the accepted sets
are independent and random with density $\rho$, each
remaining coordinate is satisfied with probability $\rho$. A single
attempt therefore satisfies an expected fraction
$n/m+\left(1-n/m\right)\rho$. The optimizer-code rate $n/m$ is the fraction
Prange completion satisfies deliberately, and the rest are satisfied
at the rate of a random guess
(Section~\ref{sec:qualitylaws}).

Figure~\ref{fig:threematrices} shows how the structure of $B$ shapes
decoding, and with it the guarantees DQI can offer, for OPI, max-XOR-SAT
and MaxCut. In OPI, the Vandermonde structure makes the decoder code a Reed--Solomon
code, which has an efficient worst-case decoder up to its
unique-decoding radius. This is the setting of the quantum advantage
claimed in~\cite{jordan24dqi}. In max-XOR-SAT with sparse random $B$,
the decoder code is a low-density parity-check code, and the
decoding results cited here establish average-case success rather than
a comparable worst-case guarantee~\cite{jordan24dqi,quantumdecoder26}.
In MaxCut, $B$ is a graph incidence matrix and the decoder code is
the graph's cycle space. Its minimum distance equals the girth, which
the Moore bound limits to $O(\log m)$ when every vertex has degree at
least three. The examples separate two constraints. The code's
distance limits its worst-case unique-decoding radius, and the
available decoding algorithm determines which errors can be corrected
efficiently.

\input{figures/three-matrices}

We turn the structure of DQI into four soundness rules for screening
candidate optimization problems. The rules state necessary
conditions for DQI to outperform specified classical baselines and
mark where our guarantees apply. A candidate that violates one of these conditions
has no corresponding advantage under its stated hypotheses,
and the violated condition shows which feature of the candidate must change.

These rules give a cheap first stage of evaluation for
computer-guided problem discovery. Evaluator-guided search has let
language models discover new constructions and
algorithms~\cite{romera24funsearch,novikov25alphaevolve}. For DQI, however,
evaluation can be costly, since it means building decoders and testing classical attacks,
and the rules screen proposals before those
steps. We used them in a human-guided search. A network of AI agents proposed problem
statements and code constructions and checked each against the rules. For candidates that
satisfied them and that a human chose to pursue, the agents built decoders and ran classical attacks
(Supplementary Section~\ref{sec:search}). Rule \gate{4} and the
objective-alphabet form of the rules were added during this search.

This search produced two candidate applications that survived the
checks and show advantages over the classical baselines studied
here. The first, \emph{alternant max-agreement} (AMA), asks for a
word in a public linear code that lies in as many of the
coordinate-wise accepted sets as possible. A secret key reveals the
code's alternant structure and lets DQI decode efficiently,
but the classical baseline receives only the public instance (the public matrix and the accepted sets).
The second, \emph{multiplicative polynomial intersection} (MPI),
is a cosine-optimization problem whose matrix entries are discrete
logarithms of finite-field elements. Its cosine objective has Fourier
support on $\{\pm1\}$, so DQI only has to decode sparse signed
errors, and a polynomial-lattice decoder can recover these
efficiently.

\section{Contributions}\label{sec:contributions}

Our contributions combine general limits on DQI with an
analysis that supports two new candidate applications. The screening
rules and a sharp margin cap set out the constraints that standard Hamming
decoding imposes, and the objective-dependent analysis shows how to
change the decoding requirements. The two candidates use these results
in different ways. AMA keeps Hamming decoding but gives only DQI access
to the decoder, and MPI changes the objective so that a signed-error
decoder suffices.

\paragraph{Soundness rules for candidate DQI advantages.}
We derive a screening framework for DQI with worst-case decoding. It
expresses the reach needed to outperform Prange completion as a
closed-form function of the optimizer-code rate and the accepted-set
density. For random accepted sets in the long-instance limit, this
comparison turns upper bounds on decoding reach into necessary
conditions for an advantage. The framework combines distance and
packing bounds with stronger classical baselines based on optimal
encoding and subfield linear algebra, and it requires
an efficient licensed decoder, one certified to correct every error up to a given weight
(Definition~\ref{def:licensed}), for its guarantees to apply
(Theorem~\ref{thm:master}, Corollary~\ref{cor:rules} and
Proposition~\ref{prop:remaining}).

\paragraph{A sharp margin cap over Prange completion.}
For indicator objectives, which count satisfied rows, and worst-case Hamming decoding, we prove
that DQI's asymptotic expected-quality margin over one Prange
completion is at most
$\kappa(\alpha)=\sqrt{(\alpha/2)(1-\alpha/2)}-\alpha/2$,
where $\alpha$ is the optimizer-code rate. Maximizing over the rate
gives the universal cap $(\sqrt{2}-1)/2$. OPI attains
the fixed-rate cap asymptotically as the field size grows and the accepted-set
density tends to $\tfrac12$. Its known performance is therefore the best
possible within the class covered by the
Hamming-license analysis, and it gives a benchmark for other
candidate constructions
(Proposition~\ref{prop:cap}).

\paragraph{Objective-dependent decoding guarantees.}
We extend the analysis from counting objectives to real-valued ones, in which
every row uses the same function up to a shift (over prime fields). Two
properties of this function matter. Its moments give an exact
formula for DQI's expected score at every length.
Its Fourier support determines which error values the decoder has to correct
(the error alphabet), so a decoder built for just those values is enough. For the cosine
objective the support is $\{\pm1\}$, and decoders for these signed
errors can certify reach beyond half the minimum Hamming distance
(Lemma~\ref{lem:objlaw} and
Proposition~\ref{prop:alphabet}).

\paragraph{Alternant max-agreement with hidden decoding structure.}
We study AMA instances built from McEliece-type
public keys. The secret key exposes the alternant structure
needed for Hamming decoding, which gives DQI a guarantee and,
at the studied parameters, an advantage over a
finite-budget Prange-restart baseline
(Section~\ref{sec:mckey} and Theorem~\ref{thm:converse}). Against classical
searches that see only the public instance,
the comparison turns in DQI's favor between $m=6000$ and $m=14000$, and
at $m=14000$ and $20000$ the searches we tested
stay below DQI even when given the key. With the key, the
problem becomes OPI with trace-defined accepted sets, so AMA
lets us ask how quantum and classical
optimization use the same decoding representation. Our hardness
assumptions concern solvers that see
only the public instance
(Section~\ref{sec:mchardness}).

\paragraph{Multiplicative polynomial intersection with a public signed-error decoder.}
We introduce MPI, a cosine-optimization problem whose matrix entries are
discrete logarithms of finite-field elements. Its public
polynomial-lattice decoder uses rational reconstruction to recover
signed errors, which are exactly the errors the cosine
objective produces. This gives the restricted-error license that
the objective-dependent law requires, and with it a rigorous DQI quality
guarantee, with quantum and classical solvers both receiving the
decoding structure
(Theorem~\ref{thm:logcauchy} and
Proposition~\ref{prop:alphabet}).
Against the tested classical solvers, including lattice reduction, the comparison
turns in DQI's favor between $m=1000$ and $m=2000$, and DQI's exact expected
score is higher at $m=2000$ and $6000$.
Section~\ref{sec:logcauchy} gives the construction, these comparisons
and the associated hardness assumption.

\paragraph{Scaling of the classical searches.}
Supplementary Section~\ref{sec:supp-scaling} studies how the classical
scores grow with the search budget for OPI and both candidates.
Counting search moves, we compare each score with a Gaussian reference
for independent Prange sampling with the same number of codeword
evaluations. The
tested generic searches beat this estimate by at most about $1.4$
standard deviations of a single completion. In the same units, DQI's
margin over the mean Prange score grows like $\sqrt m$ at fixed rates.
This comparison sets the length at which we state the
AMA hardness assumptions, and it gives extrapolated
cost estimates for both candidates.

\section{Related work}\label{sec:related}

Our work builds on the DQI framework of Jordan et al.~\cite{jordan24dqi}.
They introduce the algorithm, identify the decoding problem it relies on,
derive the Fourier amplitudes of the error register, and express its
expected quality through the largest eigenvalue of a tridiagonal
matrix. Their asymptotic semicircle law connects decoding reach to
optimization quality. They also use coding bounds to limit the
regimes in which DQI can outperform classical algorithms, including a
Shannon-bound exclusion for max-XOR-SAT~\cite[\S13.1]{jordan24dqi},
and they compare DQI with Prange's algorithm on
OPI~\cite[eqs.~(16)--(17)]{jordan24dqi}.

Our reach threshold follows by inverting their quality law, and
rule \gate{1} turns coding-theoretic bounds into necessary conditions
that apply across problem families. The margin cap of Proposition~\ref{prop:cap}
is sharp within this framework. For
indicator objectives and worst-case Hamming decoding, it limits the
asymptotic margin over one Prange completion at each optimizer-code
rate and uniformly over rates, and OPI attains both caps in the
appropriate limits. OPI's performance and its comparison
with Prange were already known. What is new is that this performance
is the best possible across all matrix families covered by the
Hamming-license analysis.

Several works limit what DQI can achieve on specific classes of
instances. Some bound its quality directly. Kramer, Schubert and Eisert derive a $q$-ary Shannon
bound~\cite[\S4.3]{kramer26bounded} and show that, on instances of bounded degree $D$, DQI with
classical decoders reaches at most
$\rho+O(1/\sqrt{D\log D})$. They also show that,
unless $\mathrm{P}=\mathrm{NP}$, no efficient algorithm can guarantee
a constant improvement over $\rho$ on worst-case
max-LINSAT~\cite{kramer26tight}. Others rule out particular families. Parekh shows that MaxCut instances
admitting a nontrivial DQI guarantee can be solved exactly by
classical algorithms~\cite{parekh25maxcut}, and Anschuetz et al.\
identify spin-glass transitions that obstruct DQI on random local
instances~\cite{anschuetz25spinglass}. A third line separates the existence of good codewords
from their efficient construction. Marwaha et al.\ make this distinction by connecting binary
DQI to the covering-radius bound through code duality~\cite{marwaha25macwilliams},
and our rule \gate{3} builds on it, comparing DQI with efficient encoders for the optimizer
code at the rate--distortion limit of Chailloux and
Tillich~\cite{chailloux24rd}.

Thelen and Mauerer's DQI-Kit is particularly close to our goal of
screening candidate problems~\cite{thelen26dqikit}. It evaluates DQI
numerically, detects encodings that compromise code distance, and
bounds the decoder code's distance using short linear dependencies. Such
bounds can feed rule \gate{1}(a), and our framework
adds closed-form conditions and comparisons
with stronger classical baselines. Curve-based extensions of
OPI~\cite{matthews26cofi} fit the same picture: they evaluate functions
on an algebraic curve instead of polynomials on a finite field. This changes the code, but
DQI's quality is still set by how far its decoder can correct
errors.

Other work extends DQI beyond counting objectives.
Bu, Gu and Li derive a weighted tridiagonal quality law for
blocks of $\pm1$-valued objectives~\cite{bgl26multivariate}, and
Wang uses a kernel to shape the DQI state for a global real-valued
objective~\cite{wang25kernelized}. Jordan et al.\ already observe that the
objective's Fourier support decides which errors enter the error
register~\cite[\S8.2]{jordan24dqi}.
We treat the general case in which every row applies the same
real-valued function up to a shift. For this case we give an exact quality formula
at every length (Lemma~\ref{lem:objlaw}) and a separation condition under which
a decoder for only those errors still yields a guarantee. The signed-error
guarantee for the cosine objective is one example
(Proposition~\ref{prop:alphabet}).

On the implementation side, Patamawisut, Benchasattabuse, Hajdu\v{s}ek and
Van Meter build a DQI circuit whose decoding stage is reversible Gauss--Jordan
elimination, with gate count and depth growing
roughly as $m^2$~\cite{patamawisut25circuit}. On OPI itself, Horinaga and Yamakawa,
and independently Jo, give quantum algorithms that outperform DQI~\cite{horinaga26,jo26},
so DQI is one of several
quantum approaches to this problem. Our margin cap
(Proposition~\ref{prop:cap}) concerns the specified DQI guarantee,
not quantum algorithms for OPI in general.

Decoding inside quantum Fourier sampling has precedents
beyond DQI. Regev's reduction~\cite{regev09}, the filtering
construction of Chen, Liu and Zhandry~\cite{chen22filtering}, and
the reduction of Debris-Alazard, Remaud and
Tillich~\cite{debris24quantumreduction} each use decoding for a restricted
error set as one step of a quantum algorithm. Yamakawa and
Zhandry obtain a verifiable quantum advantage using list recovery
of folded Reed--Solomon codes~\cite{yamakawa22verifiable}.

Cryptographic trapdoors also appear in the proofs of quantumness
of Brakerski et al.\ and Kahanamoku-Meyer et al., where a classical
verifier holds the trapdoor and checks a quantum prover that lacks
it~\cite{brakerski18certifiable,kahanamoku22classically}.
Our AMA construction reverses this role: the quantum solver holds the
trapdoor, which gives DQI its decoder. To our knowledge, this use of a cryptographic trapdoor
as DQI's decoding structure has not been analyzed before.

%% file: figures/three-matrices.tex
\begin{figure}[h]
\centering
\begin{tikzpicture}[
  x=1mm, y=1mm,
  box/.style={draw, rounded corners=2pt, align=center, font=\footnotesize,
              text width=24mm, minimum height=12mm, inner sep=2pt},
  holds/.style={box, draw=black!55},
  breaks/.style={box, draw=red!70!black, very thick, fill=red!6},
  moot/.style={box, draw=black!25, dashed, text=black!40},
  good/.style={box, draw=green!40!black, thick, fill=green!7},
  bad/.style={box, draw=black!50, fill=black!7},
  head/.style={font=\footnotesize\itshape, align=center, text width=26mm},
  row/.style={font=\footnotesize\bfseries, anchor=east},
  arr/.style={-{Latex[length=1.5mm]}, black!55},
  mootarr/.style={-{Latex[length=1.5mm]}, black!25, dashed},
  badge/.style={circle, fill=white, draw, inner sep=0.4pt, font=\scriptsize},
]
  \node[head, anchor=east, align=right, text width=22mm] at (-15,11) {constraint matrix};
  \node[head] at (0,11)  {problem};
  \node[head] at (31,11) {decoder code $\deccode$};
  \node[head] at (62,11) {worst-case decoder};
  \node[head] at (93,11) {outcome};

  \node[row] at (-15,0) {Vandermonde $B$};
  \node[holds] (a1) at (0,0)  {optimal polynomial\\ intersection};
  \node[holds] (b1) at (31,0) {Reed--Solomon, meets Singleton};
  \node[holds] (c1) at (62,0) {Berlekamp--Massey to half the distance};
  \node[good]  (d1) at (93,0) {advantage claimed};

  \node[row] at (-15,-17) {sparse random $B$};
  \node[holds]  (a2) at (0,-17)  {max-XOR-SAT};
  \node[holds]  (b2) at (31,-17) {LDPC code; distance depends on ensemble};
  \node[breaks] (c2) at (62,-17) {belief propagation: average case only};
  \node[bad]    (d2) at (93,-17) {no advantage shown};

  \node[row] at (-15,-34) {graph incidence $B$};
  \node[holds]  (a3) at (0,-34)  {MaxCut};
  \node[breaks] (b3) at (31,-34) {cycle space: $d=$ girth, $O(\log m)$ at degree $\ge3$};
  \node[moot]   (c3) at (62,-34) {not reached};
  \node[bad]    (d3) at (93,-34) {asymptotic Hamming-license advantage excluded};

  \foreach \r in {1,2} {
    \draw[arr] (a\r) -- (b\r);
    \draw[arr] (b\r) -- (c\r);
    \draw[arr] (c\r) -- (d\r);
  }
  \draw[arr]     (a3) -- (b3);
  \draw[mootarr] (b3) -- (c3);
  \draw[mootarr] (c3) -- (d3);

  \foreach \n in {a1,b1,c1,a2,b2,a3}
    \node[badge, draw=black!55, text=black!75] at (\n.north east) {$\checkmark$};
  \foreach \n in {c2,b3}
    \node[badge, draw=red!70!black, text=red!70!black] at (\n.north east) {$\times$};
\end{tikzpicture}
\caption{DQI on three constraint matrices. In the max-XOR-SAT and MaxCut rows, a cross marks the step where the argument
for an advantage fails.}
\label{fig:threematrices}
\end{figure}
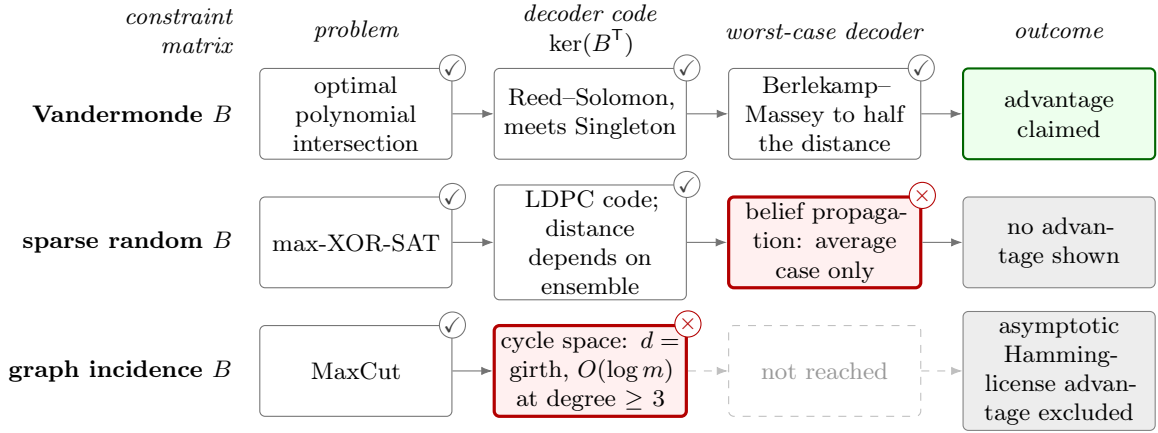

%% file: sections/02-setup.tex
\section{Soundness rules for DQI candidates}\label{sec:problem}

This section develops screening rules that state necessary conditions for
DQI to outperform specified classical baselines. After setting up
notation, we find the decoding reach needed
for an advantage and derive the rules
(Theorem~\ref{thm:master}, Corollary~\ref{cor:rules} and
Proposition~\ref{prop:remaining}). We then prove the margin cap, a sharp
upper bound on the asymptotic margin over Prange completion that
OPI attains (Proposition~\ref{prop:cap}),
discuss the cost of evaluating the rules, and give a
finite-length guarantee against Prange restarts. Throughout,
the objective counts satisfied rows and the decoder is
licensed for Hamming errors. Section~\ref{sec:rulesalphabet}
extends the analysis to objectives whose Fourier support permits
a weaker decoding requirement.

\subsection{Preliminaries}\label{sec:instances}
\label{sec:dqidecoder}\label{sec:qualitylaws}

\paragraph{Instances and their two codes.}
The optimization problems we consider ask for an assignment satisfying as many modular linear
constraints as possible. Each constraint specifies a set of accepted values for one row of a
matrix.

\begin{definition}[instance]\label{def:instance}
An \emph{instance} over $\F_q$ consists of a matrix $B\in\F_q^{m\times n}$ and an
\emph{accepted set} $F_i\subseteq\F_q$ for each row $i$. The \emph{quality} of an assignment
$x\in\F_q^n$ is
\[
  Q(x)=\frac{1}{m}\sum_{i=1}^{m}\mathbf{1}\{(Bx)_i\in F_i\},
\]
and the task is to maximize $Q(x)$.
\end{definition}

We assume that $B$ has no zero rows and that the accepted sets have a common
\emph{density} $\rho=|F_i|/q$, with $0<\rho<1$. A uniformly random assignment then has
expected quality $\rho$. Whenever we compare algorithms over random accepted sets, the
sets are drawn independently and uniformly among subsets of their prescribed size,
independently of $B$.

The matrix defines two dual linear codes of length $m$: the
\emph{optimizer code} $\optcode$, spanned by the columns of $B$, and the
\emph{decoder code} $\deccode=\optcode^{\perp}$. We write
$\alpha=\rank(B)/m$ for the optimizer-code rate, so the decoder code has rate $1-\alpha$, and $d$
and $\delta=d/m$ for the decoder code's minimum distance and relative distance.

\begin{definition}[family]\label{def:family}
A \emph{family} is a sequence of instances with $m\to\infty$ along which the optimizer-code
rate, decoder-code relative distance, and accepted-set density converge. We use
$\alpha$, $\delta$, and $\rho$ for their limits as well. The field size $q$ may vary along
the sequence.
\end{definition}

The asymptotic rules concern such families. Finite-length comparisons instead use the
parameters and exact DQI quality of the individual instance. In particular, the constructions
in Section~\ref{sec:problems} have a maximum length at any fixed choice of field parameters,
and their reported comparisons use the lengths actually studied.

\paragraph{Decoding reach.}
In DQI's superposition the error patterns $y\in\F_q^m$ have Hamming weight at most $\ell$
and amplitudes set by the accepted sets, and a final Fourier
transform turns the state into an assignment~\cite{jordan24dqi}. Erasing the error register requires recovering
each $y$ from its syndrome $B^{\mathsf T}y$, that is, decoding $\deccode$.
Our guarantees require that every such pattern is decoded,
together with a separation condition needed to compute DQI's expected quality.

\begin{definition}[licensed radius]\label{def:licensed}
A decoder for $\deccode$ is \emph{licensed at radius $\ell$} if it recovers every error
pattern of Hamming weight at most $\ell$ from its syndrome and
\[
  2\ell+1<d.
\]
Its \emph{reach} is $\lambda=\ell/m$.
\end{definition}

Ordinary unique decoding needs only $2\ell<d$. The stronger condition $2\ell+1<d$ also separates syndromes one step further:
a pattern of weight at most $\ell$ never shares its syndrome with a different pattern formed
from another such pattern by adding a nonzero value at one coordinate. Jordan et al.\ need this separation
to compute the expected quality~\cite{jordan24dqi}.

\paragraph{DQI quality and Prange completion.}
For the indicator objective, DQI's asymptotic expected quality at licensed reach $\lambda$ is
\begin{equation}\label{eq:qdqi}
  \Qdqi(\lambda,\rho)=
  \begin{cases}
    \Bigl(\sqrt{\lambda(1-\rho)}
      +\sqrt{\rho(1-\lambda)}\Bigr)^2,
      & \lambda\le1-\rho,\\
    1, & \lambda>1-\rho.
  \end{cases}
\end{equation}
This is the semicircle law of Jordan et al.~\cite{jordan24dqi}, proved there over prime fields.
We assume it over extension fields, as Gu and Jordan do over
$\F_{q^2}$~\cite{gujordan25}. At finite length, we use the exact expected quality
$\Qdqi^{(m)}$ from Lemma~9.2 of~\cite{jordan24dqi}.

Our first classical baseline is one Prange completion~\cite{prange62}, as described in the introduction.
It fixes an information set of $\optcode$, which is a set of $\rank(B)=\alpha m$ coordinates whose
values determine a unique codeword. It sets
these coordinates to accepted values and solves for the codeword.
If the information set is chosen independently of the accepted sets, and the values assigned to it
use only the accepted sets on those coordinates, each remaining row is satisfied
with probability $\rho$. One completion therefore has exact expected quality
\begin{equation}\label{eq:qpr}
  \Qpr(\alpha,\rho)
  =\frac{\alpha m+\rho(1-\alpha)m}{m}
  =\rho+(1-\rho)\alpha.
\end{equation}
Supplementary Section~\ref{sec:isbaseline} specifies the procedure and its assumptions.

\begin{definition}[margin]\label{def:margin}
The \emph{margin} over one Prange completion is
\begin{equation}\label{eq:margin}
  \mu=\Qdqi(\lambda,\rho)-\Qpr(\alpha,\rho).
\end{equation}
At finite length, the corresponding margin is
$\mu_m=\Qdqi^{(m)}-\Qpr(\alpha,\rho)$.
\end{definition}

Thus $\mu>0$ describes an asymptotic advantage over this particular baseline.
The sign of the finite-length margin $\mu_m$ need not agree with that of $\mu$.

\paragraph{The existential line.}
We also use a counting bound on the quality of codewords. The optimizer code contains
$q^{\alpha m}$ codewords, and any fixed codeword satisfies independently sampled rows
with probability $\rho$. Combining this count with a binomial tail bound gives the
following threshold.

\begin{definition}[existential line]\label{def:regime}
If $q^\alpha\rho<1$, the \emph{existential line} $\Qstar\in(\rho,1)$ is the unique solution of
\begin{equation}\label{eq:qstar}
  \KL{\Qstar}{\rho}=\alpha\ln q,
\end{equation}
where
\[
  \KL{Q}{\rho}
  =Q\ln\frac{Q}{\rho}
   +(1-Q)\ln\frac{1-Q}{1-\rho}.
\]
If $q^\alpha\rho\ge1$, set $\Qstar=1$.
\end{definition}

For every fixed $\varepsilon>0$, with high probability no assignment exceeds
$\Qstar+\varepsilon$ (supplementary Section~\ref{sec:counting}). Since there are $q^{\alpha m}$
codewords, on average $(q^\alpha\rho)^m$ of them satisfy every row. We call $q^\alpha\rho<1$ the
\emph{unsatisfiable regime}, because this average then tends to zero, and $q^\alpha\rho\ge1$ the \emph{satisfiable regime}.
The second name describes only the average, so it does not guarantee that a fully satisfying
assignment exists. Whether an efficient algorithm reaches $\Qstar$ is a separate question,
which we take up with the stronger classical baselines.

%% file: sections/03-rules.tex
\subsection{The reach required for an advantage}\label{sec:rules}\label{sec:soundness}

A useful decoder must do more than correct a positive fraction of errors: its reach must
be large enough for DQI to outperform the classical baseline. The required reach depends
on both the optimizer-code rate and the accepted-set density. Increasing the rate
strengthens Prange completion. Changing the density affects both algorithms.

The quality law \eqref{eq:qdqi} is increasing in $\lambda$ until it reaches one.
Writing $\lambda=\sin^2\theta$ and $\rho=\sin^2\phi$, its unsaturated branch becomes
$\sin^2(\theta+\phi)$. Inverting this expression identifies the point at which DQI's
expected quality crosses the Prange baseline.

\begin{theorem}[the reach threshold]\label{thm:master}
Let $0<\alpha<1$ and $0<\rho<1$, and write
$\Qpr=\rho+(1-\rho)\alpha$. Define
\begin{equation}\label{eq:lambdacrit}
  \lambdacrit(\alpha,\rho)
  =\Bigl(
    \sqrt{(1-\rho)\Qpr}
    -\sqrt{\rho(1-\Qpr)}
  \Bigr)^2.
\end{equation}
Then $0<\lambdacrit<1-\rho$, and for every $\lambda\in[0,1]$,
\[
  \Qdqi(\lambda,\rho)>\Qpr(\alpha,\rho)
  \quad\Longleftrightarrow\quad
  \lambda>\lambdacrit(\alpha,\rho).
\]
At $\rho=\tfrac12$,
\[
  \lambdacrit(\alpha,\tfrac12)
  =\frac{1-\sqrt{1-\alpha^2}}{2}.
\]
Consequently, a family whose limiting reach is at most $U$ has $\mu\le0$ whenever
$U\le\lambdacrit(\alpha,\rho)$.
\end{theorem}

The threshold separates the optimization comparison from the coding question.
Once $\alpha$ and $\rho$ are fixed, an advantage requires a decoder reaching beyond
$\lambdacrit$. An upper bound on the possible reach can therefore exclude a candidate
before a decoder is constructed. Rule \gate{1} collects the conditions obtained by
applying standard coding bounds in this way.

\begin{corollary}[rule \gate{1}]\label{cor:rules}
Consider a family with limiting parameters $(\alpha,\delta,\rho)$, where
$0<\alpha<1$ and $0<\rho<1$, the indicator objective, and a decoder licensed at
limiting reach $\lambda$. If $\mu>0$, then $\lambdacrit(\alpha,\rho)$ lies strictly
below each applicable upper bound on the limiting reach:
\begin{enumerate}
  \item[(a)] \emph{Distance.}
    The licensing condition gives $\lambda\le\delta/2$, so
    \[
      \delta>2\lambdacrit(\alpha,\rho).
    \]
    At $\rho=\tfrac12$, this becomes
    $\delta>1-\sqrt{1-\alpha^2}$.
    Any upper bound on $\delta$ can be used to exclude an advantage.
    In particular, if the relative distance tends to zero, then
    $\mu\to-(1-\rho)\alpha$.
    Conversely, if $\delta>2\lambdacrit$, a decoder licensed sufficiently close
    to $\delta/2$ would give $\mu>0$, provided such a decoder exists.

  \item[(b)] \emph{Sphere packing.}
    Along the family, the licensed reach satisfies
    \[
      \lambda\le H_q^{-1}(\alpha)+o(1),
    \]
    with the quantities evaluated at the corresponding family parameters.
    Here $H_q^{-1}$ is the inverse of the $q$-ary entropy on $[0,1-1/q]$.
    An advantage requires the resulting asymptotic upper bound to exceed
    $\lambdacrit$.

  \item[(c)] \emph{Singleton and Plotkin.}
    The Singleton bound gives $\lambda<\alpha/2$ at finite length.
    The necessary asymptotic condition
    $\alpha/2>\lambdacrit(\alpha,\rho)$ is equivalent to
    $\rho>\rhocrit(\alpha)$, where
    \begin{equation}\label{eq:rhocrit}
      \rhocrit(\alpha)
      =\frac12\left(
        1-\sqrt{\frac{2(1-\alpha)}{2-\alpha}}
      \right).
    \end{equation}
    At fixed $q$, the asymptotic Plotkin bound strengthens the reach bound to
    $(1-1/q)\alpha/2$, which must likewise exceed $\lambdacrit$.
\end{enumerate}
\end{corollary}

These bounds concern the reach a decoder could achieve. Rule \gate{2} asks whether a
polynomial-time decoder is known to satisfy the required license, which our efficient DQI guarantees need in order to apply. Without such a decoder,
this analysis leaves the candidate unresolved and gives no upper bound on the
quality DQI could attain. Decoders that succeed only on average require a different
analysis (Section~\ref{sec:discussion}).

The coding bounds constrain the size of an advantage as well as its existence.
Even a decoder approaching the Singleton reach has a limited margin over Prange
completion. Optimizing over the accepted-set density gives a cap that depends only on
the optimizer-code rate. This cap gives a common scale for comparing candidate families.

\subsection{A sharp margin cap}

\begin{proposition}[margin cap]\label{prop:cap}
Let the objective be the indicator and $0<\alpha<1$. At every finite length,
the licensed reach satisfies $\lambda<\alpha/2$, and hence
\[
  \Qdqi(\lambda,\rho)-\Qpr(\alpha,\rho)<\kappa(\alpha).
\]
In the limit,
\begin{equation}\label{eq:kappa}
  \mu\le\kappa(\alpha)
  =\sqrt{u(1-u)}-u,
  \qquad u=\alpha/2.
\end{equation}
In particular,
\[
  \mu\le\frac{\sqrt2-1}{2}.
\]
Attaining this universal cap in the limit requires
$\rho\to\tfrac12$, $\alpha\to1-1/\sqrt2$, and decoder-code relative distance
tending to $\alpha$, which in turn requires $q\to\infty$.
\end{proposition}

The quantity $\kappa(\alpha)$ is the asymptotic margin at density $\tfrac12$ and
Singleton reach $\alpha/2$.\label{sec:cap} The finite-parameter inequality above
uses the asymptotic quality function $\Qdqi$. Actual finite-length comparisons
use $\Qdqi^{(m)}$ instead. Supplementary Sections~\ref{sec:ruleproofs} and~\ref{sec:capproof}
give the proofs, and supplementary Section~\ref{sec:rulesreading} discusses them further.

\subsection{Stronger classical baselines}\label{sec:rulesbaseline}

The reach test compares DQI with an algorithm that fixes an information set and accepts
the resulting values elsewhere. A candidate may satisfy that condition yet admit a stronger
classical method. We consider two mechanisms: an encoder that uses the full collection
of accepted sets to find a high-quality codeword, and a linear-algebra method that
exploits affine subspaces within those sets.

Here an \emph{encoder} is understood in the sense of lossy source coding: given the
accepted sets, it returns a codeword of $\optcode$ satisfying many rows. If such an
encoder attains the existential line, the counting bound leaves no positive limiting
margin for DQI. For the second mechanism, expressing accepted affine subspaces by
equations over a subfield can allow a classical solver to enforce more rows than
ordinary Prange completion.
This mechanism is the deterministic case of the extended Prange algorithm of
Khattar et al.~\cite{khattar25verifiable}, who spend a budget of $bn$ base-field
constraints across codimensions chosen by linear programming and allow each
constraint to be satisfied with probability below one. Rule \gate{4} uses only the
single codimension realized by a coset contained in an accepted set, which gives a
weaker baseline but a closed-form condition on the invariants.

Against the applicable baselines, define
\begin{equation}\label{eq:muplus}
  \mu^+
  =\Qdqi(\lambda,\rho)
   -\max\Bigl\{
      \Qpr(\alpha,\rho),\
      \Qstar\ \text{if efficiently attainable},\
      \Qpr(\asub,\rho)\ \text{under \gate{4}}
    \Bigr\}.
\end{equation}
The subfield term is the limiting quality, and the following proposition gives
its finite-length correction.

\begin{proposition}[rules \gate{3} and \gate{4}]\label{prop:remaining}
\leavevmode
\begin{itemize}
  \item[\gate{3}]
    Suppose the accepted sets are independent, uniformly random, and independent of $B$.
    If $\optcode$ has a polynomial-time encoder attaining the existential line
    $\Qstar$ of Definition~\ref{def:regime}, then $\mu^+\le0$.
    An advantage over efficient encoders therefore requires that no such encoder exist.

  \item[\gate{4}]
    Let $q=p^b$. For each accepted set, consider the smallest codimension of an
    $\F_p$-affine subspace it contains. Let $c>0$ be the largest of these
    codimensions and put
    \[
      \asub=\min(1,b\alpha/c).
    \]
    Suppose:
    \begin{enumerate}
      \item[(a)] the accepted sets are independent of $B$ and of one another,
        and each has a translation-invariant law;
      \item[(b)] the rows selected for linear algebra are chosen using only
        translation-invariant properties of their own accepted sets;
      \item[(c)] affine subspaces realizing the stated codimensions can be found
        in polynomial time, and the resulting equations over $\F_p$ have full rank.
    \end{enumerate}
    Condition on the rowwise translation-invariant data determining $c$
    and the row selection. Subfield linear algebra then attains conditional
    expected quality
    \[
      \rho+(1-\rho)\frac{\lfloor\asub m\rfloor}{m}.
    \]
    Without conditioning, take the expectation of this right-hand side as well.
    The following comparisons are conditional on these data, and a limiting
    comparison presumes convergence of the resulting subfield rate.
    Consequently,
    \[
      \mu^+
      \le\Qdqi(\lambda,\rho)-\Qpr(\asub,\rho)+O(1/m),
    \]
    and a positive limiting margin requires
    $\Qdqi(\lambda,\rho)>\Qpr(\asub,\rho)$.
    At $\rho=\tfrac12$, this condition fails for every licensed Hamming reach when
    \[
      \alpha\ge\frac{2c^2}{b^2+c^2},
    \]
    up to the finite-length $O(1/m)$ correction.
\end{itemize}
\end{proposition}

Proofs of rules \gate{3} and \gate{4} appear in supplementary
Sections~\ref{sec:rd} and~\ref{sec:subfield}.

Together, these results give the four screening rules summarized in
Table~\ref{tab:rules}. They address three questions: whether the code permits enough
decoding reach, whether an efficient decoder achieves the required guarantee, and
whether a stronger classical method removes the resulting advantage.

We call an exclusion rule
\emph{sound} when an advantage implies its condition under the stated hypotheses,
so violating \gate{1}, \gate{3} or \gate{4} excludes an advantage over
the corresponding baseline.
The rules do not establish an advantage over all classical algorithms, and
without a known encoder attaining $\Qstar$, rule \gate{3} stays unresolved.

\begin{table}[htbp]
\centering
\small
\begin{tabular}{l L{0.47\textwidth} L{0.2\textwidth} l}
\toprule
rule & condition & comparison & statement\\
\midrule
\gate{1} reach
  & Every applicable upper bound on the licensed reach must exceed $\lambdacrit$:
    distance, sphere packing, Singleton, and Plotkin.
  & one Prange completion
  & Cor.~\ref{cor:rules}\\
\gate{2} decoder
  & An efficient decoder must be known to satisfy the required license.
  & applicability
  & \S\ref{sec:rules}\\
\gate{3} encoder
  & No polynomial-time encoder for the optimizer code may attain $\Qstar$.
  & an encoder attaining $\Qstar$
  & Prop.~\ref{prop:remaining}\\
\gate{4} subfield
  & DQI must outperform the applicable subfield linear-algebra baseline.
  & linear algebra over $\F_p$
  & Prop.~\ref{prop:remaining}\\
\bottomrule
\end{tabular}
\caption{The soundness rules for the indicator objective.
Rules \gate{1}, \gate{3}, and \gate{4} give necessary conditions for an advantage
under their respective hypotheses. Rule \gate{2} is an applicability condition.
Figure~\ref{fig:rulesmap} identifies the instance properties used by each rule.}
\label{tab:rules}
\end{table}

\input{figures/rules-map}

\subsection{Evaluating the rules}\label{sec:checkcost}

Evaluating the rules numerically uses
a small collection of code and objective parameters, although establishing the required
bounds and algorithmic guarantees can itself require substantial work.

For the indicator objective, parts (b) and (c) of rule \gate{1} use only $q$, $\alpha$,
and $\rho$. The entropy inverse in part (b) can be evaluated by bisection, and the
rate follows from the rank of $B$, computed by Gaussian elimination in
$O(mn^2)$ field operations.

Part (a) can exclude a candidate using any sufficiently small upper bound on the
decoder-code distance. Such a bound may come from the construction, a general
coding bound, or a low-weight codeword found through a short-dependency search,
as in DQI-Kit~\cite{thelen26dqikit}. Computing the exact distance is
NP-hard~\cite{vardy97}, but exclusion does not require it. Conversely, a sufficiently
large certified lower bound, such as a designed distance, can establish that the
distance condition is satisfied. A decoder with the required licensed reach must
still be supplied separately.

The remaining checks use algorithmic information: a decoder guarantee for \gate{2},
an encoder attaining $\Qstar$ that would trigger \gate{3}, or efficiently available
affine subspaces satisfying the hypotheses of \gate{4}. A missing certificate leaves
the corresponding question unresolved. Once the required bounds and certificates
are available, the remaining calculations are scalar inequalities.

Evaluating these inequalities from the supplied invariants took under two seconds
across the $109$ families in our search
(supplementary Section~\ref{sec:search}), whereas constructing decoders and testing classical
solvers took hours per candidate (Section~\ref{sec:problems}).

\subsection{A finite-length comparison with Prange restarts}\label{sec:converse}

Beyond these necessary conditions, a positive finite margin
gives a guarantee against a specified number of Prange completions.
Restarting improves the classical score by keeping the best completion, but
concentration bounds limit that improvement when each restart chooses its information
set independently of the accepted sets.
The relevant quantity is the exact finite margin $\mu_m$, not its asymptotic
limit. We state the result for both the indicator objective and the cosine objective
of Section~\ref{sec:rulesalphabet}, whose objective ranges are $1$ and $2$.

\begin{theorem}[DQI against $N$ Prange restarts]\label{thm:converse}
Suppose DQI has a decoder licensed at radius $\ell$: in the sense of
Definition~\ref{def:licensed} for the indicator objective, or licensed for
$\{\pm1\}$ in the sense of Definition~\ref{def:licensedS} for the cosine.
Let $\Qdqi^{(m)}$ be its exact expected quality at length $m$, let the objective
have range $R$, and suppose
\[
  \mu_m=\Qdqi^{(m)}-\Qpr>0.
\]
For the indicator objective, take $R=1$ and accepted sets drawn independently
and uniformly, independently of $B$. For the cosine, take $R=2$ and centers
drawn independently and uniformly, independently of $B$, as in supplementary
Lemma~\ref{lem:centered}.

Let $N\ge1$ and $\eta\in(0,1)$. Each restart must follow supplementary
Section~\ref{sec:isbaseline}, with its information set and tie-breaking chosen
independently of the accepted sets or centers. If
\[
  m\ge
  \frac{(1-\alpha)R^2\ln(N/\eta)}{2\mu_m^2},
\]
then, with probability at least $1-\eta$, the best of the $N$ restarts has total
score strictly below $m\Qdqi^{(m)}$.
\end{theorem}

The proof combines Hoeffding's inequality with a union bound over the restarts
(supplementary Sections~\ref{sec:converseproof} and~\ref{sec:converserange}).
The guarantee concerns the specified completion procedure. Algorithms that use
accepted sets outside the chosen information set to improve their assignments,
including the solvers in Section~\ref{sec:comparators}, require separate comparisons.
We quote thresholds at $N=2^{40}$ and $\eta=0.05$.

The theorem compares classical scores with DQI's expected quality. To relate that
expectation to individual samples, let $Z\le1$ be the quality of a DQI output with
$\mathbb EZ\ge Q$. For every $\tau<Q$,
\begin{equation}\label{eq:reversemarkov}
  \Pr[Z\ge\tau]\ge\frac{Q-\tau}{1-\tau}.
\end{equation}
Under the worst-case decoding license, this applies on every instance with
$Q=\Qdqi^{(m)}$. A target separated from that expectation by a fixed positive
gap is therefore attained with constant probability.

\section{Objective-dependent decoding guarantees}\label{sec:rulesalphabet}

The preceding bounds concern decoders licensed for all errors within a Hamming
radius. DQI's decoding requirements can be narrower, because the objective determines
which error values occur in the quantum state.
Lemma~\ref{lem:objlaw} gives the resulting quality law for general objectives, and
Proposition~\ref{prop:alphabet} specializes it to the cosine objective used by
our MPI candidate.

\subsection{Quality and decoding requirements}

Consider a real-valued objective of the form
\[
  \sum_{i=1}^{m}F\bigl((Bx)_i-c_i\bigr),
\]
where each row uses the same function $F:\F_q\to\mathbb R$ up to a shift $c_i$.
The nonzero entries of DQI's error register lie in the nonzero Fourier support
of $F$~\cite[\S8.2]{jordan24dqi}. Thus the decoder needs to recover only patterns
using those error values, together with the separation required by the quality
calculation.

\begin{definition}[licensed for a set]\label{def:licensedS}
For $S\subseteq\F_q^*$, let
\[
  \mathcal E_{S,\ell}
  =\{y\in(S\cup\{0\})^m:\operatorname{wt}(y)\le\ell\}.
\]
A decoder for $\deccode$ is \emph{licensed for $S$ at radius $\ell$} if it
recovers every $y\in\mathcal E_{S,\ell}$ from its syndrome and satisfies
\[
  B^{\mathsf T}y'\neq B^{\mathsf T}(y+ae_i)
\]
for every $y,y'\in\mathcal E_{S,\ell}$, $a\in S$, and coordinate $i$ whenever
$y'\neq y+ae_i$. Here $e_i$ is the $i$-th unit vector.
\end{definition}

The second condition is the separation hypothesis of supplementary
Lemma~\ref{lem:cosinelaw}. For the full error alphabet, the distance condition
in Definition~\ref{def:licensed} supplies this separation. For a restricted
alphabet, it can hold at larger radii because fewer patterns need to be
distinguished.

As Lemma~\ref{lem:objlaw} below shows, the objective's mean, variance and skewness set the quality obtained at a given
reach, and its Fourier support sets the decoding problem that must be solved
to attain that reach.

\begin{lemma}[the general-objective law]\label{lem:objlaw}
Let $q$ be prime and let the per-row objectives be
$f_i(v)=F(v-c_i)$, where $F:\F_q\to\mathbb R$ has mean $\bar F$, standard
deviation $\sigma>0$, skewness $\gamma$, and nonzero Fourier support $S$.
Suppose the decoder is licensed for $S$ at radius $\ell$
(Definition~\ref{def:licensedS}).

With its shell weights chosen optimally~\cite{jordan24dqi}, DQI has exact
expected objective
\[
  m\bar F+\sigma\lambda_{\max}(A),
\]
where $A$ is the tridiagonal matrix indexed by $\{0,\dots,\ell\}$ with
\[
  A_{kk}=k\gamma,
  \qquad
  A_{k-1,k}=A_{k,k-1}=\sqrt{k(m-k+1)}
  \quad (1\le k\le\ell).
\]
If $\ell/m\to\lambda\in(0,\tfrac12]$, and the moments converge when $q$ or
$F$ varies along the family, the expected objective per row tends to
\begin{equation}\label{eq:objlaw}
  \bar F+\sigma
  \max_{0\le t\le\lambda}
  \left(\gamma t+2\sqrt{t(1-t)}\right),
\end{equation}
with the limiting moments used on the right-hand side.

For centers drawn independently and uniformly, independently of $B$, one
Prange completion has expected objective per row
\[
  \alpha\max F+(1-\alpha)\bar F.
\]
\end{lemma}

The classical baseline has the same interpretation as before: the solver sets
the information-set coordinates to maximize their row objectives, and the
remaining coordinates contribute the mean $\bar F$ in expectation.

The proof and a direct finite-instance verification of the identity appear in
supplementary Sections~\ref{sec:objlawproofs} and~\ref{sec:checks}.

\subsection{The cosine objective}

The cosine gives a simple example in which the Fourier support narrows the decoding problem. Over an odd prime field,
\[
  F(v)=\cos(2\pi v/q)
\]
has Fourier support $\{\pm1\}$. DQI therefore uses only error patterns with
entries in $\{0,\pm1\}$, rather than all patterns in a Hamming ball.
A decoder licensed for these signed patterns can reach beyond half the
Hamming distance. Figure~\ref{fig:pmone} illustrates why restricting the error alphabet can
increase the correctable radius. The Hamming balls overlap, but the sets
obtained by allowing only single-coordinate $\pm1$ errors are disjoint.
Thus failure of Hamming unique decoding need not prevent recovery of
the errors used by the cosine objective.

\input{figures/pm-one-balls}

For $q>3$, the cosine has mean zero, variance $\tfrac12$, and skewness zero.
The general quality law consequently reduces to an explicit function of the
reach. The classical Prange baseline is $\alpha$, since the cosine has maximum
one and mean zero. These expressions give a new reach threshold, and counting
signed patterns gives the corresponding packing restriction.

\begin{proposition}[rules \gate{1} and \gate{2} for the cosine]\label{prop:alphabet}
Let $q>3$ be prime and let $F(v)=\cos(2\pi v/q)$, with centers drawn
independently and uniformly, independently of $B$. Suppose the decoder is
licensed for $\{\pm1\}$ at reach $\lambda\le\tfrac12$.
\begin{enumerate}
  \item[(i)] \emph{Packing bound.}
    Sphere packing gives
    \[
      \alpha
      \ge\frac{H_2(\lambda)+\lambda}{\log_2q}-o(1),
    \]
    where $H_2$ is binary entropy in bits.
    An advantage requires this bound to permit a reach above the threshold
    in (ii).

  \item[(ii)] \emph{Quality and threshold.}
    DQI's asymptotic objective per row is
    $\sqrt{2\lambda(1-\lambda)}$, and the Prange baseline is $\alpha$.
    For $0<\alpha<1/\sqrt2$, DQI beats this baseline exactly when
    \[
      \lambda>\lambdacrit^{\cos}(\alpha)
      =\frac{1-\sqrt{1-2\alpha^2}}{2}.
    \]
    No strict advantage is possible when $\alpha\ge1/\sqrt2$.

    Rule \gate{1}(a) therefore compares upper bounds on the
    $\{\pm1\}$-licensed reach with $\lambdacrit^{\cos}$, and rule \gate{2}
    asks for an efficient decoder with that license. The Hamming-based
    Singleton restriction in \gate{1}(c) does not apply.
    Theorem~\ref{thm:converse} applies with objective range $R=2$.

  \item[(iii)] \emph{Existential line.}
    The existential line is the largest mean of $\cos(2\pi v/q)$ over
    distributions on $\F_q$ whose relative entropy from uniform is at most
    $\alpha\ln q$. For $\alpha<1$, it is attained by a finite
    exponentially tilted distribution and is strictly below one.
    At $\alpha=1$, it equals one and is attained by the point mass
    at zero, obtained as the infinite-tilt limit.
\end{enumerate}
\end{proposition}

The screening principle remains the same: the available reach must exceed
the reach required to beat the baseline. What changes is the decoding problem
to which that comparison applies. For the cosine, a Hamming-distance bound
alone need not bound the licensed reach on signed patterns, so a family
excluded by the full-alphabet analysis may remain viable for this objective.

Evaluating the modified rules requires the objective's Fourier support,
obtained from a transform of length $q$, together with bounds and a decoder
guarantee for the resulting error alphabet. Once these are supplied, the
quality comparison again reduces to scalar inequalities. MPI
(Section~\ref{sec:logcauchy}) uses this restricted-alphabet route.

%% file: figures/rules-map.tex
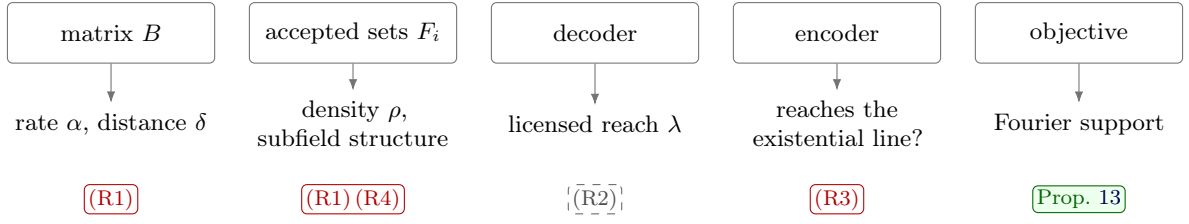
\begin{figure}[t]
\centering
\begin{tikzpicture}[
  x=1mm, y=1mm,
  part/.style={draw, draw=black!55, rounded corners=2pt, align=center, font=\footnotesize,
               text width=26mm, minimum height=8mm, inner sep=2pt},
  fixes/.style={align=center, font=\footnotesize, text width=28mm},
  rule/.style={draw=red!70!black, text=red!70!black, rounded corners=2pt, font=\scriptsize,
               inner sep=1.5pt},
  check/.style={draw=black!60, dashed, text=black!70, rounded corners=2pt, font=\scriptsize,
               inner sep=1.5pt},
  admit/.style={draw=green!40!black, text=green!40!black, fill=green!7, rounded corners=2pt,
                font=\scriptsize, inner sep=1.5pt},
  arr/.style={-{Latex[length=1.5mm]}, black!55},
]
  \foreach \x/\name/\what in {
      0/{matrix $B$}/{rate $\alpha$, distance $\delta$},
      32/{accepted sets $F_i$}/{density $\rho$,\\ subfield structure},
      64/{decoder}/{licensed reach $\lambda$},
      96/{encoder}/{reaches the existential line?},
      128/{objective}/{Fourier support}} {
    \node[part] (p\x) at (\x,0) {\name};
    \node[fixes] (f\x) at (\x,-12) {\what};
    \draw[arr] (p\x) -- (f\x);
  }
  \node[rule] at (0,-22)  {\gate{1}};
  \node[rule] at (32,-22) {\gate{1}\,\gate{4}};
  \node[check] at (64,-22) {\gate{2}};
  \node[rule] at (96,-22) {\gate{3}};
  \node[admit] at (128,-22) {Prop.~\ref{prop:alphabet}};
\end{tikzpicture}
\caption{What each rule depends on. Rule \gate{1} bounds the reach through the threshold of
Theorem~\ref{thm:master}, \gate{2} (dashed) asks that a licensed decoder be known, so that the
theorems apply, and \gate{3} and \gate{4} raise the classical baseline. The objective's Fourier
support sets the errors over which the reach is measured (Proposition~\ref{prop:alphabet}). The
threshold itself depends on the density, so \gate{1} also reads the accepted sets.}
\label{fig:rulesmap}
\end{figure}

%% file: figures/pm-one-balls.tex
\begin{figure}[h]
\centering
\begin{tikzpicture}[
  x=4.5mm, y=4.5mm,
  pt/.style={circle, fill=black!22, inner sep=0pt, minimum size=1.2mm},
  ham/.style={circle, fill=red!70!black, inner sep=0pt, minimum size=1.8mm},
  pm/.style={circle, fill=green!40!black, inner sep=0pt, minimum size=1.8mm},
  cw/.style={circle, fill=black, inner sep=0pt, minimum size=2.6mm},
  lab/.style={font=\footnotesize},
  head/.style={font=\footnotesize\itshape, align=center, text width=40mm, anchor=south},
]
  \begin{scope}
    \fill[red!7, rounded corners=2pt] (-0.45,1.55) rectangle (6.45,2.45);
    \fill[red!7, rounded corners=2pt] (2.55,-0.45) rectangle (3.45,6.45); \fill[red!7, rounded corners=2pt] (-0.45,4.55) rectangle (6.45,5.45); \fill[red!7, rounded corners=2pt] (3.55,-0.45) rectangle (4.45,6.45);
    \foreach \x in {0,...,6} \foreach \y in {0,...,6} \node[pt] at (\x,\y) {};
    \node[ham] at (4,2) {}; \node[ham] at (3,5) {}; \draw[red!70!black, thick] (4,2) circle (0.45); \draw[red!70!black, thick] (3,5) circle (0.45);
    \node[cw] at (3,2) {};
    \node[cw] at (4,5) {}; \node[lab] at (2.45,1.45) {$c$}; \node[lab] at (4.6,5.55) {$c'$};
    \node[head] at (3,6.8) {Hamming balls of radius 1 around $c$ and $c'$ meet};
  \end{scope}
  \begin{scope}[xshift=52mm]
    \foreach \cx/\cy in {3/2,4/5} \fill[green!8, rounded corners=2pt] (\cx,\cy-1.4) -- (\cx+1.4,\cy) -- (\cx,\cy+1.4) -- (\cx-1.4,\cy) -- cycle;
    \foreach \x in {0,...,6} \foreach \y in {0,...,6} \node[pt] at (\x,\y) {};
    \foreach \cx/\cy in {3/2,4/5} { \node[pm] at (\cx-1,\cy) {}; \node[pm] at (\cx+1,\cy) {}; \node[pm] at (\cx,\cy-1) {}; \node[pm] at (\cx,\cy+1) {}; \node[cw] at (\cx,\cy) {};
    }
    \node[lab] at (2.45,1.45) {$c$}; \node[lab] at (4.6,5.55) {$c'$};
    \node[head] at (3,6.8) {$\pm1$ balls of radius 1 around $c$ and $c'$};
  \end{scope}
\end{tikzpicture}
\caption{Two codewords $c=(3,2)$ and $c'=c+(1,3)$ of the code spanned by $(1,3)$ in $\F_7^2$, drawn on a $7\times7$ grid whose edges
wrap around. The code has Hamming distance $2$. Left: the Hamming balls of radius 1 around $c$ and $c'$ share the points $(4,2)$ and $(3,5)$, so no decoder corrects
one error of arbitrary value. Right: an error with a single entry $\pm1$ moves $c$ and $c'$ to
disjoint sets, and the same holds for every pair of codewords, so every such error is corrected.
This illustrates unique recovery; the additional separation condition in
Definition~\ref{def:licensedS} must be checked separately.}
\label{fig:pmone}
\end{figure}
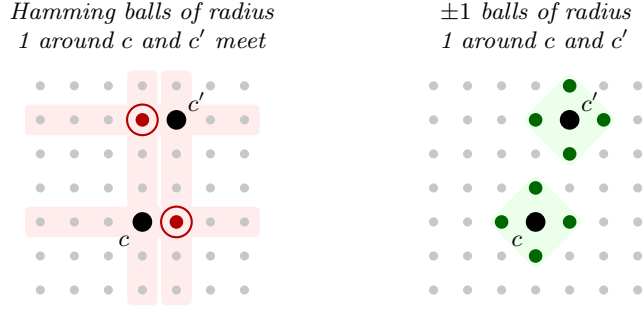

%% file: sections/04-problems.tex
\section{Two new candidate problems}\label{sec:problems}

Of the proposals in the rule-guided search, two candidates were retained for detailed study.
Many others reproduced
OPI through a change of notation or encoding, and the rest were excluded
by coding bounds, lacked a decoding guarantee that could certify
an advantage, or were overtaken by classical algorithms. Supplementary
Section~\ref{sec:supp-search} records all proposals and the outcomes
of the checks.

Because OPI already
attains the margin cap for the indicator objective with a
Hamming-licensed decoder (Proposition~\ref{prop:cap}), a new candidate
must change the computational setting or the decoding requirement, not just the code.
AMA changes access to the decoder. A McEliece-type
key gives DQI an efficient Hamming decoder, and with the key
the problem has a trace-structured OPI representation. We compare DQI
with classical solvers given only the public instance
and also report searches given the key.
MPI keeps the decoder public and changes the
objective. The cosine's Fourier support restricts the required errors to
$\{\pm1\}$, so the analysis of
Section~\ref{sec:rulesalphabet} applies without a matching
Hamming-distance guarantee.
Corollary~\ref{cor:mckey} and Theorem~\ref{thm:logcauchy} give
the decoding and DQI guarantees, and
supplementary Section~\ref{sec:supp-candidates} gives further construction details, proofs and attacks.

\paragraph{What the comparisons establish.}
Theorem~\ref{thm:converse} compares DQI's exact expected quality with a
specified number of Prange completions whose information sets are
chosen independently of the accepted sets or centers. The experiments
compare that expectation with the scores reached by named classical
algorithms at fixed budgets. Assumptions (A), (B), and (C) concern
broader classical hardness at stated parameters and computational
budgets. The experimental comparisons stand independently of those
assumptions.

\phantomsection\label{sec:comparators}%
Our generic classical \emph{portfolio} uses $B$ and the accepted sets
or objective centers. It consists of Prange restarts, line search
along the columns of a systematic basis of $\optcode$, simulated
annealing~\cite{kirkpatrick83} with the same moves, and, where stated,
CP-SAT~\cite{ortools} warm-started from the best assignment found.
We supplement this portfolio with attacks that exploit the structure
of each candidate: polynomial searches using the secret key for AMA
and lattice reduction for MPI.

The portfolio runs are limited by wall-clock time. Supplementary
Section~\ref{sec:supp-scaling} repeats line search and simulated
annealing with budgets counted in moves, up to $10^7$ moves per chain.
To compare different lengths, we express each score's improvement
over $\Qpr$ in units of the standard deviation of one Prange completion.
The tested searches exceed the Gaussian reference for the best of
$N$ independent completions by at most about $1.4$ of these units,
where $N$ is the number of codewords scored. This reference
grows like $\sqrt{2\ln N}$, whereas DQI's improvement grows like
$\sqrt m$ on the same scale. Below, we use this observed relationship
to estimate how the classical effort depends on the instance length.

\phantomsection\label{sec:opi}%
OPI provides a reference benchmark~\cite{jordan24dqi}. Its decoder code
is MDS, so at density $\tfrac12$ it attains the margin cap
$\kappa(\alpha)$ of Proposition~\ref{prop:cap}. Using $\alpha=0.1$,
as in Jordan et al., the fixed-time portfolio beats DQI's exact finite
value at the tested lengths up to $m=210$ and falls below it
from $m=400$ onward (Table~\ref{tab:experiments} and supplementary Section~\ref{sec:solvers}).
With $10^6$ moves per chain the same searches match or exceed DQI's
value at $m=400$, but at the longer lengths tested, up to
$m=4000$, they remain far below it
(supplementary Section~\ref{sec:scalingconsequences}).
An asymptotic margin alone therefore does not determine the outcome at a
particular length and budget, which is why the exact finite law matters. Soft-decoding methods can further improve DQI's
performance on OPI~\cite{chailloux24rd,opisoft25}.

\subsection{Alternant max-agreement on a McEliece public key}\label{sec:mckey}

\paragraph{The optimization problem.}
The public input consists of a parity-check matrix
$\Hpub\in\F_p^{n\times m}$ of an alternant code, presented as a
Niederreiter- or McEliece-type public key~\cite{niederreiter86,mceliece78}, where $p$ is
prime, and nonempty proper accepted sets $F_i\subset\F_p$ for
$i=1,\ldots,m$. The \emph{alternant max-agreement problem} asks
for an assignment $x\in\F_p^n$ maximizing
\[
  Q(x)=\frac1m\sum_{i=1}^{m}
  \mathbf{1}\!\left\{(\Hpub^{\mathsf T}x)_i\in F_i\right\}.
\]
Equivalently, it asks for a vector $c=\Hpub^{\mathsf T}x$ in the
code spanned by the rows of $\Hpub$ that belongs to the prescribed
accepted set on as many coordinates as possible. Thus the feasible
vectors lie in the dual of the cryptosystem's code $\ker\Hpub$.
We use accepted sets of a common size, approximately half the
field in the reported instances. The matrix and accepted sets
fully specify both the feasible vectors and their scores.
The secret key supplies auxiliary decoding information to DQI.
The access assumptions for the classical comparisons are specified below.

\paragraph{The decoding structure.}
Set
\[
  B=\Hpub^{\mathsf T}.
\]
Since
\[
  \deccode=\ker(\Hpub),
\]
the secret decoding algorithm of the cryptosystem is precisely a
decoder for the code DQI needs.

Here $\Hpub$ is a scrambled parity-check matrix of an
\emph{alternant code}, obtained by restricting a generalized
Reed--Solomon code over $\F_{p^s}$ to codewords with entries in
$\F_p$. The underlying Reed--Solomon description has $r$ parity
rows. Knowing that description permits correction of up to
$\lfloor r/2\rfloor$ Hamming errors. For DQI we choose the radius
to satisfy the stronger license of Definition~\ref{def:licensed}.
Classic McEliece~\cite{classicmceliece22} uses binary Goppa codes,
a special case of the alternant construction. Our instances instead
use a prime alphabet and an accepted-set objective.

The standard alternant distance and decoding guarantees give the
following application of the general DQI results.

\begin{corollary}[DQI for alternant max-agreement]\label{cor:mckey}
Let $p$ be prime, and let $\Hpub$ be a public parity-check matrix
of an alternant code over $\F_p$ obtained from a generalized
Reed--Solomon code over $\F_{p^s}$ with $r\ge1$ parity rows.
Suppose $\rank(\Hpub)=sr<m$, and let
$B=\Hpub^{\mathsf T}$. Given the decoding key, set
\[
  \alpha=\frac{sr}{m},
  \qquad
  \ell=\left\lfloor\frac{r-1}{2}\right\rfloor,
  \qquad
  \lambda=\frac{\ell}{m}.
\]
Then the keyed decoder is licensed at radius $\ell$ in the sense
of Definition~\ref{def:licensed}. For accepted sets of common
density $\rho\in(0,1)$, DQI attains the exact finite expected
quality $\Qdqi^{(m)}$ of the indicator quality law.

For independently and uniformly drawn accepted sets, independently
of $B$, whenever
\[
  \mu_m
  =\Qdqi^{(m)}-\bigl(\rho+(1-\rho)\alpha\bigr)>0,
\]
the comparison against Prange restarts in
Theorem~\ref{thm:converse} applies.

Along any admissible family with fixed extension degree $s$ and
limiting parameters $0<\alpha<1$ and $0<\rho<1$, the licensed reach tends
to $\alpha/(2s)$. Its limiting margin over one Prange completion
is positive exactly when
\[
  \frac{\alpha}{2s}>\lambdacrit(\alpha,\rho).
\]
\end{corollary}

\begin{proof}
The alternant code has designed distance at least $r+1$, and
the keyed decoder corrects up to $\lfloor r/2\rfloor$ errors.
The chosen radius satisfies
\[
  2\ell+1<r+1\le d,
\]
establishing the license. The quality and comparison statements
follow from the finite DQI law, Theorem~\ref{thm:master}, and
Theorem~\ref{thm:converse}.
\end{proof}

The asymptotic statement allows the base field to grow, but the
finite-length conclusions apply directly at fixed field parameters.

\paragraph{Parameter choice.}
At density $\tfrac12$, the asymptotic margin supplied by
Corollary~\ref{cor:mckey} is
\[
  \mu=\sqrt{\lambda(1-\lambda)}-\frac{\alpha}{2},
  \qquad
  \lambda=\frac{\alpha}{2s}.
\]
For fixed $s$, this expression is maximized at
\[
  \alpha^*
  =s\left(1-\frac{s}{\sqrt{1+s^2}}\right).
\]
This calculation guides our choice of parameters. We take $p=11$,
$s=6$, $\rho=5/11$, and nominal rate $\alpha=0.08$, close to
$\alpha^*=0.0816$. The main instance length is $m=14000$, with
$r=187$ parity rows and $n=sr=1122$ columns.

\paragraph{The role of the key.}
The key reveals the support points $a_i$ and multipliers $v_i$ of
the alternant representation. With this information, the optimization
problem can be written as
\[
  \max_{\deg f<r}
  \sum_{i=1}^{m}
  \mathbf{1}\{\Tr(v_i f(a_i))\in F_i\},
\]
where $f$ is a polynomial over $\F_{p^s}$ and $\Tr$ is the field
trace onto $\F_p$. The key holder therefore faces OPI with
trace-defined accepted sets.

This representation explains both the connection to OPI and the
question posed by the construction. DQI uses the key to decode
errors for $\deccode$, and the optimization task searches for a
high-quality word in its dual $\optcode$
(Section~\ref{sec:instances}). Efficient decoding of the former
does not immediately supply an efficient optimizer for the latter.

\paragraph{Hardness assumptions without the secret key.}
\phantomsection\label{sec:mchardness}%
We formulate two assumptions at
$(p,s,\alpha)=(11,6,0.08)$ and $m=40000$, with
$\varepsilon=0.01$ and a budget of $2^{80}$ bit operations.
At this length $r=533$, $n=3198$, $\ell=266$ and
$\Qdqi^{(m)}\simeq0.534$. This length is larger than those tested:
the search results below motivate extrapolating to $m=40000$, but
we have not generated an instance of that length.
\begin{itemize}
  \item \emph{Assumption (A).}
    Given $\Hpub$ as a uniformly random basis of its row space,
    rather than in systematic form, no classical algorithm within
    the budget distinguishes it from a uniformly random matrix of
    the same shape with advantage at least $\varepsilon$.

  \item \emph{Assumption (B).}
    For a uniformly random matrix with independently drawn accepted
    sets of size $5$, no classical algorithm within the budget
    achieves expected quality above
    $\Qdqi^{(m)}-\varepsilon\simeq0.524$.
\end{itemize}

Together, these assumptions imply that no classical algorithm without
the secret key reaches DQI's expected quality of approximately $0.534$
within the budget, after accounting for the reduction overhead.
Indeed, running such an optimizer on a challenge matrix and testing
a uniformly chosen row of its output would distinguish the
public-key distribution from the random-matrix distribution: the
acceptance probability of that test is exactly the optimizer's
expected quality. This argument applies to algorithms given only
the public matrix and accepted sets. It does not establish hardness for a solver
supplied with the secret key.

We state these assumptions at a larger length because extrapolating
the tested searches suggests that these searches could reach the target of
Assumption (B) within its budget at the experimental lengths. At
$m=14000$, the best line-search chains reach $0.5189$ to $0.5210$
on the three seeds, with about $10^6$ moves per seed,
against $\Qdqi^{(m)}-\varepsilon\simeq0.522$
(supplementary Table~\ref{tab:scaling}). On the scale of
supplementary Section~\ref{sec:scalinglaw}, $\Qdqi^{(m)}-\varepsilon$
lies $5.9$ standard deviations above $\Qpr$ at $m=14000$. The
Gaussian reference estimates that reaching this level requires about
$2^{29}$ independent completions, and the tested searches perform
close to that benchmark. At
$m=40000$ the same quantity lies $10.8$ standard deviations above
$\Qpr$, about $2^{89}$ independent completions. A search that kept the largest improvement
over the Gaussian reference seen in our runs would still need about $2^{69}$. In both cases the estimated cost exceeds
$2^{80}$ field operations, and hence $2^{80}$ bit operations
(supplementary Section~\ref{sec:scalingconsequences}). These cost estimates
extrapolate the measured behavior of the tested searches and do
not bound other algorithms.

On a uniformly random matrix of the key's shape at $m=20000$, our portfolio reaches $0.5111$. This lies within the portfolio's range on the three keys, $0.5092$ to $0.5128$, consistent with Assumption (A).

Known structural attacks inform the parameter choice.
Ghoshal, Ishai, Jain, and Sun give an unconditional quasipolynomial
distinguisher for codes contained in generalized Reed--Solomon
codes, including alternant codes~\cite{ghoshal26quasipoly}.
Porting their cost model to $\F_{11}$ gives an estimate of about
$2^{123}$ binary operations for the theorem-backed variant and about
$2^{106}$ for the heuristic variant at $m=14000$, and about $2^{121}$
and $2^{117}$ at $m=40000$. The additional
binary-field optimizations do not supply a separate $\F_{11}$ estimate
(supplementary Section~\ref{sec:mcattacks}).
These estimates concern that attack only and do not
establish Assumption (A). The cost of Vedenev's key-recovery
method~\cite{vedenev26keyrecovery} has not been estimated at either
length. The estimated distinguisher cost increases with the
extension degree, motivating our choice $s=6$
(supplementary Section~\ref{sec:mcattacks}).

\paragraph{Decoder guarantee and experiments.}
\phantomsection\label{sec:mcexperiments}%
At $m=14000$, Corollary~\ref{cor:mckey} gives the licensed radius
$\ell=93$, with $2\ell+1<d$ certified by $d\ge188$. The finite
DQI margin is positive, and the instance satisfies the length
condition of Theorem~\ref{thm:converse} for the Prange-restart budget
in Table~\ref{tab:candrules}.

The generic portfolio remains below DQI's exact expected quality
on all three seeds at $m=14000$ and at $m=20000$
(Table~\ref{tab:experiments}). The longer searches with budgets counted
in moves also remain below DQI at these lengths. In supplementary
Table~\ref{tab:scaling}, the best chains reach $0.521$ at $m=14000$
and $0.516$ at $m=20000$, and the annealing means are $0.519$ and
$0.514$. At $m=2000$ and $m=6000$
the same searches exceed DQI's value. We also supplied the secret
key to line search and simulated annealing along extension-field
lines of message polynomials. Starting from the portfolio's best
assignments, neither method improved those assignments on the three
tested seeds at $m=14000$, even when run for
$10^4$ moves (supplementary Section~\ref{sec:scalingkey}).
At $m=20000$ the same searches, run for $1800$~s on each of the three seeds, also found no improvement (supplementary Section~\ref{sec:mckeyholder}).
These runs compare DQI with specific algorithms that
exploit the key. Whether a different classical key holder can
match DQI remains open.

\subsection{Multiplicative polynomial intersection}\label{sec:logcauchy}

MPI applies the objective-dependent framework
of Section~\ref{sec:rulesalphabet} with a public decoder for signed
errors.

Fix a field $\F_{2^h}$ for which $M=2^h-1$ is a Mersenne prime,
and choose a generator $g$ of its multiplicative group. Let
$t_1,\dots,t_m,\zeta_1,\dots,\zeta_n$ be distinct field elements.
An assignment consists of exponents $x\in\mathbb Z_M^n$. At each
point $t_i$, it produces the nonzero field value
\[
  \prod_{j=1}^{n}(t_i+\zeta_j)^{x_j}.
\]
The objective rewards its discrete logarithm for being close,
modulo $M$, to a prescribed center $c_i$:
\begin{equation}\label{eq:mpi}
  \text{maximize}\quad
  \sum_{i=1}^{m}
  \cos\left(
    \frac{2\pi}{M}\bigl((Bx)_i-c_i\bigr)
  \right),
  \qquad
  B_{ij}=\log_g(t_i+\zeta_j)\in\F_M.
\end{equation}
The centers are drawn independently and uniformly, independently
of $B$. Taking logarithms converts the multiplicative expressions
into linear forms over $\F_M$, so \eqref{eq:mpi} is an instance
of the real-valued linear objective studied in
Lemma~\ref{lem:objlaw}. We call $B$ the \emph{log-Cauchy matrix}.

The discrete logarithms of shifted field elements are related to the
Chor--Rivest construction~\cite{chor88knapsack}, and the matrix occurs in the
polynomial lattices of Li, Ling, Xing, and
Yeo~\cite{li17polylattice}. Lapiha gives a rational-reconstruction
decoder for these lattices~\cite{lapiha21polylattice}.
Attacks on the Chor--Rivest cryptosystem either recover its secret field representation, generator and permutation, partly through subfields~\cite{vaudenay98chorrivest}, or solve its knapsack by lattice reduction~\cite{schnorr95chorrivest}. Here every parameter is public, and $\F_{2^{h}}$ with $h$ prime has no subfield beyond $\F_2$. Only the lattice attack therefore applies, and our lattice study below covers it. Our contribution is to pair this structure with the cosine objective and analyze the
resulting DQI guarantee.

\paragraph{Why signed errors can be decoded.}
For an error pattern $y\in\{0,\pm1\}^m$, exponentiating its
syndrome gives
\[
  g^{(B^{\mathsf T}y)_j}
  =\prod_{i=1}^{m}(t_i+\zeta_j)^{y_i}
  =\frac{P_y(\zeta_j)}{Q_y(\zeta_j)},
\]
where
\[
  P_y(T)=\prod_{i:y_i=1}(T+t_i),
  \qquad
  Q_y(T)=\prod_{i:y_i=-1}(T+t_i).
\]
Thus the syndrome supplies evaluations of a rational function whose
numerator and denominator record the positive and negative error
positions. When their degrees are sufficiently small, rational
reconstruction recovers the function, and its roots identify the
errors.

This mechanism depends on the error values being $\pm1$. It does
not require a decoder for arbitrary Hamming errors. The additional
separation needed for DQI is established under the parity condition
in the following theorem.

\begin{theorem}[multiplicative polynomial intersection]\label{thm:logcauchy}
Let $M=2^h-1>3$ be prime, let $m+n\le2^h$, and set
$\ell=\lfloor(n-1)/2\rfloor$. Let $B$ be the log-Cauchy matrix
of \eqref{eq:mpi}, with distinct points and rank $n$, so that
$\alpha=n/m$.
\begin{enumerate}
  \item[(i)] \emph{Worst-case signed-error decoding.}
    The syndrome map is injective on $\{0,\pm1\}$ patterns of
    weight at most $\ell$. It can be inverted in
    $O(n^3+m\ell+nh)$ field operations by exponentiating the
    syndrome, reconstructing a rational function from its $n$
    values, dividing out the numerator--denominator gcd, and
    scanning the $t_i$ for roots.

    When $n$ is even, the decoder is licensed for $\{\pm1\}$ in
    the sense of Definition~\ref{def:licensedS}. Its reach
    approaches $\alpha/2$ along an unbounded sequence with
    $n/m\to\alpha$.

  \item[(ii)] \emph{Decoding random signs.}
    For a pattern of weight $1\le k\le2\ell$ with independent
    uniform signs, the same decoder recovers the pattern whenever
    both the numerator and denominator degrees are at most
    $\ell$. The probability that this condition fails is
    \[
      \varepsilon_k
      =\Pr\left[
        \left|\mathrm{Bin}(k,\tfrac12)-\frac{k}{2}\right|
        >\ell-\frac{k}{2}
      \right]
      \le
      2\exp\left(
        -\frac{2(\ell-k/2)^2}{k}
      \right).
    \]
    The zero pattern is recovered exactly.
    Consequently, the decoder is \emph{licensed in measure} for signed
    patterns at every reach $\lambda<\alpha$: the fraction of
    failures at each allowed weight vanishes along an unbounded
    admissible sequence.

  \item[(iii)] \emph{DQI quality.}
    Let $n$ be even, so that $n\ge2\ell+2$.
    Under the cosine objective, DQI using the worst-case license
    in (i) achieves expected mean cosine
    \[
      \sqrt{2\lambda(1-\lambda)}-o(1)
    \]
    at reach $\lambda=\ell/m\to\alpha/2$.
\end{enumerate}
\end{theorem}

Part (iii) follows by applying the cosine quality law of
Proposition~\ref{prop:alphabet} to the license in part (i).
Part (ii) gives an additional average-case decoding guarantee,
but the stated DQI quality uses only the worst-case one.
Supplementary Section~\ref{sec:multiplicative} gives the proof, the provenance of the construction and the lattice
analysis.

The asymptotic statements are conditional on an unbounded sequence
of admissible Mersenne-prime parameters. Whether infinitely many
Mersenne primes exist is open. At the finite parameters used below,
the decoder license and the exact quality from
Lemma~\ref{lem:objlaw} apply directly.

\paragraph{Application of the rules.}
The construction supplies signed-error reach approaching
$\alpha/2$, which exceeds the cosine threshold at the parameters
in Table~\ref{tab:candrules}. A Hamming license at the same finite
radius would require a decoder-code distance near the Singleton
bound, and that distance is unknown. We therefore evaluate
rules \gate{1} and \gate{2} using their restricted-alphabet forms
from Proposition~\ref{prop:alphabet}.

The rational-reconstruction decoder, whose data classical algorithms also receive, solves a sparse syndrome
problem for $\deccode$. A classical optimizer would instead need
to find a high-quality word in $\optcode$, or equivalently round
the target toward that code under the cosine objective.
We have no method converting the supplied decoder into such an
optimizer at DQI's quality.

\paragraph{Classical attacks and hardness assumption.}
In the exponent representation, the objective can be viewed as
optimization over an $M$-ary lattice. We attack it using lattice
reduction on randomly selected subsets of rows, with LLL and its
stronger block variant BKZ, at block size $20$ in the runs of
Table~\ref{tab:experiments} and at block sizes up to $60$ in
supplementary Section~\ref{sec:scalinglattice}.

At $h=13$ and $m=600$, the best score obtained with BKZ-20
exceeds DQI's expected quality on all three tested seeds. On seed
$3101$, the mean score over subsets also exceeds it at the tested block sizes of $30$ and above
(supplementary Table~\ref{tab:scalinglattice}). Lattice reduction therefore
succeeds at that size, and longer runs of the generic searches
also exceed DQI there. At $m=2000$ and $m=6000$, every tested solver
remains below DQI. In the fixed-time runs the best scores stay
at least $0.09$ and $0.13$ below it, respectively
(Table~\ref{tab:experiments}). With budgets counted in moves, the best
line-search chains reach $0.169$ to $0.174$ at $m=2000$
($3\times10^6$ moves per chain), against DQI's $0.2433$, and
$0.126$ to $0.129$ at $m=6000$ ($10^6$ moves per chain),
against $0.2518$
(supplementary Table~\ref{tab:scaling}). At $h=17$ and $m=20000$ ($M=131071$, $n=1400$), the decoder recovers every pattern we tried at the licensed radius $\ell=699$, and the restart threshold of Theorem~\ref{thm:converse} is $1648$. There the portfolio, with $600$~s per search, stays below DQI's $0.2562$ by at least $0.158$ on three seeds.

At the additional length $m=1000$, individual subsets reduced at
block sizes $50$ and $60$ yield scores above DQI's value (up to
$0.240$ against $0.2338$), although the mean score over subsets
remains below it (at most $0.204$). We also model how this mean
score varies with block size, comparing the model with measurements
at block sizes up to $60$. It predicts a mean matching DQI at block
size about $125$ for $m=1000$ and, more roughly, about $440$ for
$m=2000$. For $m\ge4000$, it predicts no such match within its range
of validity (supplementary Section~\ref{sec:scalinglattice}). This model
describes only the mean score over subsets for this attack. A second model tracks the best subset and is calibrated to over-predict every run we made, which makes it more favorable to the attack. It puts DQI's value within reach at block size about $150$ for $m=2000$, a core-SVP cost of about $2^{44}$, and needs block size about $750$ at $m=6000$. At $h=17$ and $m=20000$, no block size up to $2000$ suffices. These are order-of-magnitude estimates, and the
experiments leave the hardness of the larger instances open.

The broader claim is stated separately as \emph{Assumption (C)}:
at $h=13$ and $m=6000$, no classical algorithm using at most
$2^{80}$ operations achieves expected mean cosine at least
$\Qdqi^{(m)}-0.01\simeq0.242$ on the candidate distribution. The same is assumed at $h=17$ and $m=20000$, with target $\Qdqi^{(m)}-0.01\simeq0.246$.
Like DQI, the solver is given the public decoding structure as well as the
optimization instance.
The tested attacks are far from this target at $m=6000$. On the
scale of supplementary Section~\ref{sec:scalinglaw} the target lies
$19.5$ standard deviations above $\Qpr$. The Gaussian reference
estimates that reaching this level requires about $2^{280}$ independent completions
(supplementary Section~\ref{sec:scalingconsequences}). The modeled
mean score of the lattice attack is $0.13$ at block size $274$,
where the core-SVP estimate for one shortest-vector oracle call is
$2^{80}$, and $0.227$ at block size about
$1500$, the largest for which the model applies. These are again
extrapolations of the tested attacks.

\subsection{Comparison of the candidates}

Table~\ref{tab:experiments} reports DQI's exact finite expected
quality and the best assignments found by the classical algorithms
within fixed wall-clock budgets. Supplementary
Table~\ref{tab:scaling} gives the higher scores reached with budgets
counted in moves.
For OPI and AMA, quality is the fraction of
satisfied rows, and for MPI it is
the mean cosine. The classical values describe the tested
implementations and budgets, so they are achievable scores, not
upper bounds on classical optimization.

\begin{table}[t]
\centering
\small
\begin{tabular}{lrrlrr}
\toprule
family & $m$ & $\Qpr$ & classical, fixed time & DQI (finite) & classical $-$ DQI\\
\midrule
OPI
  & 400 & 0.549 & 0.6700--0.6725 & 0.6856 & $-0.016$ to $-0.013$\\
OPI
  & 4000 & 0.550 & 0.5850--0.5900 & 0.7109 & $-0.126$ to $-0.121$\\
AMA ($s=6$)
  & 14000 & 0.498 & 0.5144--0.5171 & 0.5321 & $-0.018$ to $-0.015$\\
AMA ($s=6$)
  & 20000 & 0.498 & 0.5092--0.5128 & 0.5329 & $-0.024$ to $-0.020$\\
MPI, lattice
  & 600 & 0.070 & 0.2569--0.2873 & 0.2235 & $+0.033$ to $+0.064$\\
MPI
  & 2000 & 0.070 & 0.1493--0.1521 & 0.2433 & $-0.094$ to $-0.091$\\
MPI
  & 6000 & 0.070 & 0.1146--0.1153 & 0.2518 & $-0.137$\\
MPI, $h=17$ & 20000 & 0.070 & 0.0966--0.0981 & 0.2562 & $-0.160$ to $-0.158$\\
\bottomrule
\end{tabular}
\caption{Finite-length comparisons with classical algorithms.
OPI, AMA and MPI denote optimal polynomial intersection, alternant
max-agreement and multiplicative polynomial intersection.
The classical column gives the range, over three seeds, of the
best scores found within the fixed wall-clock budgets of
supplementary Table~\ref{tab:solver-budgets}.
The row marked lattice reports the attack of
Section~\ref{sec:logcauchy}, using the best score over the subsets examined.
The remaining rows report the generic portfolio. Searches with
larger budgets, counted in moves, are reported in supplementary
Table~\ref{tab:scaling}. The DQI column is the exact expected quality from
Lemma~9.2 of~\cite{jordan24dqi} or Lemma~\ref{lem:objlaw}.
A negative final column means that the tested classical score
is below DQI's expectation.}
\label{tab:experiments}
\end{table}

The classical budgets of Table~\ref{tab:experiments} are fixed in
wall-clock time, so increasing
the instance size also reduces the amount of search that can be
performed within a run. As a budget sensitivity check, tripling
the line-search and annealing budgets on one AMA seed improved
the classical score by $19$ rows ($0.0014$), leaving a gap of
$0.016$. On one MPI seed at $m=6000$, the
corresponding improvement was $0.0023$ in mean cosine, leaving a
gap of $0.135$. The runs with budgets counted in moves in supplementary
Section~\ref{sec:supp-scaling} show similarly slow improvement over up to
five orders of magnitude in the number of moves: about $0.002$ per
tenfold increase on the AMA instance at $m=14000$. Among the
rows of Table~\ref{tab:experiments}, they change the sign of the
comparison only for OPI at $m=400$.

Table~\ref{tab:candrules} connects these instances to the screening
framework. Both constructions have explicit decoder licenses and
positive margins over Prange completion. Neither is excluded by
the evaluated reach bounds, and the instances specified in
Table~\ref{tab:candrules} satisfy the Prange-restart threshold. Rule \gate{3} remains unresolved
for both: no encoder attaining the relevant existential line is
known. Since the optimization alphabets are prime, rule \gate{4}
supplies no stronger proper-subfield baseline.

\begin{table}[t]
\centering
\small
\begin{tabular}{l L{0.37\textwidth} L{0.37\textwidth}}
\toprule
check
  & alternant max-agreement\newline
    ($p=11$, $s=6$, $\alpha\simeq0.08$, $\rho=5/11$, $m=14000$)
  & multiplicative polynomial intersection\newline
    ($M=8191$, $\alpha=0.07$, cosine,
    limiting reach $\lambda=\alpha/2$)\\
\midrule
\gate{1}(a)
  & Distance condition satisfied:
    $\delta\ge\alpha/s\simeq0.0134
      >2\lambdacrit\simeq0.0038$.
  & Licensed signed-error reach approaches
    $\alpha/2=0.035>\lambdacrit^{\cos}\simeq0.0025$.\\
\gate{1}(b)
  & Packing bound permits sufficient reach:
    $H_p^{-1}(\alpha)\simeq0.028
      >\lambdacrit\simeq0.0019$.
  & Signed-pattern packing bound permits reach
    $0.196>0.0025$.\\
\gate{1}(c)
  & $\rho=5/11>\rhocrit\simeq0.0105$;
    Plotkin reach bound
    $(10/11)\alpha/2\simeq0.036>0.0019$.
  & Hamming-based restriction does not apply.\\
\gate{2}
  & Licensed decoder available with the key:
    $\ell=93$, $2\ell+1<d$, $d\ge188$
    (Corollary~\ref{cor:mckey}).
  & Public $\{\pm1\}$ license
    (Theorem~\ref{thm:logcauchy}).\\
\gate{3}
  & Unresolved: no encoder attaining $\Qstar$ is known.
  & Unresolved: no encoder attaining
    $\Qstar\simeq0.726$ is known.\\
\gate{4}
  & No stronger proper-subfield baseline.
  & No stronger proper-subfield baseline.\\
\midrule
margin
  & $+0.041$ at $\rho=\tfrac12$;
    $+0.038$ at $\rho=5/11$;
    $+0.034$ exact at $m=14000$.
  & $+0.19$ asymptotically;
    $+0.18$ exact at $m=6000$.\\
regime
  & Unsatisfiable:
    $\Qstar\simeq0.76$, $\regfac\simeq0.55$.
  & Unsatisfiable.\\
Prange restarts
  & $14000\ge12373$.
  & $6000\ge1729$.\\
\bottomrule
\end{tabular}
\caption{Evaluation of the screening conditions.
For MPI, rules \gate{1} and
\gate{2} use Proposition~\ref{prop:alphabet}.
The last row compares the instance length with the threshold
of Theorem~\ref{thm:converse} at $N=2^{40}$ and $\eta=0.05$.
Displayed numerical values are rounded.
At fixed field parameters, the constructions have bounded lengths:
$m\le p^s=1{,}771{,}561$ and $m+n\le2^h$, respectively.
The asymptotic expressions therefore describe the chosen parameter
regime. The finite licenses, margins and restart comparisons
apply directly to the stated instances.
At $m=20000$, the AMA instance has $\ell=133$, $d\ge268$,
finite margin $+0.035$, and restart threshold $20000\ge11770$.
At $m=40000$, the length at which Assumptions (A) and (B) are stated,
$\ell=266$, $d\ge534$, the finite margin is $+0.036$, and the restart
threshold is $40000\ge10991$.}
\label{tab:candrules}\label{tab:mcrules}\label{tab:mpirules}
\end{table}

%% file: sections/05-discussion.tex
\section{Discussion}\label{sec:discussion}\label{sec:scope}

The structure of DQI makes it possible to separate questions that are often addressed
together when proposing a quantum optimization algorithm: what quality a decoder would
provide, whether the code permits the required decoding reach, and whether a classical
algorithm already attains comparable quality. Theorem~\ref{thm:master} makes the first
comparison explicit, and the soundness rules turn coding bounds and classical
algorithms into conditions that candidate problems must satisfy.

\paragraph{What the rules exclude.}
Within the indicator objective and worst-case Hamming-decoding setting, the instance
parameters alone exclude some regimes. A vanishing decoder-code relative distance, or an
accepted-set density at most $\rhocrit(\alpha)$ whatever the code, precludes a positive
asymptotic margin over Prange completion (Corollary~\ref{cor:rules}). Proposition~\ref{prop:cap}
also limits the size of any advantage, with OPI attaining the margin cap. Whole
parameter regimes can therefore be screened before implementing a decoder or a classical attack.

The distinction between a decoding guarantee and an upper bound on decoding reach is
essential when applying these tests. Expander codes, for example, admit worst-case
decoders for a positive fraction of errors~\cite{sipser96expander}. For the expander
and Tanner decoders examined in our search, the established radii lie below
$\lambdacrit$, so they do not certify an advantage. Because these radii come from specific
decoders rather than from distance bounds, they exclude neither better decoders for the same codes
nor sparse codes as a class. The guarantees used
for other LDPC ensembles, and for the polar and locally recoverable codes examined
in the search, concern average-case decoding and require a different analysis
(supplementary Table~\ref{tab:search}).

\paragraph{Beyond Hamming decoding.}
The objective-dependent analysis changes the decoding
problem itself. Because the Fourier support of the objective fixes which errors must be decoded
(Lemma~\ref{lem:objlaw}), a restricted
support can admit a useful decoding guarantee even when the corresponding
Hamming guarantee is unavailable. MPI realizes
this mechanism by pairing the cosine objective with rational reconstruction for
signed errors (Theorem~\ref{thm:logcauchy}).

Other extensions arise from changing the type of decoding guarantee.
List, soft, and average-case decoders can support DQI beyond the reach covered by
our worst-case licenses~\cite{chailloux24rd,gao26soft}. A construction excluded under
the Hamming-license hypotheses may therefore remain useful under a different decoding
analysis.

\paragraph{Interpreting the two candidates.}
Corollary~\ref{cor:mckey} applies the established alternant decoder to AMA
and supplies a precise DQI guarantee. With the key, the problem becomes
trace-structured OPI, so the open question is how algorithms exploit that representation.
The key supplies the decoder DQI needs, but whether it also helps classical
optimization is an algorithmic question that our key-holder searches do not settle.
For solvers without the key, the hardness argument rests on Assumptions (A) and (B),
stated at $m=40000$ because the generic searches come within about $0.01$ of DQI
at $m=14000$ (supplementary Section~\ref{sec:scalingconsequences}).

For MPI, both quantum and classical algorithms
receive the same information, including the decoding structure.
The signed-error license establishes DQI's quality, but lattice methods give classical
algorithms a separate route to good assignments, and at the smallest tested lengths they
beat DQI (Section~\ref{sec:logcauchy}). A valid DQI guarantee can therefore coexist with
a stronger classical algorithm. At the larger tested lengths DQI stays ahead, and the
scaling study suggests that the effort the tested attacks need to match it grows with
length, although this extrapolates beyond the measured
budgets (supplementary Section~\ref{sec:supp-scaling}).
Assumption (C) states the broader hardness claim that these experiments leave open.

For both constructions, the evidence comes in three strengths. The comparison with Prange
restarts is a theorem, the comparisons with our tested solvers are experiments, and broader
classical hardness is stated as an assumption. One question also remains open. An efficient
classical encoder that reaches the existential line would remove the advantage, and none is
known (rule \gate{3}). Meeting the screening conditions therefore makes each problem a
candidate, not a proven quantum advantage.

\paragraph{Quality and computational cost.}
The guarantees concern expected optimization quality under the stated decoding
licenses. The Prange comparisons and counting bounds use random accepted sets or
random centers, and equation~\eqref{eq:reversemarkov} relates DQI's expectation to
the probability of reaching a lower target on an individual run.
The general extension-field formulation assumes the corresponding semicircle law
(supplementary Section~\ref{sec:proofstatus}). Both candidate objectives are over
prime fields, even though their constructions use auxiliary extension fields.

Turning a quality advantage into a runtime advantage requires accounting for
state preparation, coherent decoding, and the number of samples needed to obtain
the target score. Patamawisut et al.\ estimate gate counts growing roughly as
$m^2$ for their DQI circuit construction~\cite{patamawisut25circuit}.
Resource estimates for the present decoders and for preparation of the
cosine-objective state would make the comparisons more informative about
computational advantage.

\paragraph{Implications for automated discovery.}
The agent-guided search shows how the rules can organize exploration.
A proposal can be excluded by a certified bound, kept because it has a suitable
decoder and survives the tested attacks, or left unresolved because a required
guarantee is missing. These outcomes should stay distinct. Evaluating a scalar
inequality is cheap, but obtaining the distance bounds, decoding certificates
or classical algorithms that justify it can be the real research task.

In our search, decoder access and objective design proved productive directions,
though they are not an exhaustive classification of possible DQI applications.
Further progress could come from objectives with useful moments and restricted
Fourier support, from stronger decoding guarantees, or from classical algorithms that show
which apparent advantages survive closer scrutiny. The framework gives these investigations
a common language and names the specific mathematical or
algorithmic question that each candidate leaves open.

%% file: supplement/code-availability-snippet.tex
\section*{Code and data availability}
The scripts behind every numerical claim, the library they use and the stored results are available
at \url{https://github.com/magic-state-labs/dqi-screening-code}. Script $k$ of supplementary
Table~\ref{tab:checks} is the file in \texttt{scripts/} whose name begins with $k$, the results are
under \texttt{results/}, and the file \texttt{REPRODUCE.md} gives the command for each row of that
table.

\section*{Use of AI Tools}
AI tools were used extensively in generating and developing the core ideas of this work, as well as in drafting, revising, and
editing the manuscript. Computer-aided discovery was a central
component of the research: AI agents contributed to proposing
candidate problems and code constructions, evaluating them against
the soundness rules, and developing implementations and classical
attacks. This process helped shape the problem formulations and
research directions pursued in the paper, under human guidance
and selection. The search methodology is described in
Section~\ref{sec:intro} and supplementary
Section~\ref{sec:search}. The authors take responsibility for
the final manuscript and its scientific claims.

%% file: supplement/S0-notation.tex
\section*{Notation}

\begin{table}[h]
\centering
\footnotesize
\begin{tabular}{L{0.18\textwidth} L{0.60\textwidth} l}
\toprule
symbol & meaning & defined in\\
\midrule
$\F_q$, $q$ & the field of the instance and its size, $q=p^{b}$ in rule \gate{4} & Def.~\ref{def:instance}\\
$m$, $n$ & rows and columns of $B$, and $m$ is the length of both codes & Def.~\ref{def:instance}\\
$B$ & the instance matrix in $\F_q^{m\times n}$ & Def.~\ref{def:instance}\\
$b_i$ & the rows of $B$ & Lemma~\ref{lem:cosinelaw}\\
$F_i$, $\rho$ & the accepted set of row $i$ and the density $|F_i|/q$ & Def.~\ref{def:instance}\\
$\optcode$, $\deccode$ & the optimizer code $\{Bx\}$ and the decoder code that DQI decodes & \S\ref{sec:instances}\\
$\alpha$ & the rate $\rank(B)/m$ of the optimizer code & \S\ref{sec:instances}\\
$d$, $\delta$ & the minimum distance of the decoder code and $\delta=d/m$ & \S\ref{sec:instances}\\
$\ell$, $\lambda$ & the decoding radius and the reach $\ell/m$, licensed when $2\ell+1<d$ & Def.~\ref{def:licensed}\\
$\Qdqi(\lambda,\rho)$, $\Qdqi^{(m)}$ & the semicircle law and DQI's exact value at length $m$ & \eqref{eq:qdqi}\\
$\Qpr(\alpha,\rho)$ & one Prange completion, $\rho+(1-\rho)\alpha$ & \eqref{eq:qpr}\\
$\mu$, $\mu_m$ & the margin $\Qdqi-\Qpr$ and its value at length $m$ & \eqref{eq:margin}\\
$\mu^{+}$ & the margin against the baselines of \gate{3} and \gate{4} & \eqref{eq:muplus}\\
$\lambdacrit$, $\lambdacrit^{\cos}$ & the reach threshold for an indicator and for the cosine & \eqref{eq:lambdacrit}, Prop.~\ref{prop:alphabet}\\
$\rhocrit(\alpha)$ & the density threshold of \gate{1}(c) & \eqref{eq:rhocrit}\\
$\kappa(\alpha)$ & the margin cap & \eqref{eq:kappa}\\
$\Qstar$, $\regfac$ & the existential line and the regime factor & \eqref{eq:qstar}\\
$\KL{Q}{\rho}$ & the binary divergence, in nats & \eqref{eq:qstar}\\
$H_q$, $H_2$, $\dGV$ & the $q$-ary entropy, the binary entropy in bits, the Gilbert--Varshamov relative distance & \S\ref{sec:ruleproofs}\\
$b$, $c$, $\asub$ & the degree of $\F_q$ over $\F_p$, the subfield codimension and the rate seen over $\F_p$ & \S\ref{sec:subfield}\\
$\beta$ & the BKZ block size & \S\ref{sec:latticestudy}\\
$F$, $c_i$ & the per-row objective $f_i(v)=F(v-c_i)$ and its centers & Lemma~\ref{lem:objlaw}\\
$\bar F$, $\sigma$, $\gamma$ & the mean, standard deviation and skewness of $F$ under the uniform distribution & Lemma~\ref{lem:objlaw}\\
$S$ & the Fourier support $\{a\neq0:\hat F(a)\neq0\}$ of $F$ & Lemma~\ref{lem:objlaw}\\
$R$ & the length of an interval holding the values of $F$ & Thm.~\ref{thm:converse}\\
$N$, $\eta$ & the number of Prange restarts and the failure probability & Thm.~\ref{thm:converse}\\
$p$, $s$, $r$ & alternant max-agreement (AMA): the base prime, the extension degree and the number of parity rows & \S\ref{sec:mckey}\\
$M$, $h$, $t_i$, $\zeta_j$ & multiplicative polynomial intersection (MPI): the Mersenne prime $2^{h}-1$ and the points of the log-Cauchy matrix & \eqref{eq:mpi}\\
$N_{\mathrm{mv}}$ & the number of moves per search chain & \S\ref{sec:scaling}\\
$z(Q)$ & the number of Prange standard deviations by which $Q$ exceeds $\Qpr$ & \eqref{eq:zscore}\\
$\hat z(N)$, $N^*$ & the Gaussian reference level for $N$ completions, and the $N$ with $\hat z(N)=z(\Qdqi^{(m)})$ & \eqref{eq:samplinglevel}\\
$m'$, $m''$, $\delta_0$ & lattice attack: the subset size, the number of Gram--Schmidt norms below $M$, and the root-Hermite factor & \S\ref{sec:latticestudy}, \S\ref{sec:scalinglattice}\\
\bottomrule
\end{tabular}
\caption{Notation. The last column points to the definition in the main text or below.}
\label{tab:notation}
\end{table}

%% file: supplement/S1-proofs.tex
\section{Proofs of the rules, the objective law and the comparison with Prange restarts}
\label{sec:supp-proofs}

This section supplies the proofs for
Sections~\ref{sec:problem} and~\ref{sec:rulesalphabet}.
We first derive the reach threshold and the coding-theoretic
restrictions, then prove the sharp margin cap.
The next two parts establish the classical comparisons based on
source coding and subfield linear algebra. We finish with the
general-objective quality law and the finite-length comparison
against Prange restarts.

\subsection{Scope and assumptions}
\label{sec:rulesreading}
\label{sec:proofstatus}

The rules address three different questions. Rule \gate{1} asks
whether the code permits enough decoding reach to exceed the
Prange baseline. Rule \gate{2} asks whether an efficient decoder
with the required guarantee is available. Rules \gate{3} and
\gate{4} ask whether a stronger classical algorithm already reaches
DQI's quality. An upper bound on reach
can exclude an advantage under the stated quality law. The
absence of a known decoder does not establish an impossibility.

The three parts of rule \gate{1} use complementary information.
Part (a) compares the required reach with half the decoder-code
distance. Part (b) counts error patterns and syndromes. Part (c)
uses rate--distance bounds, so it can be evaluated without knowing
the particular code's distance. The sharp margin cap uses the
same rate constraint, but maximizes the resulting margin over
the accepted-set density and then over the rate.

Throughout the indicator-objective analysis, $\Qdqi$ denotes the
asymptotic quality function in \eqref{eq:qdqi}, and the exact
finite-length value is denoted by $\Qdqi^{(m)}$.
An inequality obtained by substituting finite parameters into
$\Qdqi$ is therefore an inequality for that function, rather
than automatically an inequality for $\Qdqi^{(m)}$.
The reach threshold, rule \gate{1}, and the sharp margin cap
use the field assumption stated in
Section~\ref{sec:qualitylaws}. The general-objective argument
below is proved over prime fields.

The classical baselines also require a probability model.
For the usual Prange comparison, the accepted sets are independent
of $B$ and of one another, and each is uniform among subsets of
its prescribed size. An information set is chosen without inspecting
these sets, and its completion uses only the accepted sets on
that information set. For shifted real objectives, the corresponding
assumption is that the centers are independent and uniform.
Rule \gate{4} allows a more general translation-invariant law,
subject to the selection and rank conditions stated in
Proposition~\ref{prop:remaining}.

The finite spectral identity requires the separation condition
in Definition~\ref{def:licensedS}. Efficient implementation
additionally requires a suitable decoder and state-preparation
procedure. The identity itself does not depend on their running
times.

\subsection{The reach threshold and rule \gate{1}}
\label{sec:ruleproofs}

\subsubsection{Inverting the quality law}

For fixed $\rho$, the quality function $\Qdqi(\lambda,\rho)$
increases from $\rho$ to $1$ as $\lambda$ increases from $0$
to $1-\rho$, and remains equal to $1$ thereafter. The threshold
is therefore obtained by finding where this increasing branch
crosses the Prange baseline.

\begin{proof}[Proof of Theorem~\ref{thm:master}]
Since $0<\alpha<1$ and $0<\rho<1$,
\[
  \rho<\Qpr=\rho+(1-\rho)\alpha<1.
\]
Write
\[
  \rho=\sin^2\phi,
  \qquad
  \Qpr=\sin^2\psi,
  \qquad
  0<\phi<\psi<\frac{\pi}{2}.
\]
On the increasing branch, put $\lambda=\sin^2\theta$ with
$0\le\theta\le\pi/2-\phi$. Then
\[
  \Qdqi(\lambda,\rho)
  =
  \bigl(\sin\theta\cos\phi+\cos\theta\sin\phi\bigr)^2
  =
  \sin^2(\theta+\phi).
\]
Consequently, the unique crossing occurs at
$\theta=\psi-\phi$, giving
\[
  \lambda^*
  =
  \sin^2(\psi-\phi)
  =
  \left(
    \sqrt{(1-\rho)\Qpr}
    -
    \sqrt{\rho(1-\Qpr)}
  \right)^2
  =
  \lambdacrit(\alpha,\rho).
\]
This proves both $0<\lambdacrit<1-\rho$ and
\[
  \Qdqi(\lambda,\rho)>\Qpr(\alpha,\rho)
  \quad\Longleftrightarrow\quad
  \lambda>\lambdacrit(\alpha,\rho).
\]

At $\rho=\tfrac12$, the baseline is $(1+\alpha)/2$, so
\[
  \lambdacrit
  =
  \frac12
  \left(1-2\sqrt{\Qpr(1-\Qpr)}\right)
  =
  \frac{1-\sqrt{1-\alpha^2}}{2}.
\]
Finally, if a family's limiting reach is at most
$U\le\lambdacrit$, continuity of the quality law gives
$\mu\le0$.
\end{proof}

\subsubsection{The distance restriction: rule \gate{1}(a)}

A worst-case decoding guarantee requires distinct correctable
patterns to have distinct syndromes. The Hamming license imposes
the stronger condition $2\ell+1<d$, which also prevents the
collisions relevant to the expected-quality calculation.

\begin{proof}[Proof of rule \gate{1}(a)]
At each finite length, licensing gives
\[
  \frac{\ell}{m}<\frac{d}{2m}.
\]
The inequality may become non-strict in the limit, so the
limiting parameters satisfy $\lambda\le\delta/2$.
If $\bar\delta$ is any upper bound on the limiting relative
distance, then
\[
  \lambda\le\frac{\delta}{2}\le\frac{\bar\delta}{2}.
\]
Thus $\bar\delta\le2\lambdacrit$ implies
$\lambda\le\lambdacrit$, and Theorem~\ref{thm:master}
gives $\mu\le0$. A positive margin therefore requires
$\bar\delta>2\lambdacrit$.

At density $\tfrac12$, substituting the threshold from
Theorem~\ref{thm:master} gives
\[
  \delta>1-\sqrt{1-\alpha^2}.
\]
If instead the relative distance tends to zero, the licensed
reach also tends to zero. Since $\Qdqi(0,\rho)=\rho$,
\[
  \mu\longrightarrow
  \rho-\Qpr(\alpha,\rho)
  =-(1-\rho)\alpha.
\]

Conversely, if $\delta>2\lambdacrit$, any available decoder
whose licensed limiting reach lies in
$(\lambdacrit,\delta/2)$ gives a positive margin.
The distance inequality permits such a reach, but it does not
by itself supply the decoder.
\end{proof}

\subsubsection{The packing restriction: rule \gate{1}(b)}

The packing argument counts the patterns that a decoder must
distinguish. It applies independently of how the decoder is
implemented.

We use the $q$-ary entropy
\[
  H_q(x)
  =
  x\log_q(q-1)
  -x\log_q x
  -(1-x)\log_q(1-x),
  \qquad
  0\le x\le1-\frac1q.
\]
On this interval, $H_q$ is continuous and strictly increasing
from $0$ to $1$.

\begin{lemma}[syndrome counting]\label{lem:realizability}
Let the decoder code have length $m$ and rate $1-\alpha$.
Suppose every error pattern of weight at most $\ell$ can be
recovered from its syndrome, and put $\lambda=\ell/m$.
If $\lambda\le1-1/q$, then
\[
  \alpha\ge H_q(\lambda)-o(1)
  \qquad (m\to\infty).
\]
\end{lemma}

\begin{proof}
The syndrome space has $q^{\alpha m}$ elements.
Unique recovery therefore requires
\[
  V_q(m,\ell)
  :=
  \sum_{j=0}^{\ell}\binom mj(q-1)^j
  \le q^{\alpha m}.
\]
Equivalently, the Hamming balls of radius $\ell$ around the
decoder-code words are disjoint.

For $\lambda=\ell/m\le1-1/q$, the elementary binomial bounds give
\[
  \frac{q^{mH_q(\lambda)}}{m+1}
  \le
  V_q(m,\ell)
  \le
  (m+1)q^{mH_q(\lambda)}.
\]
For the lower bound, retain the term at $j=\ell$.
For the upper bound, bound each term using $H_q(j/m)$ and
use monotonicity on the stated interval.
Taking logarithms yields
\[
  \alpha
  \ge
  H_q(\lambda)-\frac{\log_q(m+1)}{m},
\]
which proves the claim. The error tends to zero uniformly
over $q\ge2$.
\end{proof}

\begin{corollary}\label{cor:gvceiling}
Under the hypotheses of Lemma~\ref{lem:realizability},
\[
  \lambda\le H_q^{-1}(\alpha)+o(1).
\]
The quantity $H_q^{-1}(\alpha)$ is also the
Gilbert--Varshamov relative-distance threshold $\dGV$
for codes of rate $1-\alpha$.
\end{corollary}

\begin{proof}
For fixed $q$, the conclusion follows by continuity of
$H_q^{-1}$ on $[0,1]$.

The conclusion also holds when the field size varies.
Indeed, as $q\to\infty$,
\[
  H_q(x)=x+O\!\left(\frac1{\ln q}\right)
\]
uniformly in $x\in[0,1-1/q]$. Hence
$H_q^{-1}(\alpha)=\alpha+o(1)$ uniformly in
$\alpha\in[0,1]$, and the conclusion follows from the lemma.
Together with the fixed-$q$ argument, this covers arbitrary
sequences of field sizes.

The Gilbert--Varshamov threshold satisfies
$1-\alpha=1-H_q(\dGV)$, giving
$\dGV=H_q^{-1}(\alpha)$.
The two thresholds share an entropy expression. The bound on
decoding reach here follows from packing.
\end{proof}

For binary instances with $\rank B=n$, this recovers the
entropy restriction on $\ell/m$ used by Jordan et al.\
for max-XOR-SAT~\cite{jordan24dqi}.

\begin{proof}[Proof of rule \gate{1}(b)]
A Hamming-licensed radius has $\lambda<1/2$ at finite
length, and $1/2\le1-1/q$. Corollary~\ref{cor:gvceiling}
therefore applies. If the resulting limiting upper bound
on reach is at most $\lambdacrit$, then
Theorem~\ref{thm:master} gives $\mu\le0$.
At fixed $q$, a positive limiting margin consequently requires
\[
  H_q^{-1}(\alpha)>\lambdacrit(\alpha,\rho).
\]
When $q$ varies, the same comparison is made with the limiting
packing bound.
\end{proof}

\subsubsection{Rate--distance restrictions: rule \gate{1}(c)}

The Singleton bound replaces the code-specific distance by
a bound depending only on the optimizer-code rate.
The following calculation identifies the densities at which
even this largest permitted reach cannot beat Prange.

\begin{lemma}\label{lem:densitywindow}
For $0<\rho<1$, $0<\alpha<1$, and
$0<\lambda\le1-\rho$,
\begin{equation}\label{eq:general-window}
  \Qdqi(\lambda,\rho)>\Qpr(\alpha,\rho)
  \quad\Longleftrightarrow\quad
  \lambda-2\lambda\rho
  +2\sqrt{\lambda\rho(1-\lambda)(1-\rho)}
  >\alpha(1-\rho).
\end{equation}
If $\lambda=\alpha/2\le1-\rho$, this is equivalent to
\begin{equation}\label{eq:mds-window}
  \rho(1-\rho)>
  \frac{\alpha}{8(1-\alpha/2)}.
\end{equation}
\end{lemma}

\begin{proof}
On the increasing branch,
\[
  \Qdqi(\lambda,\rho)
  =
  \lambda+\rho-2\lambda\rho
  +2\sqrt{\lambda\rho(1-\lambda)(1-\rho)}.
\]
Subtracting $\Qpr=\rho+\alpha(1-\rho)$ gives
\eqref{eq:general-window}.

At $\lambda=\alpha/2$, the comparison reduces to
\[
  2\sqrt{
    \frac{\alpha}{2}
    \left(1-\frac{\alpha}{2}\right)\rho(1-\rho)
  }
  >
  \frac{\alpha}{2}.
\]
Both sides are positive. Squaring and rearranging gives
\eqref{eq:mds-window}.
\end{proof}

\begin{proof}[Proof of rule \gate{1}(c)]
The decoder code has dimension $(1-\alpha)m$.
The Singleton bound~\cite{singleton64} gives
$d\le\alpha m+1$, so the licensing condition implies
\[
  2\ell+1<d\le\alpha m+1,
  \qquad
  \lambda=\frac{\ell}{m}<\frac{\alpha}{2}.
\]
It remains to determine when the limiting Singleton reach
$\alpha/2$ exceeds $\lambdacrit$.

If $\alpha/2>1-\rho$, then
$\Qdqi(\alpha/2,\rho)=1>\Qpr$.
In this case
$\rho>1-\alpha/2>1/2>\rhocrit(\alpha)$.

Now suppose $\alpha/2\le1-\rho$.
By Lemma~\ref{lem:densitywindow},
\[
  \Qdqi(\alpha/2,\rho)\le\Qpr
  \quad\Longleftrightarrow\quad
  \rho(1-\rho)\le
  \frac{\alpha}{4(2-\alpha)}.
\]
The two roots of the corresponding equality are
\[
  \rho_\pm
  =
  \frac12
  \left(
    1\pm\sqrt{\frac{2(1-\alpha)}{2-\alpha}}
  \right).
\]
The lower root is $\rhocrit(\alpha)$.
The upper root lies outside the branch under consideration,
because
\[
  \rho_+>1-\frac{\alpha}{2}.
\]
To verify this, square the equivalent inequality
\[
  \sqrt{\frac{2(1-\alpha)}{2-\alpha}}>1-\alpha;
\]
after division by $1-\alpha>0$, it becomes
$2>(1-\alpha)(2-\alpha)$, which holds for $0<\alpha<1$.
Thus
\[
  \Qdqi(\alpha/2,\rho)\le\Qpr
  \quad\Longleftrightarrow\quad
  \rho\le\rhocrit(\alpha).
\]
Theorem~\ref{thm:master} gives the Singleton form of the rule.

For fixed $q$, the asymptotic rate--distance form of the
Plotkin bound gives
\[
  \delta\le(1-1/q)\alpha+o(1).
\]
The limiting licensed reach is therefore at most
$(1-1/q)\alpha/2$. Comparing this value with
$\lambdacrit$ gives the Plotkin form.
\end{proof}

\begin{corollary}\label{cor:exact}
Suppose each accepted set contains one value, so $\rho=1/q$.
For each fixed $\alpha\in(0,1)$, the Singleton form of
rule \gate{1}(c) excludes a positive asymptotic margin whenever
\[
  q\ge\frac1{\rhocrit(\alpha)}.
\]
\end{corollary}

\begin{proof}
For $0<\alpha<1$, the quantity $\rhocrit(\alpha)$ is positive.
The stated condition is exactly
$\rho=1/q\le\rhocrit(\alpha)$.
\end{proof}

\subsection{The sharp margin cap}
\label{sec:capproof}

The same Singleton restriction also bounds how large an
advantage can be. We first optimize over density at a fixed
rate, then optimize over the rate. The argument below
bounds the asymptotic quality function $\Qdqi$, not $\Qdqi^{(m)}$, and the resulting
limiting margin.

\begin{proof}[Proof of Proposition~\ref{prop:cap}]
Let $0<\alpha<1$ and put $u=\alpha/2$.
At every finite length the licensed reach satisfies $\lambda<u$.

Suppose first that $u\le1-\rho$.
Strict monotonicity on the increasing branch gives
\[
  \Qdqi(\lambda,\rho)-\Qpr(\alpha,\rho)
  <
  \Qdqi(u,\rho)-\Qpr(\alpha,\rho).
\]
Since $\Qpr=\rho+2u(1-\rho)$, expansion yields
\[
  \Qdqi(u,\rho)-\Qpr(\alpha,\rho)
  =
  2\sqrt{u(1-u)\rho(1-\rho)}-u.
\]
The inequality $\rho(1-\rho)\le1/4$ therefore gives
\[
  \Qdqi(\lambda,\rho)-\Qpr(\alpha,\rho)
  <
  \sqrt{u(1-u)}-u
  =
  \kappa(\alpha).
\]

If $u>1-\rho$, use $\Qdqi\le1$ instead:
\[
  \Qdqi(\lambda,\rho)-\Qpr(\alpha,\rho)
  \le
  (1-\rho)(1-2u)
  <
  u(1-2u).
\]
This is also at most $\kappa(\alpha)$, because
\[
  \kappa(\alpha)-u(1-2u)
  =
  \sqrt{u(1-u)}
  \left(1-2\sqrt{u(1-u)}\right)
  \ge0.
\]
Thus the comparison with $\kappa(\alpha)$ holds at every
density. Passing to a family limit gives
\[
  \mu\le\kappa(\alpha).
\]
The inequality is no longer necessarily strict, since
$\lambda$ may approach $u$.

To maximize over $\alpha$, regard
$\sqrt{u(1-u)}-u$ as a function of $u\in[0,1/2]$.
It is strictly concave in the interior, and its derivative
vanishes when
\[
  \frac{1-2u}{2\sqrt{u(1-u)}}=1.
\]
The relevant solution is
\[
  u=\frac{2-\sqrt2}{4},
  \qquad
  \alpha=1-\frac1{\sqrt2}.
\]
At this point $u(1-u)=1/8$, giving
\[
  \max_\alpha\kappa(\alpha)
  =
  \frac{\sqrt2-1}{2}.
\]

We next identify the conditions for attaining this value
in the limit. The saturated branch satisfies
\[
  \mu\le u(1-2u)\le\frac18
  <
  \frac{\sqrt2-1}{2},
\]
so it cannot attain the global maximum.
On the other branch, the expression
\[
  2\sqrt{u(1-u)\rho(1-\rho)}-u
\]
has its unique global maximum at
$\rho=1/2$ and $u=(2-\sqrt2)/4$.
A sequence attaining the cap must therefore have
\[
  \rho\to\frac12,
  \qquad
  \alpha\to1-\frac1{\sqrt2}.
\]
It must also have $\lambda-u\to0$.
Otherwise, a subsequence with $\lambda\le u-\varepsilon$
would retain a strictly positive gap in the continuous,
strictly increasing quality law near these limiting parameters.

Licensing and Singleton now force $\delta\to\alpha$.
For fixed $q$, the Plotkin bound prevents this, since
$\delta\le(1-1/q)\alpha+o(1)$.
The same obstruction applies to any subsequence with bounded
field size, so attaining the cap requires $q\to\infty$.

Finally, OPI supplies an attaining sequence.
Its Reed--Solomon decoder code has $d=n+1$, where
$\alpha=n/m$. Taking the largest licensed radius
$\ell=\lfloor(n-1)/2\rfloor$ gives $\ell/m\to\alpha/2$.
With $q\to\infty$ and accepted-set density tending to $1/2$,
its limiting margin is $\kappa(\alpha)$.
Choosing $\alpha\to1-1/\sqrt2$ attains the global cap,
using the established OPI quality law~\cite{jordan24dqi}.
\end{proof}

\subsection{Counting high-quality codewords}
\label{sec:counting}

The existential line comes from counting how many codewords
are expected to reach a given score.
This argument needs independence across the accepted sets,
but not independence between the scores of different codewords.

We use natural logarithms in relative entropy:
\[
  \KL{Q}{\rho}
  =
  Q\ln\frac Q\rho
  +(1-Q)\ln\frac{1-Q}{1-\rho},
\]
with the usual convention $0\ln0=0$.
Binary entropy $H_2$ is measured in bits.

\begin{lemma}\label{lem:firstmoment}
Let the accepted sets be independent, uniform subsets of
$\F_q$ of density $\rho$, independent of an optimizer code
with $q^{\alpha m}$ codewords.
Let $N_Q$ count codewords of quality at least $Q$, where
$\rho<Q\le1$, and put $j=\lceil Qm\rceil$.
Then
\[
  \frac{
    \exp\!\left(m\alpha\ln q-m\KL{j/m}{\rho}\right)
  }{m+1}
  \le
  \mathbb E[N_Q]
  \le
  \exp\!\left(m\alpha\ln q-m\KL{Q}{\rho}\right).
\]
In particular, for fixed $\rho\in(0,1)$,
\[
  \mathbb E[N_Q]
  =
  \exp\!\left(
    m\bigl(\alpha\ln q-\KL{Q}{\rho}+o(1)\bigr)
  \right),
\]
uniformly for $Q$ in compact subsets of $(\rho,1]$.

Moreover:
\begin{enumerate}
  \item[(i)]
  The expected number of perfect codewords is
  $(\regfac)^m$.

  \item[(ii)]
  If $\regfac<1$, then every fixed quality strictly above
  the limiting existential line is absent with probability
  tending to one. More explicitly, using the finite-parameter
  line $\Qstar$, for $\varepsilon>0$ with
  $\Qstar+\varepsilon\le1$,
  \[
    \Pr[N_{\Qstar+\varepsilon}\ge1]
    \le e^{-2m\varepsilon^2}.
  \]

  \item[(iii)]
  At $\rho=1/2$, the condition $\regfac>1$ is equivalent to
  $\alpha\log_2q>1$.
  If $\alpha\log_2q<1$, then
  \[
    H_2(\Qstar)=1-\alpha\log_2q,
    \qquad \Qstar>\frac12.
  \]
\end{enumerate}
\end{lemma}

\begin{proof}
Fix a codeword $c$.
The events $c_i\in F_i$ are independent, each with probability
$\rho$. Its number of satisfied rows therefore has distribution
$\mathrm{Bin}(m,\rho)$.
Linearity of expectation gives
\[
  \mathbb E[N_Q]
  =
  q^{\alpha m}
  \Pr[\mathrm{Bin}(m,\rho)\ge Qm].
\]
The Chernoff bound gives the displayed upper bound.
For the lower bound, retain the binomial term at $j$ and use
\[
  \binom mj
  \ge
  \frac{\exp(mh(j/m))}{m+1},
\]
where $h$ is binary entropy in nats.
The asymptotic expression follows by continuity of
$\KL{\cdot}{\rho}$, including at $Q=1$.

For (i), a fixed codeword satisfies all rows with probability
$\rho^m$, so
\[
  \mathbb E[N_1]=q^{\alpha m}\rho^m=(\regfac)^m.
\]

For (ii), the defining equation is
$\KL{\Qstar}{\rho}=\alpha\ln q$.
The divergence is increasing for $Q>\rho$, and
\[
  \frac{\partial^2}{\partial Q^2}\KL{Q}{\rho}
  =
  \frac1{Q(1-Q)}
  \ge4.
\]
Hence
\[
  \KL{\Qstar+\varepsilon}{\rho}
  -
  \KL{\Qstar}{\rho}
  \ge2\varepsilon^2.
\]
The expectation bound and Markov's inequality give the result.
If $\Qstar+\varepsilon>1$, the event is empty.

For (iii), $\regfac=q^\alpha/2$, and
\[
  \KL{Q}{1/2}
  =
  \bigl(1-H_2(Q)\bigr)\ln2.
\]
Substitution into the equation for $\Qstar$ completes the proof.
\end{proof}

The preceding result is an upper bound on the optimum for a prescribed code.
A large expected count does not, by itself, prove that suitable
codewords exist with high probability. In particular, the term
``satisfiable regime'' refers to the first-moment condition
$\regfac\ge1$. It is not an existence theorem for every
optimizer code. An existence result requires an additional
argument.

\subsection{Rule \gate{3} and lossy source coding}
\label{sec:rd}

Max-agreement can also be viewed as lossy compression.
The source is the sequence of accepted sets, and a reconstruction
is a word in the optimizer code. A coordinate incurs distortion
one when the reconstructed value lies outside its accepted set.
The connection with the rate--distortion limit is discussed by
Chailloux and Tillich~\cite{chailloux24rd}. We derive the
particular form needed here.

\begin{lemma}\label{lem:setmemrd}
Let the source emit independent, uniformly random subsets
$F\subseteq\F_q$ of size $\rho q$, and use distortion
$\mathbf1[y\notin F]$.
For $0\le D\le1-\rho$, the rate--distortion function,
in nats per coordinate, is
\[
  R_\rho(D)
  =
  \ln q
  -(1-D)\ln(\rho q)
  -D\ln((1-\rho)q)
  -H_2(D)\ln2
  =
  \KL{1-D}{\rho}.
\]
For $D\ge1-\rho$, it is zero.
For singleton accepted sets, this reduces to
\[
  R_{1/q}(D)
  =
  \ln q-D\ln(q-1)-H_2(D)\ln2
\]
on the nonzero branch.
\end{lemma}

\begin{proof}
The source law and distortion are unchanged by applying the
same permutation of $\F_q$ to $F$ and to $y$.
Average any feasible reconstruction channel over these
permutations. Its expected distortion is unchanged, and
convexity of mutual information in the channel ensures that
its mutual information does not increase.

We may therefore use a permutation-invariant channel.
For some $D'\le D$, it assigns total probability $1-D'$
uniformly to $F$ and total probability $D'$ uniformly to
the complement. Its output is uniform on $\F_q$, and
\[
  H(Y\mid F)
  =
  (1-D')\ln(\rho q)
  +D'\ln((1-\rho)q)
  +H_2(D')\ln2.
\]
Thus
\[
  I(F;Y)=\KL{1-D'}{\rho}.
\]
This is decreasing in $D'$ on $[0,1-\rho]$, so the optimum
under $D'\le D$ occurs at $D'=D$.

For $D\ge1-\rho$, a reconstruction independent of $F$
already has expected distortion $1-\rho$ and requires
zero rate.
\end{proof}

\begin{corollary}\label{cor:rdisexistential}
Fix $q$ and $\rho$, and let $\Qstar$ be the line in
Definition~\ref{def:regime}.
Among arbitrary reconstruction codebooks of asymptotic rate
$\alpha$, the optimal expected distortion tends to
$1-\Qstar$.

At finite length, any codebook with at most $q^{\alpha m}$
words, and any encoder taking values in that codebook,
has expected distortion at least $1-\Qstar$, with $\Qstar$
evaluated at those finite parameters.
\end{corollary}

\begin{proof}
Let $F^m=(F_1,\ldots,F_m)$ be the source and let
$Y^m$ be the encoder's output.
Write $D_i=\Pr[Y_i\notin F_i]$ and
$D=m^{-1}\sum_iD_i$.
Independence of the source coordinates and convexity of
$R_\rho$ give
\[
  mR_\rho(D)
  \le
  \sum_{i=1}^m R_\rho(D_i)
  \le
  \sum_{i=1}^m I(F_i;Y_i)
  \le
  I(F^m;Y^m)
  \le
  H(Y^m)
  \le
  \alpha m\ln q.
\]
Therefore $R_\rho(D)\le\alpha\ln q$.
By Lemma~\ref{lem:setmemrd}, the least such $D$ is
$1-\Qstar$, proving the finite converse.

The achievability part of the rate--distortion theorem
provides codebooks approaching this distortion.
If $\regfac<1$, the boundary is determined by
$\KL{1-D}{\rho}=\alpha\ln q$.
If $\regfac\ge1$, then
$R_\rho(0)=\ln(1/\rho)\le\alpha\ln q$, and the limiting
distortion is zero.
\end{proof}

The achievability statement concerns the choice of codebook.
It does not establish that a prescribed optimizer code reaches
the line, nor does it provide an efficient encoder for that code.
Rule \gate{3} rests on this distinction.

\begin{proposition}[rule \gate{3}]\label{prop:rdgeneral}
Let the accepted sets be independent, uniformly random of
density $\rho$, and independent of $B$.
For an optimizer code of rate $\alpha$, every algorithm
returning a word in that code has expected quality at most
the finite-parameter value $\Qstar$.

Consequently, whenever the stated asymptotic quality law
applies,
\[
  \Qdqi(\lambda,\rho)\le\Qstar.
\]
If the optimizer-code family has a polynomial-time encoder
whose expected quality tends to $\Qstar$, then
$\mu^+\le0$.
\end{proposition}

\begin{proof}
Condition on $B$.
The optimizer code is a codebook of size $q^{\alpha m}$
independent of the accepted sets, so the finite converse in
Corollary~\ref{cor:rdisexistential} applies.
It also applies to randomized outputs, including the
measurement outcome of DQI.

Passing to the limit gives the inequality for $\Qdqi$.
An efficient encoder attaining $\Qstar-o(1)$ then supplies
a classical baseline with the same limiting upper bound,
and hence
\[
  \mu^+
  \le
  \Qdqi-\Qstar
  \le0.
\]
\end{proof}

Polar codes give examples where efficient source coding
is possible. For a binary symmetric source with Hamming
distortion, they attain the rate--distortion limit using
an $O(m\log m)$ encoder~\cite{korada09polar}.
At a fixed prime alphabet, polar constructions attain the
symmetric rate--distortion bound, defined by restricting
to test channels with uniform reconstruction
marginal~\cite{karzand10qary}.
The optimal channel in Lemma~\ref{lem:setmemrd} has this
uniform marginal, so the symmetric and unrestricted bounds
coincide for the source considered here.
Translation invariance also allows a good affine codebook
to be translated to its underlying linear code without
changing expected distortion.

These examples already violate the distance condition for
a positive DQI margin. To see this, let $m=2^t$ and
\[
  G=
  \begin{pmatrix}
    1&0\\
    1&1
  \end{pmatrix}.
\]
If the optimizer code is spanned by rows of $G^{\otimes t}$
indexed by $A$, its dual is spanned by the complementary
rows of
\[
  (G^{-{\mathsf T}})^{\otimes t}
  =
  \begin{pmatrix}
    1&-1\\
    0&1
  \end{pmatrix}^{\otimes t}.
\]
The row indexed by $i$ has weight
$2^{t-\mathrm{wt}(i)}$.
At a rate bounded away from one, the complementary index set
contains a positive fraction of all indices, and therefore
contains an index with
$\mathrm{wt}(i)\ge t/2-O(\sqrt t)$.
The dual distance is consequently at most
\[
  2^{t/2+O(\sqrt t)}
  =
  m^{1/2+o(1)},
\]
so its relative distance vanishes.
Polar source codes thus illustrate efficient optimal
encoding, but they do not supply such an encoder for an
arbitrary candidate optimizer code with positive dual distance.

\subsection{Rule \gate{4}: linear algebra over a subfield}
\label{sec:subfield}

When $q=p^b$, each field element is a vector of dimension
$b$ over $\F_p$. An assignment in $\F_q^n$ therefore
contains $bn$ prime-field coordinates.
Because multiplication by a fixed element of $\F_q$ is
$\F_p$-linear, every $(Bx)_i$ is a linear function of
these coordinates.

Suppose an accepted set contains an affine subspace $a+V$,
where $V$ has codimension $c$ over $\F_p$.
Choose independent linear forms
$\psi_1,\ldots,\psi_c$ whose common kernel is $V$.
Then
\[
  y\in a+V
  \quad\Longleftrightarrow\quad
  \psi_j(y)=\psi_j(a)
  \quad\text{for }j=1,\ldots,c.
\]
Thus satisfying this row costs $c$ prime-field equations,
whereas fixing its value completely costs $b$.
If the resulting equations are independent, this can satisfy
more rows than an information set over $\F_q$.

A set of density $\rho$ can contain such a coset only if
$p^{b-c}\le\rho p^b$, so
\[
  c\ge\log_p(1/\rho).
\]
At density $1/2$ over $\F_{2^b}$, equality $c=1$
means that the accepted set itself is a hyperplane coset.
Over a prime field, every proper nonempty affine subspace
is a point, so $c=b=1$ and the construction reduces to
the ordinary Prange baseline.

\begin{proposition}[information sets over the prime field]
\label{prop:subfield}
Let $q=p^b$ and fix $B$.
Suppose the accepted sets are independent, each has size
$\rho q$, and each has a translation-invariant distribution.
Let $c\ge1$ and put
\[
  \asub=\min\left(1,\frac{b\alpha}{c}\right).
\]
Choose a set $J$ of $\lfloor\asub m\rfloor$ rows, using
only $B$, independent randomness, and translation-invariant
data from the accepted sets.
For every $i\in J$, suppose a coset
$a_i+V_i\subseteq F_i$ of codimension $c$ is available,
with its selection depending only on the selected rows'
sets and the data already used to choose $J$.

If the resulting $c|J|$ equations are linearly independent
over $\F_p$, Gaussian elimination finds an assignment
satisfying every row in $J$.
A completion using only the selected-row data has expected
quality
\[
  \rho+(1-\rho)\frac{|J|}{m}
  =
  \Qpr(|J|/m,\rho)
  =
  \Qpr(\asub,\rho)-O(1/m).
\]
The procedure is polynomial time when the cosets and their
defining equations can be obtained in polynomial time.
\end{proposition}

\begin{proof}
Each selected row contributes
\[
  \psi_{i,j}\bigl((Bx)_i\bigr)=\psi_{i,j}(a_i),
  \qquad j=1,\ldots,c.
\]
These are linear equations over $\F_p$.
The image of $x\mapsto Bx$ has prime-field dimension
$b\alpha m$, and the chosen number of equations satisfies
$c|J|\le b\alpha m$.
Full row rank makes the linear map onto its equation space
surjective, so every right-hand side is consistent.
Elimination therefore returns a solution satisfying all
rows in $J$.

For the other rows, condition first on $B$, the independent
randomness used to select $J$, and all the rowwise
translation-invariant data used in that selection.
Under this conditioning, the accepted sets remain independent,
and their conditional laws remain translation invariant.
Now condition on the selected rows' accepted sets and on the
randomness used to complete the linear system.
These determine $x$, while an unselected $F_i$ retains its
translation-invariant law.

For any fixed $v\in\F_q$, translation invariance and
$|F_i|=\rho q$ imply
\[
  \Pr[v\in F_i]=\rho.
\]
Thus each unselected row is satisfied with conditional
probability $\rho$. Summing the contributions gives
\[
  \mathbb E[Q(x)]
  =
  \frac{|J|+(m-|J|)\rho}{m}.
\]
A uniform draw from the affine solution space is one
valid completion rule, but uniformity is not needed
for this expectation.
\end{proof}

The same argument permits row-dependent codimensions,
provided the selected equations have full row rank.
When the main statement uses the largest codimension $c$,
a coset of smaller codimension can be replaced by a
contained coset of codimension $c$.
The rank condition must still be checked for the chosen
equations.

The selection restriction allows inspection of
translation-invariant features, such as the dimensions of
contained affine subspaces. It does not allow the completion
to use the locations of accepted values in unselected rows.
That restriction is what preserves the baseline expectation.

\begin{proof}[Proof of rule \gate{4}]
Under the hypotheses of Proposition~\ref{prop:remaining},
Proposition~\ref{prop:subfield} supplies a classical score
\[
  \rho+(1-\rho)\frac{\lfloor\asub m\rfloor}{m}
  =
  \Qpr(\asub,\rho)-O(1/m).
\]
Therefore
\[
  \mu^+
  \le
  \Qdqi(\lambda,\rho)-\Qpr(\asub,\rho)+O(1/m).
\]

At $\rho=1/2$, the Singleton restriction gives
\[
  \Qdqi(\lambda,1/2)
  <
  \frac12+\frac12\sqrt{2\alpha-\alpha^2},
\]
while the subfield baseline is
\[
  \frac12+\frac12\min(1,b\alpha/c).
\]
If $b\alpha/c\ge1$, the baseline is one.
Otherwise, it is at least the Singleton expression whenever
\[
  \left(\frac{b\alpha}{c}\right)^2
  \ge
  2\alpha-\alpha^2,
\]
or equivalently
\[
  \alpha\ge\frac{2c^2}{b^2+c^2}.
\]
Since $c\le b$ implies
\[
  \frac{2c^2}{b^2+c^2}\le\frac cb,
\]
this threshold also covers the case where the baseline
has saturated at one. The rounding contributes only
$O(1/m)$.
\end{proof}

For hyperplane accepted sets over $\F_{16}$,
$b=4$ and $c=1$, so the threshold is $\alpha\ge2/17$.
At $\alpha=0.12$, the subfield baseline is $0.74$,
whereas the Singleton-reach expression is approximately
$0.7375$.

\paragraph{Finding contained affine subspaces.}
The accepted-set density alone does not determine the
subfield baseline, but it can sometimes guarantee a useful
contained subspace. For example, every eight-element subset
of $\F_2^5$ contains an affine plane.

Here is a short proof. If an eight-element set contained
no affine plane, translate one of its elements to zero
and place the other seven elements as columns of a
$5\times7$ binary matrix. Its kernel has dimension at least
two. The absence of an affine plane implies that every
nonzero kernel vector has weight at least five:
a relation of weight three gives a plane together with zero,
and a relation of weight four gives a plane directly.
But two distinct binary vectors of length seven and weight
at least five have a nonzero difference of weight at most
four, a contradiction.
The bound is sharp: the seven points
$0,e_1,\ldots,e_5,e_1+\cdots+e_5$ contain no affine plane.
Thus a half-density set in $\F_{32}$ always contains a
two-dimensional coset and has subfield codimension at most
three.

For random accepted sets, the expected number of contained
affine $k$-spaces has a simple exact expression.
If $r=\rho p^b$, it is
\[
  \binom bk_p\,p^{b-k}
  \frac{\binom{p^b-p^k}{r-p^k}}{\binom{p^b}{r}},
\]
where $\binom bk_p$ is the Gaussian binomial coefficient,
and the expression is zero when $r<p^k$.
The first factor counts linear $k$-spaces, the second their
cosets, and the ratio is the probability that all points
of one coset belong to the accepted set.

At half density, these expectations are approximately
$62.8$ and $0.76$ for dimensions two and three in $\F_{32}$.
In $\F_{256}$ they are approximately $29.8$ and
$2\times10^{-5}$ for dimensions four and five.
Markov's inequality therefore bounds the probability of
a contained five-dimensional coset by $2\times10^{-5}$.
An expectation above one, however, does not establish
that a four-dimensional coset is present with high
probability.

\paragraph{Checking the rank condition.}
Contained cosets are useful only if their equations supply
enough independent constraints. Algebraic structure can
prevent this. When the same linear form is applied to
many polynomial evaluations, the resulting equations can
inherit dependencies from the Frobenius map.

For example, in a Reed--Solomon calculation over
$\F_{256}$, $B$ is the transpose of a parity-check matrix and has $240$ rows and
$28$ columns. Applying the same last-coordinate form to
every row gives rank $109$, rather than the available
prime-field dimension $8\cdot28=224$.
The value $109$ counts the exponents in the binary
cyclotomic cosets modulo $255$ that meet
$\{0,\ldots,27\}$, and the three sampled restrictions to
$240$ evaluation points retained this rank.
With random nonzero row multipliers, the tested systems
had rank $224$.

Row multipliers or varying linear forms can remove particular
dependencies, but they do not replace a rank check for the
selected system.

\subsection{The objective alphabet and the general-objective law}
\label{sec:objlawproofs}

We now prove Lemma~\ref{lem:objlaw} and
Proposition~\ref{prop:alphabet}.
The calculation follows the Fourier-shell construction of
Jordan et al.~\cite[\S8.2]{jordan24dqi}.
Its expected objective can be computed
without assuming that the row functions take only two values.

Throughout this subsection, $q$ is prime and the Fourier
transform is unitary:
\[
  \hat h(a)
  =
  \frac1{\sqrt q}
  \sum_{v\in\F_q}h(v)e^{2\pi iav/q}.
\]

\subsubsection{The finite spectral identity}

\begin{lemma}[DQI shells for a real objective]
\label{lem:cosinelaw}
Let $F:\F_q\to\mathbb R$ have uniform mean $\bar F$,
standard deviation $\sigma>0$, and skewness $\gamma$.
Write
\[
  S=\{a\ne0:\hat F(a)\ne0\},
  \qquad
  f_i(v)=F(v-c_i),
  \qquad
  f(x)=\sum_{i=1}^m f_i(b_i\cdot x),
\]
where the $b_i$ are the rows of $B$.

Let $0\le\ell\le m$, and define
\[
  \mathcal Y_\ell(S)
  =
  \{y\in(S\cup\{0\})^m:|y|\le\ell\}.
\]
Assume the separation condition
\[
  B^{\mathsf T}y'
  \ne
  B^{\mathsf T}(y+ae_i)
\]
whenever $y,y'\in\mathcal Y_\ell(S)$, $a\in S$,
and $y'\ne y+ae_i$.
For $S=\F_q^*$, this condition is equivalent to
$2\ell+1<d$.

The DQI shell state with real weights
$w=(w_0,\ldots,w_\ell)$ and $\|w\|_2=1$ has expected objective
\[
  \langle f\rangle
  =
  m\bar F+\sigma w^{\mathsf T}A^{(m,\ell,\gamma)}w,
\]
where
\[
  A_{kk}=k\gamma,
  \qquad
  A_{k-1,k}=A_{k,k-1}=\sqrt{k(m-k+1)}
  \quad (1\le k\le\ell).
\]
Its optimal expected objective is therefore
$m\bar F+\sigma\lambda_{\max}(A^{(m,\ell,\gamma)})$.

If $\ell/m\to\lambda\in(0,1/2]$, then
\[
  \max_w\frac{\langle f\rangle}{m}
  \longrightarrow
  \bar F+
  \sigma\max_{0\le t\le\lambda}
  \left(\gamma t+2\sqrt{t(1-t)}\right).
\]
The maximum occurs at $t=\lambda$ when
\[
  \gamma\ge
  -\frac{1-2\lambda}{\sqrt{\lambda(1-\lambda)}}.
\]
If the field or objective varies with $m$, the same conclusion
holds with the limiting moments, provided
$\bar F_m$, $\sigma_m$, and $\gamma_m$ converge to finite
limits.
\end{lemma}

\begin{proof}
We first normalize the shell states, then compute the
matrix of the objective on their span.

\emph{Normalization and Fourier support.}
Put
\[
  g_i(v)=\frac{f_i(v)-\bar F}{\sqrt q\,\sigma}.
\]
These functions satisfy
\[
  \hat g_i(0)=0,
  \qquad
  \sum_v g_i(v)^2
  =
  \sum_a|\hat g_i(a)|^2
  =
  1.
\]
Let $\mathcal E_k$ denote the elementary symmetric
polynomial of degree $k$, and define
\[
  |\Phi_k\rangle
  =
  \frac1{\sqrt{q^{n-k}\binom mk}}
  \sum_{x\in\F_q^n}
  \mathcal E_k
  \bigl(g_1(b_1\cdot x),\ldots,g_m(b_m\cdot x)\bigr)
  |x\rangle.
\]
The state under consideration is
$|\Psi_w\rangle=\sum_{k=0}^{\ell}w_k|\Phi_k\rangle$.
For a general real objective, this definition uses the
elementary-symmetric shells directly and does not require
their sum to be a univariate polynomial in $f(x)$.

Fourier transformation gives
\[
  |k\rangle
  :=
  \widehat{|\Phi_k\rangle}
  =
  \binom mk^{-1/2}
  \sum_{|y|=k}\hat g(y)\,|B^{\mathsf T}y\rangle,
  \qquad
  \hat g(y)=\prod_{i:y_i\ne0}\hat g_i(y_i).
\]
This calculation uses only Fourier inversion and
$\hat g_i(0)=0$, as in the original DQI construction.
The coefficient $\hat g(y)$ vanishes unless every nonzero
entry of $y$ belongs to $S$.

The separation condition implies injectivity of the syndrome
map on $\mathcal Y_\ell(S)$.
Indeed, if distinct $y,y'$ had equal syndromes, interchange
them if necessary so that $y'\ne0$, choose
$i$ with $y'_i\ne0$, and put
$z=y'-y'_ie_i$.
The condition applied to $z$, $y$, and $a=y'_i$ would
exclude the collision.

For a fixed support of size $k$, summing
$|\hat g(y)|^2$ over its nonzero entries gives one.
There are $\binom mk$ such supports, so
\[
  \sum_{|y|=k}|\hat g(y)|^2=\binom mk.
\]
Injectivity now shows that the $|k\rangle$ are orthonormal.
Thus $\|w\|_2=1$ normalizes $|\Psi_w\rangle$.

\emph{Compression of the objective.}
We have
\[
  f(x)=m\bar F+\sqrt q\,\sigma\sum_i g_i(b_i\cdot x).
\]
In the Fourier basis, multiplication by
$\sqrt q\,g_i(b_i\cdot x)$ becomes
\[
  \sum_{a\in S}\hat g_i(a)T_{i,a},
  \qquad
  T_{i,a}|z\rangle=|z+a b_i\rangle.
\]
In particular,
\[
  T_{i,a}|B^{\mathsf T}y\rangle
  =
  |B^{\mathsf T}(y+ae_i)\rangle.
\]
The separation assumption gives
\[
  \langle B^{\mathsf T}y'|
  B^{\mathsf T}(y+ae_i)\rangle
  =
  \mathbf1[y'=y+ae_i]
\]
for supported $y,y'$.
A shift changes at most one coordinate, so the compressed
operator couples only shells of weights $k-1$, $k$, and
$k+1$.
The shell span need not be invariant under the full operator.
Its compression is sufficient to compute the expectation.

If $y_i=0$, the shift raises the weight.
Using
$\hat g(y+ae_i)=\hat g_i(a)\hat g(y)$, its matrix entry is
\[
  \begin{aligned}
  A_{k+1,k}
  &=
  \frac{
    \displaystyle
    \sum_{|y|=k}|\hat g(y)|^2
    \sum_{i:y_i=0}\sum_a|\hat g_i(a)|^2
  }{
    \sqrt{\binom mk\binom m{k+1}}
  }\\
  &=
  \frac{\binom mk(m-k)}
       {\sqrt{\binom mk\binom m{k+1}}}
  =
  \sqrt{(k+1)(m-k)}.
  \end{aligned}
\]
The lowering entry is its adjoint.

For the diagonal entry, the shifted coordinate must be
nonzero both before and after the shift.
Its contribution is
\[
  \sum_{u,u'\in\F_q}
  \overline{\hat g_i(u)}
  \hat g_i(u-u')\hat g_i(u').
\]
Terms with $u=0$, $u'=0$, or $u=u'$ vanish.
The convolution theorem and Parseval's identity give
\[
  \sum_{u,u'}
  \overline{\hat g_i(u)}
  \hat g_i(u-u')\hat g_i(u')
  =
  \sqrt q\sum_v g_i(v)^3
  =
  \gamma.
\]
For each support of size $k$, there are $k$ possible
coordinates contributing this value.
The other coordinates contribute unit norm, so
$A_{kk}=k\gamma$.

We have therefore obtained
\[
  \langle f\rangle
  =
  m\bar F+\sigma w^{\mathsf T}Aw.
\]
Since $A$ is real symmetric, optimizing over real unit
vectors $w$ gives its largest eigenvalue.

\emph{The limiting eigenvalue.}
Define
\[
  G_\gamma(t)=\gamma t+2\sqrt{t(1-t)}.
\]
The row-sum bound for a symmetric matrix with nonnegative
off-diagonal entries gives
\[
  \frac{\lambda_{\max}(A)}m
  \le
  \max_{0\le t\le\ell/m}G_\gamma(t)+o(1).
\]
To see the uniform error, divide each off-diagonal entry
by $m$ and compare it with $\sqrt{t(1-t)}$ at $t=k/m$.
The discrepancy is $O(m^{-1/2})$, uniformly in $k$.
At the endpoints, omitting a missing off-diagonal entry
only improves the upper bound.

For the lower bound, fix $t\in(0,\lambda)$ and choose
integers $\omega_m\to\infty$ with $\omega_m/m\to0$.
Use the unit vector that is constant on $\omega_m$
consecutive indices beginning at $\lfloor tm\rfloor$
and zero elsewhere.
Its Rayleigh quotient, divided by $m$, tends to
\[
  \gamma t+2\sqrt{t(1-t)}.
\]
Indeed, the diagonal entries in this window are
$m\gamma t+o(m)$, the off-diagonal entries are
$m\sqrt{t(1-t)}+o(m)$, and the missing endpoint
couplings contribute a vanishing fraction.
Together with $\lambda_{\max}(A)\ge A_{00}=0$ and
continuity in $t$, this proves
\[
  \frac{\lambda_{\max}(A)}m
  \longrightarrow
  \max_{0\le t\le\lambda}G_\gamma(t).
\]

The derivative
\[
  G_\gamma'(t)
  =
  \gamma+\frac{1-2t}{\sqrt{t(1-t)}}
\]
is decreasing.
Thus the maximum over $[0,\lambda]$ is at $\lambda$
when $G_\gamma'(\lambda)\ge0$, giving the stated condition.

Finally, suppose the moments vary with $m$.
The skewness changes only the diagonal:
\[
  A^{(m,\ell,\gamma_m)}
  -
  A^{(m,\ell,\gamma)}
  =
  (\gamma_m-\gamma)\operatorname{diag}(0,1,\ldots,\ell).
\]
Its operator norm is at most
$|\gamma_m-\gamma|m$.
Hence the scaled largest eigenvalue has the same limit.
Combining this with convergence of $\bar F_m$ and
$\sigma_m$ proves the assertion with varying moments,
including the case where the limiting standard deviation
is zero.
\end{proof}

\begin{proof}[Proof of Lemma~\ref{lem:objlaw}]
Lemma~\ref{lem:cosinelaw} gives the finite identity and
its asymptotic limit.

For the classical baseline, an information set contains
$\alpha m$ coordinates that can be assigned independently.
Set each one to a maximizer of its row objective.
Those rows contribute $\alpha m\max F$.
For every other row, an independent uniform center makes
the expected objective equal to $\bar F$, as proved in
Lemma~\ref{lem:centered}.
The expected baseline per row is therefore
\[
  \alpha\max F+(1-\alpha)\bar F.
\]
\end{proof}

\subsubsection{Indicator and cosine objectives}

\begin{corollary}\label{cor:cosine}
Suppose $q$ is prime and $F$ is the indicator of a set
of size $r$, where $0<r<q$.
Then
\[
  \bar F=\frac rq,
  \qquad
  \sigma=\frac{\sqrt{r(q-r)}}q,
  \qquad
  \gamma=\frac{q-2r}{\sqrt{r(q-r)}},
  \qquad
  S=\F_q^*.
\]
Writing $\rho=r/q$,
\[
  \bar F+\sigma G_\gamma(t)
  =
  \left(
    \sqrt{t(1-\rho)}
    +
    \sqrt{\rho(1-t)}
  \right)^2.
\]
Its maximum over the licensed interval $0\le t\le\lambda$
is the quality law \eqref{eq:qdqi}.

If instead $F(v)=\cos(2\pi v/q)$ with $q>3$, then
\[
  \bar F=0,
  \qquad
  \sigma=\frac1{\sqrt2},
  \qquad
  \gamma=0,
  \qquad
  S=\{\pm1\}.
\]
Under the separation hypothesis and with optimal weights,
the expected mean cosine tends to
\[
  \sqrt{2\lambda(1-\lambda)}.
\]
\end{corollary}

\begin{proof}
The indicator moments follow directly from the Bernoulli
distribution of $F(v)$ under uniform $v$.
For any nonzero frequency, its Fourier coefficient is a
sum of $r$ distinct $q$-th roots of unity.
When $q$ is prime, such a sum cannot vanish for
$0<r<q$: otherwise its polynomial would be a multiple
of the cyclotomic polynomial
$1+X+\cdots+X^{q-1}$.
Thus the nonzero Fourier support is all of $\F_q^*$.
Substituting the moments gives the displayed quality
function. Its maximum occurs at $t=1-\rho$ if that point
belongs to the interval, and at $t=\lambda$ otherwise.

For the cosine, its two Fourier frequencies are $\pm1$.
The identities
\[
  \sum_{v\in\F_q}\cos^2(2\pi v/q)=\frac q2,
  \qquad
  \sum_{v\in\F_q}\cos^3(2\pi v/q)=0
\]
hold for $q>3$.
The second follows from
$\cos^3\theta=(3\cos\theta+\cos3\theta)/4$.
Lemma~\ref{lem:cosinelaw} then gives the stated limit.
\end{proof}

\subsubsection{The cosine screening conditions}

\begin{proof}[Proof of Proposition~\ref{prop:alphabet}(i)]
A decoder for alphabet $S$ must distinguish at least the
$\binom m\ell|S|^\ell$ supported patterns of weight
exactly $\ell$.
There are only $q^{\alpha m}$ syndromes, so
\[
  \binom m\ell|S|^\ell\le q^{\alpha m}.
\]
With $\lambda=\ell/m$, taking logarithms gives
\[
  H_2(\lambda)+\lambda\log_2|S|
  \le
  \alpha\log_2q+O(\log m/m).
\]
For the cosine, $|S|=2$, yielding
\[
  \alpha
  \ge
  \frac{H_2(\lambda)+\lambda}{\log_2q}-o(1).
\]
This is a necessary counting condition. The additional
separation condition is still needed for the spectral law.
\end{proof}

\begin{proof}[Proof of Proposition~\ref{prop:alphabet}(ii)]
Corollary~\ref{cor:cosine} gives the quality law
$\sqrt{2\lambda(1-\lambda)}$.
The Prange baseline is $\alpha$, since $\max F=1$
and $\bar F=0$.
For $0\le\lambda\le1/2$, a strict advantage is equivalent to
\[
  \lambda(1-\lambda)>\frac{\alpha^2}{2}.
\]
There is no solution when $\alpha\ge1/\sqrt2$.
Otherwise the inequality holds exactly when
\[
  \lambda>
  \frac{1-\sqrt{1-2\alpha^2}}2.
\]
The cosine values lie in $[-1,1]$, so the restart
comparison may use the range bound $R=2$.
\end{proof}

\begin{proof}[Proof of Proposition~\ref{prop:alphabet}(iii)]
Let $U_q$ denote the uniform distribution on $\F_q$,
and define
\[
  \Lambda_q(\vartheta)
  =
  \ln\mathbb E_{U_q}[e^{\vartheta F(v)}],
  \qquad
  \Lambda_q^*(Q)
  =
  \sup_{\vartheta\ge0}
  \{\vartheta Q-\Lambda_q(\vartheta)\}.
\]
For a fixed codeword and independent uniform centers,
the row values have the same distribution as independent
copies of $F(v)$ with $v\sim U_q$.
The Chernoff bound therefore gives the finite inequality
\[
  \mathbb E[N_Q]
  \le
  \exp\!\left(
    m\bigl(\alpha\ln q-\Lambda_q^*(Q)\bigr)
  \right).
\]
This inequality holds at each field size.
At fixed $q$, the corresponding exponential rate is also
the one given by Cram\'er's theorem~\cite{dembo98}.

For $\bar F<Q<\max F$, the Gibbs variational principle gives
\[
  \Lambda_q^*(Q)
  =
  \min_{\nu:\,\mathbb E_\nu F=Q}
  \KL{\nu}{U_q}.
\]
The minimizer is
\[
  \nu_\vartheta(v)
  =
  \frac{e^{\vartheta F(v)}}
       {\sum_{u\in\F_q}e^{\vartheta F(u)}}.
\]
Consequently, the first-moment line is
\[
  \Qstar
  =
  \max_{\nu:\,\KL{\nu}{U_q}\le\alpha\ln q}
  \mathbb E_\nu F.
\]

For the cosine, the mean is zero and its maximum one is
attained only at $v=0$. A distribution of mean one must
therefore be concentrated at zero, giving
\[
  \Lambda_q^*(1)=\ln q.
\]
Since $\alpha<1$, the entropy budget $\alpha\ln q$ is
strictly smaller, so $\Qstar<1$.
For $0<\alpha<1$, the boundary is attained by the unique
positive tilt satisfying
\[
  \KL{\nu_\vartheta}{U_q}
  =
  \alpha\ln q.
\]
This proves the stated characterization.
\end{proof}

The same bound validates the encoder comparison for the
cosine objective. At fixed $q$, every score strictly above
$\Qstar$ has exponentially small expected count, so
Markov's inequality bounds the optimum with high probability.
The conclusion also remains uniform when $q$ varies:
since $F\in[-1,1]$,
\[
  \Lambda_q''(\vartheta)
  =
  \operatorname{Var}_{\nu_\vartheta}(F)
  \le1.
\]
Thus the upper-tail rate function has curvature at least
one in its interior, and
\[
  \Lambda_q^*(\Qstar+\varepsilon)
  -
  \Lambda_q^*(\Qstar)
  \ge\frac{\varepsilon^2}{2}
\]
whenever $\Qstar+\varepsilon\le1$.
No codeword exceeds the line by a fixed positive amount
with nonvanishing probability.
An efficient encoder attaining that line would therefore
remove the proposed limiting quality advantage.
The comparison is therefore sound, but this does
not establish the existence of such an encoder.

For numerical evaluation, the finite-field tilt is obtained
from the sums in $\Lambda_q$.
At fixed $\vartheta$, their large-$q$ limits are the Bessel
expressions
\[
  \mathbb E_{\nu_\vartheta}[F]
  \longrightarrow
  \frac{I_1(\vartheta)}{I_0(\vartheta)},
  \qquad
  \KL{\nu_\vartheta}{U_q}
  \longrightarrow
  \vartheta\frac{I_1(\vartheta)}{I_0(\vartheta)}
  -\ln I_0(\vartheta).
\]
At $\alpha=0.07$ and $q=8191$, the entropy budget is
approximately $0.6308$.
The finite sums give $\vartheta\simeq2.186$ and
$\Qstar\simeq0.726$, agreeing with the Bessel approximation
at the displayed precision.
The corresponding values at $q=2^{17}-1$ and
$q=2^{31}-1$ are approximately $0.803$ and $0.945$.
For fixed $\alpha>0$, the line remains below one at
every finite $q$ and approaches one as $q$ grows.

The spectral identity should be distinguished from efficient
preparation of its state. An implementation needs to prepare
the local Fourier states, load the shell weights, and run
the decoder coherently. For the cosine, the local Fourier
state has only two nonzero amplitudes.
For indicator objectives, preparation depends on how the
accepted sets are supplied. Explicit tables with polynomial
alphabet size and suitable oracle models are treated in
the original DQI construction~\cite{jordan24dqi}.

The small primes require separate formulas.
At $q=3$, the cosine has skewness $1/\sqrt2$, and its
quality law on $0\le\lambda\le1/2$ is
\[
  \frac{\lambda}{2}
  +
  \sqrt{2\lambda(1-\lambda)}.
\]
At $q=2$, the cosine is an affine rescaling of a singleton
indicator, so the indicator analysis applies.

\subsection{The finite-length comparison with Prange restarts}
\label{sec:converseproof}

Theorem~\ref{thm:converse} compares DQI's exact expected
score with the best score among a prescribed number of
classical restarts.
Once the finite DQI value is known, the proof is a
concentration argument for the classical baseline.

The important restriction is what each restart may inspect.
It may use the accepted sets, or centers, on its own
information set. Its choice of information set must not
use those data, and its completion must not use data
from the remaining rows.

\subsubsection{One restart for accepted sets}
\label{sec:isbaseline}

Write $r=\rank B=\alpha m$ and choose an information set
$I$ of the optimizer code, with $|I|=r$.
Restriction to $I$ is a bijection from the optimizer code
to $\F_q^I$.
Consequently, choosing the values $y_i$ for $i\in I$
determines a unique codeword $y\in\optcode$, from which
an assignment $x$ with $Bx=y$ can be recovered.

A restart chooses $I$ using only $B$ and randomness
independent of the accepted sets.
It then chooses $y_i\in F_i$ for $i\in I$, using only
the sets on $I$ and independent auxiliary randomness,
and completes the codeword.
The choice of accepted values may depend on the selected
sets, but the random seed and the selection rule must not
reveal information about unselected sets.

\begin{lemma}\label{lem:isscore}
Suppose the accepted sets are independent of each other
and of $B$, and each is uniform among subsets of size
$\rho q$.
Conditionally on $B$, $I$, and the chosen values on $I$,
the number $X$ of satisfied rows outside $I$ has distribution
\[
  X\sim\mathrm{Bin}(m-r,\rho).
\]
The total score is $W=r+X$, and
\[
  \frac{\mathbb E[W]}m
  =
  \alpha+(1-\alpha)\rho
  =
  \Qpr.
\]
\end{lemma}

\begin{proof}
The selected coordinates contribute exactly $r$.
Conditioning on their chosen values fixes the completed
codeword, while the accepted sets outside $I$ remain
independent and uniform.
Each fixed unselected coordinate therefore belongs to
its accepted set independently with probability $\rho$.
\end{proof}

\subsubsection{One restart for a bounded real objective}
\label{sec:converserange}

For $f_i(v)=F(v-c_i)$, use an information set chosen
under the same restriction.
Set
\[
  (Bx)_i-c_i\in\arg\max F,
  \qquad i\in I,
\]
and complete the codeword.
Any tie-breaking uses only the selected centers and
independent randomness.
For the cosine, the selected coordinates are simply
$(Bx)_i=c_i$.

\begin{lemma}[the unselected coordinates]
\label{lem:centered}
Suppose the centers are independent, uniform on $\F_q$,
and independent of $B$.
For one restart, condition on $B$, $I$, the centers on $I$,
and the randomness used to choose its completion.
Then the values
\[
  (Bx)_i-c_i,
  \qquad i\notin I,
\]
are independent and uniform on $\F_q$.
Their row objectives have mean $\bar F$ and lie in an
interval of length $R$, where $R$ is any valid range
bound for $F$.
\end{lemma}

\begin{proof}
The conditioned data determine the completed codeword.
Every unselected center remains independent and uniform,
so subtracting it from the fixed value $(Bx)_i$ preserves
uniformity and independence.

The resulting product distribution does not depend on the
conditioned values. The same conclusion therefore holds
after averaging over those values.
\end{proof}

A restart consequently has expected score
\[
  \mathbb E[W]
  =
  r\max F+(m-r)\bar F
  =
  \Qpr m.
\]
For the cosine, $\bar F=0$, $\max F=1$, and one may
take $R=2$. For accepted-set objectives, the analogous
parameters are $\rho$, $1$, and $R=1$.

\begin{lemma}[concentration for one restart]
\label{lem:rangescore}
Under the hypotheses of Lemma~\ref{lem:centered},
suppose $r<m$ and $R>0$.
For every $t>0$,
\[
  \Pr[W\ge\Qpr m+t]
  \le
  \exp\!\left(
    -\frac{2t^2}{(1-\alpha)mR^2}
  \right).
\]
\end{lemma}

\begin{proof}
The selected rows contribute the deterministic score
$r\max F$.
The remaining $m-r$ row objectives are independent,
each has mean $\bar F$, and each lies in an interval
of length $R$.
Hoeffding's inequality therefore applies with sum of
squared ranges $(m-r)R^2$.
The conditional bound is independent of the selected data,
so it also holds unconditionally.
\end{proof}

\begin{proof}[Proof of Theorem~\ref{thm:converse}]
Let $W_1,\ldots,W_N$ be the restart scores, with each
restart satisfying the information restrictions above.
Put
\[
  \mu_m=\Qdqi^{(m)}-\Qpr>0.
\]
For each restart, Lemma~\ref{lem:rangescore} with
$t=\mu_m m$ gives
\[
  \Pr[W_j\ge\Qdqi^{(m)}m]
  \le
  \exp\!\left(
    -\frac{2\mu_m^2m}{(1-\alpha)R^2}
  \right).
\]
For accepted-set objectives, the same inequality follows
from Lemma~\ref{lem:isscore} and Hoeffding's inequality,
with $R=1$.

A union bound now yields
\[
  \Pr\!\left[
    \max_{1\le j\le N}W_j\ge\Qdqi^{(m)}m
  \right]
  \le
  N\exp\!\left(
    -\frac{2\mu_m^2m}{(1-\alpha)R^2}
  \right).
\]
The right-hand side is at most $\eta$ whenever
\[
  m\ge
  \frac{(1-\alpha)R^2\ln(N/\eta)}{2\mu_m^2}.
\]
With probability at least $1-\eta$, every restart
therefore scores strictly below DQI's exact expected score.
\end{proof}

Independence between restart scores is not needed for
the union bound, and they may be correlated through the common
instance. What matters is that each restart separately
satisfies the information restrictions used in its tail
bound. A procedure that uses earlier scores to adapt its
information sets or that optimizes using all row data
requires a different analysis.

The comparison is with DQI's expectation.
The probability of obtaining a lower target score from
an individual DQI sample is addressed separately by
\eqref{eq:reversemarkov}.

%% file: supplement/S2-candidates.tex
\section{The two candidates}\label{sec:supp-candidates}

This section supplies the construction details, decoding proofs, and classical-attack
analysis behind Section~\ref{sec:problems}. For alternant max-agreement (AMA), the central question is
how access to the hidden code structure affects the optimization problem. For multiplicative
polynomial intersection (MPI), it is how rational reconstruction provides the restricted-error
decoder required by the cosine objective. We treat these questions separately from the
experimental comparisons: the decoding guarantees are mathematical statements, whereas the
classical scores describe the algorithms and budgets tested.

\subsection{Alternant max-agreement}\label{sec:mcattacks}

\subsubsection{Construction and parameter choices}\label{sec:mcconstruction}

Let $K=\F_{p^{s}}$, and choose distinct points $a_1,\ldots,a_m\in K$ and nonzero
multipliers $v_1,\ldots,v_m\in K$. Consider the matrix
\[
  H_{ji}=v_i a_i^{j},
  \qquad 0\le j<r,\quad 1\le i\le m.
\]
Its kernel over $K$ is a generalized Reed--Solomon code with minimum distance $r+1$.
The alternant code consists of the vectors in this kernel whose coordinates lie in the
base field:
\[
  \deccode=\{y\in\F_p^{m}:Hy=0\}.
\]
Expanding the entries of $H$ in an $\F_p$-basis of $K$ gives a parity-check matrix over
$\F_p$. Scrambling its rows and permuting its columns gives the public matrix $\Hpub$,
and the optimization instance uses $B=\Hpub^{\mathsf T}$.

The expanded parity-check matrix has at most $sr$ independent rows. We verify that its
rank is $sr$ for every reported instance, so the optimizer rate is
\[
  \alpha=\frac{sr}{m},
\]
and the decoder code has minimum distance at least $r+1$. Access to the support and
multipliers permits standard alternant decoding of every error pattern of weight at most
$\lfloor r/2\rfloor$.

The license used by the DQI quality law is slightly stronger than ordinary unique
decoding. Definition~\ref{def:licensed} requires $2\ell+1<d$, so the designed-distance
bound certifies
\[
  \ell=\left\lfloor\frac{r-1}{2}\right\rfloor,
  \qquad
  \frac{\ell}{m}=\frac{\alpha}{2s}+O(m^{-1}).
\]
In particular, at $r=187$ we use $\ell=93$ and $d\ge188$. The algebraic decoding
guarantee applies to every pattern within this radius, and the random-pattern checks reported
below test the implementation.

These expressions explain the choice of extension degree. Substituting
$\lambda=\alpha/(2s)$ into the asymptotic quality law at density $\rho=\tfrac12$ gives
\[
  \mu(\alpha,s)
  =
  \sqrt{\frac{\alpha}{2s}
             \left(1-\frac{\alpha}{2s}\right)}
  -\frac{\alpha}{2}.
\]
For fixed $s$, this expression is maximized at
\[
  \alpha^{*}
  =
  s\left(1-\frac{s}{\sqrt{1+s^{2}}}\right),
  \qquad
  \mu(\alpha^{*},s)
  =
  \frac{\sqrt{1+s^{2}}-s}{2}.
\]
At $s=1$ this recovers the sharp margin cap of Proposition~\ref{prop:cap}.
Increasing $s$ reduces the available margin and raises the cost predicted by the
distinguisher model discussed below. We choose $s=6$ and $\alpha\approx0.08$, near
$\alpha^{*}\approx0.0816$. The actual accepted-set density is $\rho=5/11$, and the
reported DQI values use the exact finite-length law at that density.

The base field also controls the regime and the available lengths. For independently
sampled accepted sets, the first-moment unsatisfiable regime is
$p^{\alpha}\rho<1$. Meanwhile, the distinct evaluation points require $m\le p^{s}$.
Our choice $p=11$, $s=6$ allows $m\le1{,}771{,}561$. Thus, at fixed field parameters,
the construction supplies a finite range of instances. The asymptotic expressions above
guide parameter selection, and the comparisons at the reported lengths use their finite
parameters, as in Table~\ref{tab:candrules}.

\subsubsection{What the secret key gives a classical solver}\label{sec:mckeyholder}

The key exposes the optimizer code as a trace code. After undoing the secret coordinate
permutation, its codewords have the form
\[
  \bigl(\Tr(v_i f(a_i))\bigr)_{i=1}^{m},
  \qquad f\in K[z],\quad \deg f<r,
\]
where
\[
  \Tr(z)=z+z^{p}+\cdots+z^{p^{s-1}}
\]
is the trace from $K$ to $\F_p$. Consequently, a classical key holder seeks a polynomial
$f$ satisfying
\[
  \Tr(v_i f(a_i))\in F_i
\]
on as many rows as possible. This is the trace-structured polynomial-intersection
problem described in Section~\ref{sec:mckey}.

One possible use of this representation is Reed--Solomon list recovery. Each accepted
set $F_i$ lifts to the list
\[
  L_i=\{y\in K:\Tr(v_i y)\in F_i\}.
\]
Since $v_i\ne0$ and the trace is surjective, every trace value has $p^{s-1}$ preimages.
Thus $|L_i|=\rho p^{s}$. The following standard guarantee explains why the direct
interpolation approach does not certify a useful solution at our parameters.

\begin{lemma}[Reed--Solomon list-recovery threshold]\label{lem:listrecovery}
Consider a Reed--Solomon code of length $m$ and dimension $k\ge2$, with an input list
of size at most $L$ at each evaluation point. With suitable interpolation parameters,
Guruswami--Sudan list recovery returns every codeword agreeing with at least $t$ lists
whenever
\[
  t>\sqrt{(k-1)mL}.
\]
This sufficient condition is vacuous when $(k-1)L\ge m$.
\end{lemma}

\begin{proof}
There are at most $mL$ interpolation pairs. Assigning a sufficiently large common
multiplicity to these pairs gives the Guruswami--Sudan threshold
$t>\sqrt{(k-1)mL}$~\cite{guruswami99gs}. Equivalently, the strict inequality allows an
interpolation polynomial whose $(1,k-1)$-weighted degree is smaller than the total
multiplicity contributed by $t$ agreements. Substituting any degree-$(k-1)$ message
polynomial with that many agreements therefore gives the zero polynomial, so the
message is recovered by the factorization step. Since $t\le m$, the stated condition
cannot hold when $(k-1)L\ge m$.
\end{proof}

For the main AMA instance, $k=r=187$ and
\[
  L=\frac{5}{11}\,11^{6}=805{,}255.
\]
The threshold is therefore
\[
  t>\sqrt{186\cdot14000\cdot805255}
    \approx1.4\times10^{6},
\]
well above the number of rows. This calculation concerns the sufficient guarantee of
this list-recovery algorithm. It does not exclude other list-recovery methods or other
classical uses of the key.

We also tested a direct search in the polynomial representation. At $s=6$ and
$m=14000$, line search and simulated annealing along $K$-lines of message polynomials
were each given $600$~s, starting from the best assignment found by the portfolio
without the secret key. Each line contains $11^{6}$ points, and the implementation scores $4096$ randomly
chosen points per line, completing about $180$ moves per run. Across the three seeds,
neither method improved the starting qualities $0.5144$, $0.5151$, and $0.5171$.
Neither method improves these assignments when run for $10^4$ moves either. By comparison,
line search without the key, given the same starting assignments and the same number of
moves, improves two of the three. Each move in the polynomial representation preserves
only $r-1$ row values, compared with $n-1$ for a move of the search without the key.
This helps explain why the former search struggles to improve an already good assignment
(Section~\ref{sec:scalingkey}). At $m=20000$ the same searches ran for $1800$~s each, with $109$ to $168$ moves per run. They again improved none of the starting qualities $0.5116$, $0.5092$ and $0.5128$. Started from random polynomials, the searches with the key reach $0.471$ to $0.478$.

Two further facts limit what the key gives a classical solver. Linear algebra with the key fixes at most $n=sr$ rows, the same number as an information set of the public matrix. The key changes the basis of $\optcode$ but not its dimension over $\F_p$. The key's decoder corrects at most $\ell$ errors, $133$ at $m=20000$, so rounding a target to a codeword with it needs a target within $\ell$ of the code. The best codeword is about $0.47m$ away. DQI uses the decoder differently, on the low-weight error patterns of its own superposition.

As a separate implementation check, the key holder's decoder recovered all $50$ random
error patterns tested within the licensed radius $\ell=93$. The absence of an
improvement from the searches using the key remains an experimental
observation. The hardness assumptions of Section~\ref{sec:mchardness} concern solvers
without the secret key and do not establish hardness when the classical solver also has the key.

\subsubsection{Known attacks and the hardness assumptions}

The most directly relevant cost estimates concern the distinguisher of Ghoshal, Ishai,
Jain and Sun~\cite{ghoshal26quasipoly}. It holds out a column of a generator matrix and
tests whether a low-degree polynomial vanishing on Hamming balls around the remaining
columns also vanishes on the held-out column. Here the input code is the alternant code
$\ker\Hpub$, whose generator matrix is obtained by linear algebra. The row space of
$\Hpub$ is instead the trace code described above.

Our adaptation of the authors' cost model to $\F_{11}$ includes this conversion and
optimizes over shortening. It gives approximately $2^{123}$ binary operations for their
theorem-backed variant and $2^{106}$ for the heuristic variant at $m=14000$. The
corresponding estimates are $2^{127}$ and $2^{110}$ at $m=20000$, and $2^{121}$ and
$2^{117}$ at $m=40000$. The theorem-backed estimate is lower at $m=40000$ because the cheapest attack parameters change: at $m=14000$ and $20000$ it uses polynomials of degree $8$ vanishing to order $7$, while at $m=40000$ degree $7$ and order $6$ suffice. The same model gives a
key-recovery estimate of about $2^{158}$ at $m=14000$ and about $2^{159}$ at $m=40000$. These are estimates for
specified attacks under that model, not lower bounds on the cost of distinguishing or
recovering our keys.

The additional reduction of roughly $20$ bits discussed in Remark~5.1
of~\cite{ghoshal26quasipoly} was estimated for binary Classic McEliece parameters.
It cannot be subtracted directly from the $\F_{11}$ heuristic estimate: part of that
reduction, the nullity prediction, is already included in the heuristic variant.
Accordingly, this adjustment does not supply a separate attack estimate of $2^{86}$ at
$m=14000$ or of $2^{97}$ at $m=40000$.

Other structural attacks remain relevant. Higher-order vanishing methods and their
key-recovery extensions~\cite{hemmert25vanishing,hemmert26keyrecovery}, as well as
Vedenev's method for subfield subcodes of generalized Reed--Solomon
codes~\cite{vedenev26keyrecovery}, require an instance-specific cost analysis here.
The same applies to the quadratic-forms and tangent-space
approaches~\cite{couvreur23quadratic,lemoine25tangent}. A degree restriction from one
version of these attacks does not exclude their extensions: in particular, the
tangent-space analysis also treats generic square-distinguishable alternant codes
with $r>p$. We therefore do not use $r=187>11$ or a simplified dimension criterion
to dismiss these methods. Table~\ref{tab:search} records the other key families
considered in the search.

These qualifications matter for interpreting Assumptions (A) and (B).
Assumption (A) concerns indistinguishability of the public key distribution, and the
cost estimates above describe some known approaches to challenging it. Assumption (B)
is a separate hypothesis about max-agreement on random matrices. It does not follow
from the usual hardness assumptions for syndrome decoding. The budget-scaling runs of
Section~\ref{sec:scalingconsequences} help assess this assumption. The experiments use
public-key matrices, but the generic searches do not use the secret decoding structure,
and the Gaussian reference $\hat z(N)$ does not depend on the matrix.
Extrapolating these results suggests that the generic searches would exceed
$\Qdqi^{(m)}-0.01$ well within $2^{80}$ operations at $m=14000$,
which is why the assumptions are stated at $m=40000$. Together, Assumptions (A) and (B) support the
conditional argument in Section~\ref{sec:mchardness}, whereas the experiments in
Table~\ref{tab:experiments} establish only the performance of the solvers tested.

\subsection{Multiplicative polynomial intersection}\label{sec:multiplicative}

Throughout this subsection, let $K=\F_{2^{h}}$ and let $M=2^{h}-1>3$ be prime.
Fix a generator $g$ of $K^{*}$ and pairwise distinct points
$t_1,\ldots,t_m,\zeta_1,\ldots,\zeta_n\in K$. The matrix of
Section~\ref{sec:logcauchy} is
\[
  B_{ij}=\log_g(t_i+\zeta_j)\in\F_M.
\]
We assume $\rank_{\F_M}(B)=n$ when identifying the optimizer rate with $\alpha=n/m$.
The support points and the field representation are public.

\subsubsection{Relation to earlier constructions}
\label{sec:mpiprovenance}\label{sec:mpistructure}

The use of discrete logarithms of shifted field elements is related to the
Chor--Rivest knapsack construction~\cite{chor88knapsack}. The more direct connection
for our decoder is to the polynomial lattices of Li, Ling, Xing and
Yeo~\cite{li17polylattice} and their rational-reconstruction decoding studied by
Lapiha~\cite{lapiha21polylattice}.

To see this connection, put
\[
  P(z)=\prod_{j=1}^{n}(z+\zeta_j)
\]
and consider the integer lattice
\[
  \Lambda_{\mathrm{dec}}
  =
  \left\{
    u\in\mathbb Z^{m}:
    \prod_{i=1}^{m}(z+t_i)^{u_i}\equiv1\pmod{P(z)}
  \right\}.
\]
Negative exponents are interpreted in the unit group of $K[z]/(P)$. Every
$z+t_i$ is a unit because the two sets of support points are disjoint.
By the Chinese remainder theorem, the defining congruence is equivalent to
\[
  \prod_{i=1}^{m}(t_i+\zeta_j)^{u_i}=1
  \quad\text{for every }j.
\]
Taking discrete logarithms gives
\[
  \Lambda_{\mathrm{dec}}
  =
  \{u\in\mathbb Z^{m}:B^{\mathsf T}u=0\pmod M\}.
\]
Its reduction modulo $M$ is therefore exactly the decoder code. The rational
decoder below uses this established algebraic construction. Its role in our candidate
is to supply the error alphabet required by the cosine objective, leading to the
DQI guarantee in Theorem~\ref{thm:logcauchy}(iii).

This restricted-error guarantee should be distinguished from a Hamming-distance
guarantee. At $m=6000$ and $n=420$, the decoder code has dimension $5580$ over
$\F_{8191}$. A Hamming license at $\ell=209$ would require $d\ge420$, whereas the
Singleton bound is $d\le421$. The actual Hamming distance is unknown. Our proof
instead establishes the separation property for signed errors required by
Definition~\ref{def:licensedS}.

A separate structural check also distinguishes one tested optimizer code from a
Reed--Solomon code of the same parameters. The \emph{Schur square} of a code is the
span of all componentwise products of pairs of codewords. At $m=600$ and $n=60$ (a separate probe instance with $\alpha=0.1$ and build seed $40$,
not one of the reported $\alpha=0.07$ instances), the tested optimizer code has Schur-square dimension $600$, whereas a generalized
Reed--Solomon code of dimension $60$ has square dimension $119$ at this length.
This observation concerns that instance and its unshortened code and does not
exclude attacks exploiting other algebraic structure.

The asymptotic statements in Theorem~\ref{thm:logcauchy} apply along any unbounded
sequence of admissible Mersenne primes and parameters. The existence of infinitely
many Mersenne primes is open. At the fixed field sizes $h=13$ and $h=17$, all reported DQI
values use the exact finite law of Lemma~\ref{lem:objlaw}.

\subsubsection{Rational reconstruction and unique decoding}
\label{sec:logcauchyproof}

We first prove the decoding statement in Theorem~\ref{thm:logcauchy}(i).
Let
\[
  \ell=\left\lfloor\frac{n-1}{2}\right\rfloor,
\]
so $n>2\ell$. For a signed pattern $y\in\{0,\pm1\}^{m}$, define
\[
  U_y(z)=\prod_{y_i=1}(z+t_i),
  \qquad
  V_y(z)=\prod_{y_i=-1}(z+t_i).
\]
The positive entries of $y$ become roots of the numerator, and the negative
entries become roots of the denominator. If $\mathrm{wt}(y)\le\ell$, both
polynomials have degree at most $\ell$ and are coprime.

Given the syndrome $s=B^{\mathsf T}y$, exponentiation converts its linear
equations over $\F_M$ into evaluations of a rational function over $K$:
\[
  \Phi_j:=g^{s_j}
  =
  \prod_i(t_i+\zeta_j)^{y_i}
  =
  \frac{U_y(\zeta_j)}{V_y(\zeta_j)}.
\]
The denominators are nonzero because no $\zeta_j$ equals a $t_i$.

\begin{proof}[Proof of Theorem~\ref{thm:logcauchy}(i)]
Suppose two patterns $y,y'$ of weight at most $\ell$ have the same syndrome.
Then
\[
  U_y(\zeta_j)V_{y'}(\zeta_j)
  -
  U_{y'}(\zeta_j)V_y(\zeta_j)=0
  \qquad\text{for all }j.
\]
The polynomial on the left has degree at most $2\ell<n$, but has the $n$
distinct roots $\zeta_1,\ldots,\zeta_n$. It is therefore identically zero.
The two reduced rational functions coincide, and unique factorization in $K[z]$
gives $y=y'$. Thus the syndrome map is injective on the stated error set.

To invert the map, solve the homogeneous linear system
\[
  U(\zeta_j)-\Phi_j V(\zeta_j)=0,
  \qquad j=1,\ldots,n,
  \qquad \deg U,\deg V\le\ell.
\]
A nonzero solution exists because $(U_y,V_y)$ is one. Every solution satisfies
\[
  UV_y-U_yV=0:
\]
as above, this polynomial has degree at most $2\ell$ and vanishes at all $n$
evaluation points. Since $U_y$ and $V_y$ are coprime, every nonzero solution is
of the form
\[
  (U,V)=(hU_y,hV_y)
\]
for a polynomial $h$. Dividing by the polynomial gcd and normalizing therefore
recovers $(U_y,V_y)$. Scanning the support points $t_i$ for roots identifies
the positive and negative coordinates.

Exponentiating the syndrome costs $O(nh)$ field multiplications. Gaussian
elimination costs $O(n^3)$ field operations, and evaluating the reconstructed
polynomials at the $m$ support points costs $O(m\ell)$. Polynomial gcd computation
is absorbed by these bounds, giving
\[
  O(n^3+m\ell+nh)
\]
field operations.
\end{proof}

The decoder can be defined on arbitrary syndromes by checking that the reduced
numerator and denominator have the required normalization, have only simple roots
among the $t_i$, and reproduce the supplied evaluations. It returns failure if these
checks fail. We retain separate degree bounds on the numerator and denominator,
rather than imposing a total-weight bound on the output. This is what permits the
additional recovery guarantee in part (ii).

\subsubsection{Recovery beyond the worst-case radius}
\label{sec:inmeasure}\label{sec:inmeasureproof}

The same rational decoder can recover patterns with more than $\ell$ nonzero entries
when their positive and negative signs are sufficiently balanced. This is an
average-case statement about signs, distinct from the worst-case license.

For clarity, say that the decoder is \emph{licensed in measure at reach $\lambda$}
if, for each weight $k\le\lfloor\lambda m\rfloor$, its failure probability
$\varepsilon_k$ on uniformly sampled signed patterns satisfies
\[
  \max_{0\le k\le\lfloor\lambda m\rfloor}\varepsilon_k\longrightarrow0.
\]
Only Theorem~\ref{thm:logcauchy}(ii) uses this notion.

\begin{proof}[Proof of Theorem~\ref{thm:logcauchy}(ii)]
Fix a support of size $k\le2\ell$, and let $k_+$ be the number of positive signs.
Then
\[
  \deg U_y=k_+,\qquad \deg V_y=k-k_+.
\]
The reconstruction proof above applies whenever both degrees are at most $\ell$.
It used these separate degree bounds, not the stronger condition
$\mathrm{wt}(y)\le\ell$.

With independent uniform signs, $k_+\sim\mathrm{Bin}(k,\tfrac12)$. The original
pattern is therefore recovered exactly when
\[
  k_+\le\ell
  \quad\text{and}\quad
  k-k_+\le\ell.
\]
If either inequality fails, the reduced rational function of the original pattern
lies outside the decoder's degree bounds, so that pattern cannot be its output.
Counting both detected failures and incorrect outputs as failures gives, for
$1\le k\le2\ell$,
\[
  \varepsilon_k
  =
  \Pr\left[
    \left|\mathrm{Bin}(k,\tfrac12)-\frac{k}{2}\right|
    >
    \ell-\frac{k}{2}
  \right]
  \le
  \min\left\{
    1,\,
    2\exp\left(-\frac{2(\ell-k/2)^2}{k}\right)
  \right\},
\]
with $\varepsilon_0=0$.

For any fixed $\xi\in(0,1)$ and $k\le(1-\xi)2\ell$, the exponential bound is at most
\[
  2\exp\left(-\frac{\xi^2\ell}{1-\xi}\right).
\]
Since $2\ell/m\to\alpha$, every fixed $\lambda<\alpha$ eventually satisfies
$\lambda m\le(1-\xi)2\ell$ for some fixed $\xi>0$. The failure probabilities then
vanish uniformly over all weights up to $\lambda m$.
\end{proof}

The strict inequality $\lambda<\alpha$ is essential. At $k=2\ell$, recovery requires
exactly $\ell$ positive and $\ell$ negative signs, which occurs with probability
\[
  \binom{2\ell}{\ell}2^{-2\ell}
  \sim\frac{1}{\sqrt{\pi\ell}}.
\]
Thus the endpoint itself has a vanishing recovery probability.

This extension concerns the decoder's recovery probability. It does not establish
the stronger separation hypothesis used by the exact DQI quality law.
The performance guarantee in Theorem~\ref{thm:logcauchy}(iii) uses the worst-case
signed-error radius, whose separation property we prove next.

\subsubsection{Separation and the DQI guarantee}\label{sec:logcauchysep}

Unique decoding compares two patterns in the prepared error set. The objective-law
license also compares such a pattern with one obtained by adding a Fourier frequency
to another. For the cosine objective this means comparing $y'$ with $y+ae_i$,
where $a\in\{\pm1\}$. The latter vector may contain an entry $\pm2$, so injectivity
on $\{0,\pm1\}$ patterns alone is insufficient.

\begin{proof}[Verification of the separation hypothesis]
Assume $n\ge2\ell+2$, and let $y,y'\in\{0,\pm1\}^{m}$ have weight at most $\ell$.
For $a\in\{\pm1\}$, write $u=y+ae_i$ using integer coordinates in
$\{-2,-1,0,1,2\}$. Since $M>3$, these coordinates have unambiguous reductions in
$\F_M$. Define
\[
  U_u(z)=\prod_{u_i>0}(z+t_i)^{u_i},
  \qquad
  V_u(z)=\prod_{u_i<0}(z+t_i)^{-u_i}.
\]
These polynomials are monic and coprime, and
\[
  \deg U_u+\deg V_u=\|u\|_1\le\ell+1.
\]

Suppose $B^{\mathsf T}u=B^{\mathsf T}y'$. Exponentiating the syndromes shows that
\[
  U_uV_{y'}-U_{y'}V_u
\]
vanishes at every $\zeta_j$. Its degree is at most
\[
  \|u\|_1+\|y'\|_1\le2\ell+1<n,
\]
so it is identically zero. Equality of the reduced monic fractions then gives
$u=y'$ by unique factorization. Consequently,
\[
  B^{\mathsf T}y'\ne B^{\mathsf T}(y+ae_i)
  \qquad\text{whenever }y'\ne y+ae_i,
\]
as required by Definition~\ref{def:licensedS}.
\end{proof}

When $n$ is even and $\ell=\lfloor(n-1)/2\rfloor$, we have $n=2\ell+2$.
The decoder is therefore licensed for $\{\pm1\}$ at this radius. Applying
Lemma~\ref{lem:objlaw}, with the separation condition of Lemma~\ref{lem:cosinelaw}
now established, gives the exact finite DQI expectation. Along an admissible sequence
with $\ell/m\to\lambda=\alpha/2$, its expected mean cosine tends to
\[
  \sqrt{2\lambda(1-\lambda)}.
\]
This proves Theorem~\ref{thm:logcauchy}(iii).

\subsubsection{Rank and computational resources}\label{sec:mpiresources}

Distinct support points do not by themselves imply $\rank(B)=n$. For example, over
\[
  \F_8=\F_2[X]/(X^3+X+1),\qquad g=X,
\]
the choices
\[
  t=(X,X^2,X^2+X),\qquad \zeta=(0,1)
\]
give the two columns $(1,2,4)$ and $(3,6,5)$, which are proportional modulo $7$.
If the rank is smaller than $n$, the optimizer rate is $\rank(B)/m$, not $n/m$.
We therefore check the rank of every reported instance by Gaussian elimination
over $\F_M$.

The decoding bound in Theorem~\ref{thm:logcauchy}(i) counts operations in $K$.
Using quadratic-time arithmetic on $h$-bit field elements gives
\[
  O\bigl((n^3+m\ell+nh)h^2\bigr)
\]
bit operations, before the overhead of a reversible implementation.

Constructing $B$ is a separate cost because it requires $mn$ discrete logarithms.
A complete logarithm table has $O(2^h)$ entries and supplies these values by lookup.
This construction is polynomial in $m$ when $2^h\le m^C$ for a fixed constant $C$,
for example when $m$ is a fixed fraction of $2^h$. The condition
$m+n\le2^h$ alone does not imply this upper bound on the field size. The table cost
must therefore be included explicitly along other parameter sequences.
The decoding estimate also does not include a precision or gate-count analysis
for preparing the cosine state.

\subsubsection{The lattice attack}\label{sec:latticestudy}

A classical solver works in the exponent domain: it seeks $x\in\F_M^n$ for which
$(Bx)_i-c_i$ is small modulo $M$ on many rows. The attack tested here approximates
this task by a closest-vector problem on a subset of rows.

Choose a subset $J$ of $m'$ rows with $\rank(B_J)=n$, and form the lattice
\[
  \Lambda_J
  =
  \{v\in\mathbb Z^{m'}:
      v\bmod M\in\operatorname{im}(B_J)\}.
\]
The rank condition is important. It implies that the image has $M^n$ elements and
hence that
\[
  \det\Lambda_J=M^{m'-n}.
\]
Without it, the exponent is $m'-\rank(B_J)$.

Starting from a systematic basis, the attack applies LLL, optionally followed by
BKZ with block size $\beta$, and then uses Babai's nearest-plane algorithm to find
a lattice vector near an integer representative of $c_J$. This Euclidean
closest-vector objective is a surrogate for the cosine objective. An invertible
$n\times n$ submatrix of $B_J$ recovers the corresponding $x$, which is scored on
all $m$ rows. The run repeats this procedure on random subsets and keeps the best
assignment found within its budget.

For any one candidate constructed using only $B$ and $c_J$, the expected cosine on
a row outside $J$ is zero, because that row's center remains independent and uniform.
This statement applies before choosing among candidates by their full-instance
scores. Such selection can favor candidates that happen to score well outside
their construction subsets.

\begin{table}[t]
\centering
\small
\begin{tabular}{rrrrrrr}
\toprule
$m$ & $n$ & DQI (finite) & $m'$ & $\beta$ & subsets & best mean cosine\\
\midrule
600  & 42  & 0.2235 & 100 & 20 & 50 & 0.2172\\
600  & 42  & 0.2235 & 150 & 20 & 27 & 0.2873\\
2000 & 140 & 0.2433 & 180 & 2  & 50 & 0.1213\\
2000 & 140 & 0.2433 & 180 & 20 & 9  & 0.1112\\
2000 & 140 & 0.2433 & 220 & 2  & 50 & 0.1521\\
2000 & 140 & 0.2433 & 220 & 20 & 50 & 0.1390\\
2000 & 140 & 0.2433 & 260 & 2  & 8  & 0.0816\\
2000 & 140 & 0.2433 & 260 & 20 & 2  & 0.0677\\
6000 & 420 & 0.2518 & 470 & 2  & 50 & 0.0942\\
6000 & 420 & 0.2518 & 520 & 2  & 5  & 0.0848\\
\bottomrule
\end{tabular}
\caption{Row-subset lattice attacks on seed $3101$ at each size, within the time
budgets given in the text. Here $m'$ is the subset size, $\beta=2$ denotes LLL alone,
``subsets'' counts the completed subsets, and the last column is the best score over
those subsets. Every reported score is evaluated on the full instance.
DQI values are exact finite expected mean cosines. Means over fixed numbers of subsets,
at block sizes up to $60$, are given in Table~\ref{tab:scalinglattice}.}
\label{tab:lattice}
\end{table}

Table~\ref{tab:lattice} reports the completed runs on seed $3101$.
At $m=600$, each configuration ran for about $300$~s. At $m=2000$, each ran for
about $1200$~s, except BKZ-20 on $260$ rows, which took $1596$~s.
At $m=6000$, the two LLL runs took $1288$~s and $4060$~s.
These budgets completed different numbers of subsets, so the scores reflect
both the quality of each reduction and the number of candidates examined.

At $m=600$, BKZ-20 beats DQI's $0.2235$ on all three seeds. On seed $3102$
the best tested mean cosine is $0.257$ using $100$ rows, and the $150$-row run also beats DQI,
reaching $0.246$. On seed $3103$ the best is
$0.265$ using $150$ rows.
The row-subset lattice attack is the strongest classical attack tested at that size.
The values above are the best scores over the subsets examined. In the rerun of
Section~\ref{sec:scalinglattice} on seed $3101$, individual subsets give scores up to $0.307$ at block sizes
$50$ and $60$, and the mean score over subsets exceeds DQI's value at the tested block sizes of $30$ and above
(Table~\ref{tab:scalinglattice}). Longer runs of the generic searches exceed it as well
(Table~\ref{tab:scaling}).

At $m=2000$, LLL on $220$ rows reaches $0.1521$ on seed $3101$, slightly above the
generic portfolio's score of about $0.149$. Runs on seeds $3102$ and $3103$
returned scores $0.142$ and $0.126$ after LLL, before later stages failed in
\texttt{fplll}. Those scores are below the corresponding portfolio values of
about $0.152$ and $0.151$. The completed lattice runs remain below DQI's
$0.2433$, but do not establish the performance of stronger or longer reductions.

At $m=6000$, the LLL runs on seed $3101$ reach at most $0.0942$. BKZ-20 failed in
double precision on all three seeds. A multiprecision rerun on seed $3101$
completed one subset of $520$ rows, scoring $0.0696$ in $2432$~s, and one subset
of $700$ rows, scoring $0.0192$ in $2600$~s. These limited completed runs are below
DQI's $0.2518$, and the numerical failures themselves provide no evidence of hardness.
The rerun of Section~\ref{sec:scalinglattice}, which uses $212$-bit mantissas at this
size and fixed numbers of subsets, does not encounter these failures: BKZ-20 to
BKZ-40 complete, while runs at block size $50$ reach the time limit.

The lattice experiments therefore show that the attack exceeds DQI at $m=600$ and that
the tested configurations fall short of it at $m=2000$ and $6000$.
They do not bound stronger lattice methods. Section~\ref{sec:scalinglattice}
reruns the attack with a fixed protocol at block sizes up to $60$, including at $m=1000$,
where some individual subsets yield scores above DQI's value at block sizes $50$ and $60$.
That section also models the mean score over subsets, estimating the block size
needed for this mean to reach DQI's value at lengths up to $m=3000$. Assumption (C) in
Section~\ref{sec:logcauchy} remains a separate hardness hypothesis.

%% file: supplement/S3-search-checks.tex
\section{Search methodology and computational checks}\label{sec:supp-search}

This section describes how the candidate problems were selected and how the numerical
comparisons were evaluated. Section~\ref{sec:search} summarizes the search and its outcomes.
Section~\ref{sec:checks} connects the retained claims to their scripts and records, and
Section~\ref{sec:solvers} specifies the classical solvers and their budgets. Construction
details and decoding proofs appear separately in Section~\ref{sec:supp-candidates}.

\subsection{The candidate search}\label{sec:search}

We used a system of AI agents, guided by human selection, to propose optimization problems
and code constructions, evaluate their screening conditions, and develop decoders and
classical attacks for selected proposals. The search drew on both existing code families
and newly proposed combinations of codes, objectives, and access models. As described in
Section~\ref{sec:intro}, the subfield rule \gate{4} and the objective-dependent form of the
rules were developed during this process.

The resulting catalog contains $109$ problem and code-family entries. Script~10 evaluates
the screening inequalities from the parameters and bounds supplied for each entry.
Script~32 checks that each catalog entry has a recorded disposition. These scripts
organize the finite catalog. The mathematical content of a screening decision comes from
the bounds, decoder guarantees, and classical algorithms attached to that entry.
Likewise, the evaluation time reported in Section~\ref{sec:checkcost} concerns the scalar
checks once these inputs are available.

Several different outcomes led us to set a proposal aside. Some violated a necessary
condition under the rule's hypotheses. Others lacked a suitable worst-case decoder, or
the available decoding guarantee was too weak to certify an advantage. Some exposed
structure that a classical algorithm could exploit. Many evaluation-code proposals
retained the polynomial-intersection mechanism of OPI. Other application proposals
did not produce an optimization problem within the framework.

These outcomes have different implications. An insufficient guarantee for a known decoder
does not bound the reach of every possible decoder. A structural distinguisher does not
by itself solve the optimization problem, although it can undermine a proposed
hidden-structure comparison. We therefore state the search outcomes for the constructions,
parameters, and algorithms examined, rather than treating every discarded proposal as
an excluded code family.

\begin{table}[t]
\centering
\footnotesize
\setlength{\tabcolsep}{4pt}
\begin{tabular}{L{0.27\textwidth} L{0.52\textwidth} L{0.12\textwidth}}
\toprule
proposals examined & outcome of the screening or follow-up analysis & scripts\\
\midrule
Evaluation-code constructions, including Reed--Solomon, algebraic-geometry, twisted,
folded, and multivariate variants
&
The proposed objectives retained the polynomial-evaluation mechanism of OPI.
They were not pursued as a distinct mechanism for this paper.
&
10\\

Sparse and local constructions, including LDPC, polar, and locally recoverable codes
&
Depending on the construction, small dual distance violated \gate{1}(a), or the
available decoder lacked the worst-case guarantee required by \gate{2}.
These are separate reasons for setting a proposal aside.
&
10\\

Expander, Tanner, and near-MDS expander constructions
&
The available worst-case reaches did not certify an advantage at the parameters
examined. In the local constructions considered, additional dual structure also
suggested classical attacks.
&
10, literature\\

Hidden structured keys and their proposed disguises
&
For the variants examined, structural distinguishers or recovery methods exposed
features intended to remain hidden. The checks included shortening, Schur products,
and recovery of low-weight structure.
&
31, 34, 38, literature\\

BIKE, HQC, and binary Goppa parameter sets
&
The studied decoding guarantees or finite margins did not establish the proposed
advantage. The binary Goppa study included comparison with the specified
Prange-restart budget.
&
29, literature\\

Low-density parity-check objectives and partially stuck-at memory formulations
&
The former violated the density condition of \gate{1}(c) at the parameters
examined. The latter admitted a stronger subfield-linear-algebra baseline, motivating
\gate{4}.
&
15--17, 31\\

Restricted-alphabet lattice decoders and alternative maps from exponents
&
The proposals did not supply a decoder with the required reach and error-alphabet
guarantee at the chosen parameters. Some alternative maps also required an
exponentially large modulus.
&
33, 42\\

Application proposals outside the optimization framework
&
Some were not optimization problems or did not yield the required matrix-and-objective
formulation. Applications reducing to polynomial intersection are included in the
first row.
&
10\\
\midrule
Alternant max-agreement over a prime field
&
Carried forward with a secret decoding representation and the comparison model of
Section~\ref{sec:mckey}.
&
29, 30, 49\\

Log-Cauchy matrices with a cosine objective
&
Carried forward with the public signed-error decoder of
Theorem~\ref{thm:logcauchy}, using the objective-dependent rules of
Proposition~\ref{prop:alphabet}.
&
40, 41\\
\bottomrule
\end{tabular}
\caption{Grouped outcomes from the candidate catalog. Each row summarizes multiple
proposals or parameter choices, and the rows are not a count of the $109$ catalog entries.
``Literature'' denotes an assessment based on published constructions or attacks rather
than a dedicated script. The two retained candidates are analyzed in
Section~\ref{sec:problems}.}
\label{tab:search}
\end{table}

The two retained candidates change different parts of the usual OPI setting.
AMA changes access to the decoding structure. MPI
changes the objective so that its Fourier support matches
a signed-error decoder. MPI's guarantee does not require knowing the code's
Hamming distance: the relevant statement is the restricted-error license proved in
Section~\ref{sec:logcauchysep}. For both candidates, the broader classical-hardness
questions remain those stated in Section~\ref{sec:problems}, and the screening outcome
does not resolve them.

\subsection{Computational checks and stored records}\label{sec:checks}

Table~\ref{tab:checks} identifies the scripts supporting the numerical analyses retained
in the paper. Script $k$ is the Python file in \texttt{scripts/} whose name begins with
$k$, in the repository identified under Code and data availability. Stored outputs are
under \texttt{results/}, and \texttt{REPRODUCE.md} gives the commands associated with the
table.

Scalar evaluations compute values of proved
formulas. Small-instance state calculations test the implementation of the DQI
expectation formula. Decoder probes test recovery on specified inputs, while optimization
runs record assignments found within a budget. None of these numerical checks replaces
the corresponding mathematical proof.

Some scripts retain earlier terminology. In script~10, labels (D0)--(D3) correspond
to parts of \gate{1}, (D4) to \gate{2}, (D5) to \gate{3}, and (D6) to \gate{4}.
Label (D7) concerns the objective alphabet treated in
Section~\ref{sec:rulesalphabet}. Script~48 collects scalar calculations used in the
paper alongside the more specialized scripts.

\begin{table}[t]
\centering
\footnotesize
\setlength{\tabcolsep}{4pt}
\begin{tabular}{L{0.40\textwidth} L{0.18\textwidth} l L{0.22\textwidth}}
\toprule
analysis & location & scripts & type of check\\
\midrule
Reach threshold, sharp margin cap, and screening inequalities
&
\S\ref{sec:rules}, Proposition~\ref{prop:cap}
&
13, 23, 48
&
Scalar evaluations and randomized consistency checks\\

OPI finite DQI values and classical comparisons
&
\S\ref{sec:opi}, Table~\ref{tab:experiments}
&
18, 48
&
Finite spectral values and optimization records\\

Alternant parameter choices, margins, and rule evaluations
&
\S\ref{sec:mcconstruction}, Table~\ref{tab:candrules}
&
29, 30, 48
&
Scalar and finite-length calculations\\

Alternant decoding checks and public-instance solver comparisons
&
\S\ref{sec:mcexperiments}, Table~\ref{tab:experiments}
&
30, 48
&
Decoder probes and optimization records\\

Cost estimates for the attacks of~\cite{ghoshal26quasipoly}
&
\S\ref{sec:mcattacks}
&
49
&
Evaluation of the specified cost model\\

Classical optimization with the alternant key
&
\S\ref{sec:mckeyholder}
&
36
&
Optimization records\\

Reed--Solomon list-recovery threshold at the alternant parameters
&
Lemma~\ref{lem:listrecovery}
&
48
&
Scalar evaluation\\

Multiplicative polynomial intersection: rule evaluations, finite DQI values, and portfolio comparisons
&
\S\ref{sec:logcauchy}, Table~\ref{tab:experiments}
&
40, 41, 48
&
Finite values and optimization records\\

Signed-error decoding checks and code-structure calculations
&
\S\ref{sec:inmeasureproof}, \S\ref{sec:mpistructure}
&
40, 46, 48
&
Decoder probes, binomial probabilities, and ranks\\

Row-subset lattice attacks
&
\S\ref{sec:latticestudy}, Table~\ref{tab:lattice}
&
43, 48
&
Lattice runs and finite DQI comparisons\\

General-objective expectation formula
&
\S\ref{sec:objlawproofs}
&
44
&
Direct small-instance state construction\\

Subfield examples, affine-subspace counts, and equation ranks
&
\S\ref{sec:subfield}
&
15--17, 48
&
Finite-field calculations\\

Finite-length Prange-restart thresholds
&
Theorem~\ref{thm:converse}, Table~\ref{tab:candrules}
&
48
&
Evaluation using the finite DQI margin\\

Budget scaling of the generic searches
&
\S\ref{sec:scaling}, Table~\ref{tab:scaling}
&
50, 52
&
Optimization records with budgets in moves\\

Searches given the alternant key, with budgets in moves
&
\S\ref{sec:scalingkey}
&
53
&
Optimization records\\

Lattice attack at block sizes up to $60$, and its model
&
\S\ref{sec:scalinglattice}
&
51, 52
&
Lattice runs and model evaluation\\
\bottomrule
\end{tabular}
\caption{Scripts supporting the numerical analyses retained in the paper.
The DQI values are computed expectations, rather than observed quantum-sampling
averages. Attack-cost calculations are estimates within their stated models, and the
costs derived from the scaling runs are extrapolations rather than bounds.}
\label{tab:checks}
\end{table}

\paragraph{Checks of the objective law.}
Script~44 constructs small DQI states directly and compares their expected objectives
with the tridiagonal expression in Lemma~\ref{lem:cosinelaw}. The tested objectives
include indicators of five-element subsets of $\F_{11}$, the cosine over $\F_{13}$,
and the two-frequency function
\[
  F(v)=\cos(2\pi v/q)+0.3\cos(4\pi v/q).
\]
On the tested instances satisfying the separation hypothesis, the two calculations
agree to floating-point precision. These are numerical consistency checks of the
identity proved in Section~\ref{sec:objlawproofs}. Calculations on instances violating
the hypothesis lie outside that guarantee.

\paragraph{Instances and seeds.}
The reporting scripts regenerate instances from their seeds. The principal optimization
runs use seeds $2101$--$2103$ for OPI and $3101$--$3103$ for the two candidates.
The AMA instance at $m=20000$ uses the same three seeds.
Table~\ref{tab:lattice} and Section~\ref{sec:latticestudy} identify the completed lattice
runs and the failures encountered. Script~15 uses seeds $1701$--$1703$, and the
decoder probe uses seed $40$.

A seed specifies the regenerated instance. Wall-clock stopping criteria and parallel
solver execution can nevertheless change the search trajectory and best score obtained
in a rerun.

\paragraph{Saved witnesses.}
The AMA runs at $600$~s budgets and the MPI runs at $m=2000$
store the codeword $Bx$ corresponding to each solver's best assignment. Such a record
allows two independent checks: code membership can be verified by solving the linear
system $Bx=c$ for the stored vector $c$, and its objective can be recomputed from the
regenerated instance. The reported $600$~s AMA comparisons use these records.

The OPI runs and the MPI runs at $m=6000$ retain scores but not the associated
codewords. Their scores were recomputed when the assignments were found, but the saved
records do not permit the same later witness check. The runs of
Section~\ref{sec:supp-scaling} likewise store no assignments.

The AMA portfolio records are under
\texttt{results/logs-bud600/}, with filenames matching
\texttt{ext6-30-*}, and the key-aware records match \texttt{ext6-36-*}.
The longer AMA run is identified by the tag \texttt{ext6bud1800}.
The reporting commands in \texttt{REPRODUCE.md} select the records used in each table.

\subsection{Classical solvers and budgets}\label{sec:solvers}

The generic portfolio operates on the public matrix and row objectives. It consists of
Prange completions, line search, simulated annealing, and, where specified, CP-SAT.
The key-aware polynomial search and the lattice attack on MPI are additional
tests using the respective constructions' structure. Every reported score is evaluated
on the full instance, including when a solver optimizes only a subset of its rows.

\paragraph{Prange completions.}
Each restart chooses an information set, assigns accepted values on those coordinates,
and completes them to a codeword. For the cosine objective, it sets those coordinates
to their centers. The solver keeps the best result from the prescribed number of
randomized completions.

Theorem~\ref{thm:converse} applies when the information set is chosen independently
of the accepted sets or centers and the completion uses only their values on that set.
Its comparison with $2^{40}$ restarts is a theoretical bound. The much smaller numbers
of completions actually run are listed in Table~\ref{tab:solver-budgets}.

\paragraph{Line search.}
Let $k=\rank(B)$ and let $G\in\F_q^{m\times k}$ be a full-column-rank generator of
$\optcode$. For an information set $T$ of size $k$, define
\[
  L=G\,G_T^{-1}.
\]
Then $L_T=I_k$: column $j$ changes the $j$-th information coordinate and leaves the
other information coordinates unchanged. When $B$ already has full column rank,
we may take $G=B$.

Starting from a Prange completion $c\in\optcode$, the search chooses a column $j$
and evaluates the $q$ codewords
\[
  c+aL_{\cdot j},\qquad a\in\F_q.
\]
It moves to a point with the largest objective on that line and retains the best
codeword encountered.

The systematic basis can be updated by replacing the $j$-th information row with
a row $i\notin T$ whenever $L_{ij}\ne0$. This nonzero-pivot condition preserves
invertibility. After the update, moves along the other basis columns leave coordinate
$i$ unchanged. For indicator objectives, the pivot can preserve a newly satisfied
row by choosing such a row outside $T$. The search restarts after $60$ consecutive
moves without improvement on the AMA and MPI instances, and after
$40$ on the OPI instances. The runs with a fixed number of moves in
Section~\ref{sec:supp-scaling} use $60$ for all three families.

\paragraph{Simulated annealing.}
The annealer uses the same codeword-line moves, but samples the next point with
probability proportional to the exponential of its objective value times the current
inverse temperature. This heat-bath update allows moves that reduce the score.
The inverse temperature increases geometrically from $0.05$ to $4$ over the run.
As in line search, the best codeword encountered is retained.

These adaptive searches use information from the row objectives and are separate
from the restricted Prange-restart baseline of Theorem~\ref{thm:converse}.

\paragraph{CP-SAT.}
CP-SAT~\cite{ortools} is warm-started from the best assignment already found.
It is used on OPI only at $m\le210$ and is given $90$~s on the AMA instances.
At $m=14000$, its optimization model uses $3000$ rows, and the returned assignment is
then scored on all $14000$ rows. The runs on seeds $3102$ and $3103$, and the longer
budget experiment, use two workers. These are time-limited solution searches, so
their reported scores are feasible objective values rather than optimality
certificates.

\paragraph{Search with the alternant key.}
Script~36 performs line search and simulated annealing along
$\F_{p^{s}}$-lines of message polynomials, as described in
Section~\ref{sec:mckeyholder}. Both searches start from the best assignment found
by the searches using only the public matrix and accepted sets. When the extension
field has more than $2\times10^{5}$ elements, each move samples $4096$ points on its line and evaluates
each sampled point exactly. At $s=6$, this permits about $180$ moves per run,
rather than an exhaustive scan of each extension-field line.
Script~53 runs the same searches for a prescribed number of moves
(Section~\ref{sec:scalingkey}).

\paragraph{Lattice reduction.}
The row-subset lattice attack is specified in Section~\ref{sec:latticestudy}.
Its LLL and BKZ configurations, completed subsets, time budgets, and numerical
failures are reported there and in Table~\ref{tab:lattice}. Those runs are distinct
from the generic portfolio. Section~\ref{sec:scalinglattice} reports a rerun with
extended-precision arithmetic, fixed numbers of subsets and block sizes up to $60$.

\begin{table}[t]
\centering
\small
\begin{tabular}{lrrrl}
\toprule
family & $m$ & seeds & Prange completions & time per search\\
\midrule
OPI          & 400   & $2101$--$2103$ & $20000$ & $40$~s\\
OPI          & 4000  & $2101$--$2103$ & $12500$ & $40$~s\\
AMA          & 14000 & $3101$--$3103$ & $1000$  & $600$~s\\
AMA          & 20000 & $3101$--$3103$  & $1000$  & $600$~s\\
MPI & 2000 & $3101$--$3103$ & $125$ & $600$~s\\
MPI & 6000 & $3101$--$3103$ & $24$  & $40$~s and $600$~s\\
\bottomrule
\end{tabular}
\caption{Search budgets for the fixed-time comparisons of
Table~\ref{tab:experiments}. Budgets specified in moves are given in
Section~\ref{sec:supp-scaling}. The MPI instance at $h=17$ and $m=20000$ uses seeds $3101$--$3103$ and $600$~s per search.
Prange counts are per instance. The last column gives the budget for each of line
search and simulated annealing, separately. For MPI at $m=6000$,
the reported result takes the best of the $40$~s and $600$~s runs for each seed.
CP-SAT, key-aware searches, and lattice attacks have the separate budgets described
in the text. OPI, AMA and MPI denote optimal polynomial intersection, alternant
max-agreement and multiplicative polynomial intersection.}
\label{tab:solver-budgets}
\end{table}

\paragraph{Budget sensitivity and smaller instances.}
The $1800$~s budget-sensitivity runs described after Table~\ref{tab:experiments} use seed $3101$.
For the AMA instance at $m=14000$, increasing each search budget from $600$~s
to $1800$~s raises the best reported quality from $0.5144$ to $0.5157$.
For the MPI instance at $m=6000$, the corresponding best mean cosine
increases from $0.1146$ to $0.1169$, compared with DQI's finite expectation
$0.2518$.

At smaller OPI lengths, the portfolio already exceeds DQI's finite expectation.
At $m=100$, line search and annealing reach qualities $0.79$--$0.80$, compared
with $0.637$ for DQI. At $m=210$, they reach $0.729$--$0.733$, compared with
$0.674$. At $m=400$, searches with $10^6$ moves per chain match or exceed DQI's
expected quality (Section~\ref{sec:scalingconsequences}). These results support the
finite-size qualification in Section~\ref{sec:opi}.

All these comparisons concern achieved objective values at the stated budgets.
Increasing the instance size while keeping a wall-clock budget fixed also reduces
the search effort available per coordinate. The observed change in the classical
scores therefore cannot, by itself, be attributed to increasing problem hardness.

%% file: supplement/S4-scaling.tex
\section{Scaling of the classical attacks}\label{sec:supp-scaling}

The comparisons of Section~\ref{sec:problems} stop each classical search after a fixed wall-clock
time. Their scores therefore depend on the machine, and they do not show how the scores grow with
the effort spent. This section repeats the line search and the annealing with budgets counted in
moves, spanning up to five orders of magnitude, and compares the three families on
a common scale. It also reruns the searches supplied with the alternant key and the lattice attack with larger budgets.
Scripts~50--53 of the code repository produce the data. They regenerate the instances from their
seeds (Section~\ref{sec:solvers}).

\subsection{Budgets counted in moves}\label{sec:scaling}

One \emph{move} of the line search or of simulated annealing (Section~\ref{sec:solvers}) scores
the $q$ codewords $c+aL_{\cdot j}$, $a\in\F_q$, on all $m$ rows and then updates the systematic
basis. The update costs about $mn$ field operations. The scoring costs about $m\rho q$ operations
for accepted sets of density $\rho$, and about $m+q\log_2q$ for the cosine, where one fast Fourier
transform of length $q$ scores the whole line. Counting moves therefore gives a machine-independent measure of effort for a fixed instance.

Annealing is run separately at each run length $N_{\mathrm{mv}}=10^2,3\times10^2,10^3,\dots$,
because its temperature schedule depends on the run length: the inverse temperature increases
geometrically from $0.05$ to $4$ over the run. We record the best quality visited. The line search
has no such schedule. It is run for a fixed number of moves per chain, with the restart rule of
Section~\ref{sec:solvers} (restart after $60$ moves without improvement), and its best quality
so far is recorded along the run. Each run uses eight independent chains (six for OPI at $m=400$). A run length is reported once all of
its chains have finished on every seed, so the longest reported one depends on the instance. It
ranges up to $10^{7}$ moves per chain (Table~\ref{tab:scaling}).

The seeds are those of the fixed-time comparisons, $2101$--$2103$ for OPI and $3101$--$3103$ for
the two candidates. We use three seeds for OPI at $m=400$ and $4000$, AMA at $m=14000$, and MPI at $m=600$, $2000$ and $6000$, and seed $3101$, one of the three seeds of Table~\ref{tab:experiments}, for AMA at $m=20000$. Further lengths are added on
one seed, to show where the comparison with DQI changes sign: OPI at $m=1008$ and $2002$, AMA at $m=2000$ and $6000$, and MPI at $m=1000$ and $4000$.

\subsection{A common scale for the three families}\label{sec:scalinglaw}

For a Prange completion with independently drawn accepted sets or centers, the $m-n$ rows outside the information set have independent objective
values, each with variance $\sigma^2$, where $\sigma^2=\rho(1-\rho)$ for accepted sets of density
$\rho$ and $\sigma^2=\tfrac12$ for the cosine. The quality of a completion thus fluctuates around
$\Qpr$ with standard deviation $\sigma\sqrt{(1-\alpha)/m}$, and we measure every quality $Q$ in
these units:
\begin{equation}\label{eq:zscore}
  z(Q)=\frac{Q-\Qpr}{\sigma\sqrt{(1-\alpha)/m}} .
\end{equation}
As a Gaussian reference for independent sampling, we approximate completion scores by Gaussian random variables. We define the level $\hat z(N)$ by
\begin{equation}\label{eq:samplinglevel}
  N\,\overline\Phi\bigl(\hat z(N)\bigr)=1 ,
\end{equation}
where $\overline\Phi$ is the upper tail of the standard normal distribution. At this level, the expected number of exceedances among $N$ independent draws is one. To leading order
$\hat z(N)=\sqrt{2\ln N}$. For the values of $N$ in our runs the exact $\hat z(N)$ is lower than
$\sqrt{2\ln N}$ by $0.4$ to $0.6$, and the mean of the best of $N$ draws exceeds
$\hat z(N)$ by about $0.1$. Theorem~\ref{thm:converse} gives a rigorous bound for the restricted Prange-restart baseline. In these units its hypothesis reads
$z(\Qdqi^{(m)})^2\ge R^2\ln(N/\eta)/(2\sigma^2)$, which is $2\ln(N/\eta)$ for accepted sets of
density $\tfrac12$ and $4\ln(N/\eta)$ for the cosine, where Hoeffding's inequality loses a factor
of two in the exponent.

The searches are adaptive, so their scores are not independent samples. A move keeps $n-1$ information
coordinates fixed and changes the codeword on the rows outside the information set. We compare each search with the Gaussian reference for the same number of codeword evaluations: $N=qN_{\mathrm{mv}}$ for
one chain, and $N=q$ times the total number of moves for the best of the line-search chains of
one seed. This reference is a model, not a separate experiment with that many independent completions.

\begin{figure}[t]
\centering
\begin{tikzpicture}
\begin{axis}[width=0.84\textwidth,height=0.60\textwidth,clip=false,
  xlabel={$\hat z(N)$: Gaussian reference for $N$ independent completions},
  ylabel={$z$: standard deviations above $\Qpr$},
  xmin=2.8,xmax=7.2,ymin=2.0,ymax=12.4,
  legend style={at={(0.5,-0.19)},anchor=north,legend columns=3,legend cell align=left,font=\scriptsize,draw=none,/tikz/every even column/.append style={column sep=6pt}},
  tick label style={font=\footnotesize},label style={font=\small}]
\addplot[gray,line width=0.6pt,domain=2.8:7.2,samples=2,forget plot]{x};
\addplot[color=green!50!black,mark=*,mark size=1.3pt,line width=0.5pt] coordinates {(4.056,3.895) (4.306,4.305) (4.565,4.533) (4.791,4.785) (5.027,4.996) (5.234,5.171) (5.452,5.441) (5.644,5.646) (5.848,5.851)};
\addlegendentry{OPI, $m=400$}
\addplot[color=green!50!black,only marks,mark=o,mark size=1.8pt,forget plot] coordinates {(6.139,5.950) (6.139,5.950) (6.139,6.056)};
\addplot[color=green!50!black,densely dotted,line width=0.4pt,forget plot] coordinates {(2.80,5.765) (7.20,5.765)};
\node[anchor=south west,font=\tiny,text=green!50!black,inner sep=1.5pt] at (axis cs:7.20,5.765) {OPI, $400$};
\addplot[color=green!50!black,mark=square*,mark size=1.3pt,line width=0.5pt] coordinates {(4.267,4.280) (4.506,4.512) (4.755,4.629) (4.973,5.168) (5.201,5.268) (5.402,5.550) (5.614,5.633) (5.801,5.882) (5.999,6.081)};
\addlegendentry{OPI, $m=1008$}
\addplot[color=green!50!black,only marks,mark=square,mark size=1.8pt,forget plot] coordinates {(6.328,6.305)};
\addplot[color=green!50!black,densely dotted,line width=0.4pt,forget plot] coordinates {(2.80,10.176) (7.20,10.176)};
\node[anchor=north west,font=\tiny,text=green!50!black,inner sep=1.5pt] at (axis cs:7.20,10.176) {OPI, $1008$};
\addplot[color=green!50!black,mark=triangle*,mark size=1.3pt,line width=0.5pt] coordinates {(4.417,4.450) (4.649,4.750) (4.892,4.980) (5.104,5.027) (5.327,5.392) (5.523,5.581) (5.731,5.787) (5.915,6.058)};
\addlegendentry{OPI, $m=2002$}
\addplot[color=green!50!black,only marks,mark=triangle,mark size=1.8pt,forget plot] coordinates {(6.248,6.523)};
\addplot[color=green!50!black,mark=diamond*,mark size=1.3pt,line width=0.5pt] coordinates {(4.565,4.616) (4.790,4.782) (5.026,5.114) (5.233,5.277) (5.451,5.411) (5.644,5.680) (5.847,5.990)};
\addlegendentry{OPI, $m=4000$}
\addplot[color=green!50!black,only marks,mark=diamond,mark size=1.8pt,forget plot] coordinates {(6.184,6.048) (6.184,6.015) (6.184,6.015)};
\addplot[color=blue!80!black,mark=*,mark size=1.3pt,line width=0.5pt] coordinates {(3.118,3.398) (3.429,3.609) (3.743,3.855) (4.010,4.025) (4.286,4.353) (4.524,4.605) (4.773,4.897) (4.989,5.020) (5.217,5.196) (5.417,5.495)};
\addlegendentry{AMA, $m=2000$}
\addplot[color=blue!80!black,only marks,mark=o,mark size=1.8pt,forget plot] coordinates {(5.778,5.694)};
\addplot[color=blue!80!black,densely dotted,line width=0.4pt,forget plot] coordinates {(2.80,2.240) (7.20,2.240)};
\node[anchor=west,font=\tiny,text=blue!80!black,inner sep=1.5pt] at (axis cs:7.20,2.240) {AMA, $2000$};
\addplot[color=blue!80!black,mark=square*,mark size=1.3pt,line width=0.5pt] coordinates {(3.118,3.066) (3.429,3.535) (3.743,3.873) (4.010,4.032) (4.286,4.410) (4.524,4.434) (4.773,4.870) (4.989,5.292) (5.217,5.255)};
\addlegendentry{AMA, $m=6000$}
\addplot[color=blue!80!black,only marks,mark=square,mark size=1.8pt,forget plot] coordinates {(5.590,5.404)};
\addplot[color=blue!80!black,densely dotted,line width=0.4pt,forget plot] coordinates {(2.80,4.915) (7.20,4.915)};
\node[anchor=north west,font=\tiny,text=blue!80!black,inner sep=1.5pt] at (axis cs:7.20,4.915) {AMA, $6000$};
\addplot[color=blue!80!black,mark=triangle*,mark size=1.3pt,line width=0.5pt] coordinates {(3.118,3.177) (3.429,3.520) (3.743,3.831) (4.010,4.095) (4.286,4.331) (4.524,4.641) (4.773,4.853) (4.989,5.076)};
\addlegendentry{AMA, $m=14000$}
\addplot[color=blue!80!black,only marks,mark=triangle,mark size=1.8pt,forget plot] coordinates {(5.224,5.280) (5.224,5.121) (5.224,5.634)};
\addplot[color=blue!80!black,densely dotted,line width=0.4pt,forget plot] coordinates {(2.80,8.373) (7.20,8.373)};
\node[anchor=west,font=\tiny,text=blue!80!black,inner sep=1.5pt] at (axis cs:7.20,8.373) {AMA, $14000$};
\addplot[color=blue!80!black,mark=diamond*,mark size=1.3pt,line width=0.5pt] coordinates {(3.118,3.100) (3.429,3.498) (3.743,4.024) (4.010,4.074) (4.286,4.253) (4.524,4.727)};
\addlegendentry{AMA, $m=20000$}
\addplot[color=blue!80!black,only marks,mark=diamond,mark size=1.8pt,forget plot] coordinates {(4.946,5.038)};
\addplot[color=blue!80!black,densely dotted,line width=0.4pt,forget plot] coordinates {(2.80,10.261) (7.20,10.261)};
\node[anchor=south west,font=\tiny,text=blue!80!black,inner sep=1.5pt] at (axis cs:7.20,10.261) {AMA, $20000$};
\addplot[color=orange!90!black,mark=*,mark size=1.3pt,line width=0.5pt] coordinates {(4.713,4.324) (4.932,4.567) (5.162,4.756) (5.364,4.974) (5.577,5.178) (5.766,5.333) (5.965,5.542) (6.142,5.737) (6.331,5.932) (6.498,6.081) (6.677,6.228)};
\addlegendentry{MPI, $m=600$}
\addplot[color=orange!90!black,only marks,mark=o,mark size=1.8pt,forget plot] coordinates {(6.975,6.853) (6.975,6.818) (6.975,6.858)};
\addplot[color=orange!90!black,densely dotted,line width=0.4pt,forget plot] coordinates {(2.80,5.513) (7.20,5.513)};
\node[anchor=west,font=\tiny,text=orange!90!black,inner sep=1.5pt] at (axis cs:7.20,5.513) {MPI, $600$};
\addplot[color=orange!90!black,mark=square*,mark size=1.3pt,line width=0.5pt] coordinates {(4.713,4.491) (4.932,4.501) (5.162,4.639) (5.364,4.964) (5.577,5.054) (5.766,5.126) (5.965,5.491) (6.142,5.548) (6.331,5.925) (6.498,6.063)};
\addlegendentry{MPI, $m=1000$}
\addplot[color=orange!90!black,only marks,mark=square,mark size=1.8pt,forget plot] coordinates {(6.804,6.519)};
\addplot[color=orange!90!black,densely dotted,line width=0.4pt,forget plot] coordinates {(2.80,7.594) (7.20,7.594)};
\node[anchor=west,font=\tiny,text=orange!90!black,inner sep=1.5pt] at (axis cs:7.20,7.594) {MPI, $1000$};
\addplot[color=orange!90!black,mark=triangle*,mark size=1.3pt,line width=0.5pt] coordinates {(4.713,4.438) (4.932,4.573) (5.162,4.627) (5.364,4.689) (5.577,4.845) (5.766,5.073) (5.965,5.168) (6.142,5.415) (6.331,5.637) (6.498,5.775)};
\addlegendentry{MPI, $m=2000$}
\addplot[color=orange!90!black,only marks,mark=triangle,mark size=1.8pt,forget plot] coordinates {(6.804,6.479) (6.804,6.825) (6.804,6.656)};
\addplot[color=orange!90!black,densely dotted,line width=0.4pt,forget plot] coordinates {(2.80,11.363) (7.20,11.363)};
\node[anchor=west,font=\tiny,text=orange!90!black,inner sep=1.5pt] at (axis cs:7.20,11.363) {MPI, $2000$};
\addplot[color=orange!90!black,mark=diamond*,mark size=1.3pt,line width=0.5pt] coordinates {(4.713,4.515) (4.932,4.670) (5.162,4.637) (5.364,4.769) (5.577,4.861) (5.766,4.770) (5.965,4.991) (6.142,5.049) (6.331,5.327)};
\addlegendentry{MPI, $m=4000$}
\addplot[color=orange!90!black,only marks,mark=diamond,mark size=1.8pt,forget plot] coordinates {(6.644,6.689)};
\addplot[color=orange!90!black,mark=pentagon*,mark size=1.3pt,line width=0.5pt] coordinates {(4.713,4.692) (4.932,4.663) (5.162,4.641) (5.364,4.698) (5.577,4.702) (5.766,4.730) (5.965,4.896) (6.142,4.891) (6.331,5.010)};
\addlegendentry{MPI, $m=6000$}
\addplot[color=orange!90!black,only marks,mark=pentagon,mark size=1.8pt,forget plot] coordinates {(6.644,6.330) (6.644,6.534) (6.644,6.675)};
\end{axis}
\end{tikzpicture}
\caption{Classical searches and DQI on the scale of \eqref{eq:zscore}. Filled marks: simulated
annealing, mean over chains and seeds of the best quality visited, one mark per run length, with
$N=qN_{\mathrm{mv}}$. Open marks: the best line-search chain of each seed at the end of the
runs, with $N=q$ times the total number of moves over all line-search chains of the seed. The horizontal coordinate is the level
$\hat z(N)$ of \eqref{eq:samplinglevel}. The gray line marks agreement with the Gaussian reference.
Dotted lines: $z(\Qdqi^{(m)})$ for the instance named at the right. The levels of the other
instances lie outside the plotted range (Table~\ref{tab:scalingz}). OPI, AMA and MPI denote optimal polynomial intersection, alternant
max-agreement and multiplicative polynomial intersection.}
\label{fig:scaling}
\end{figure}
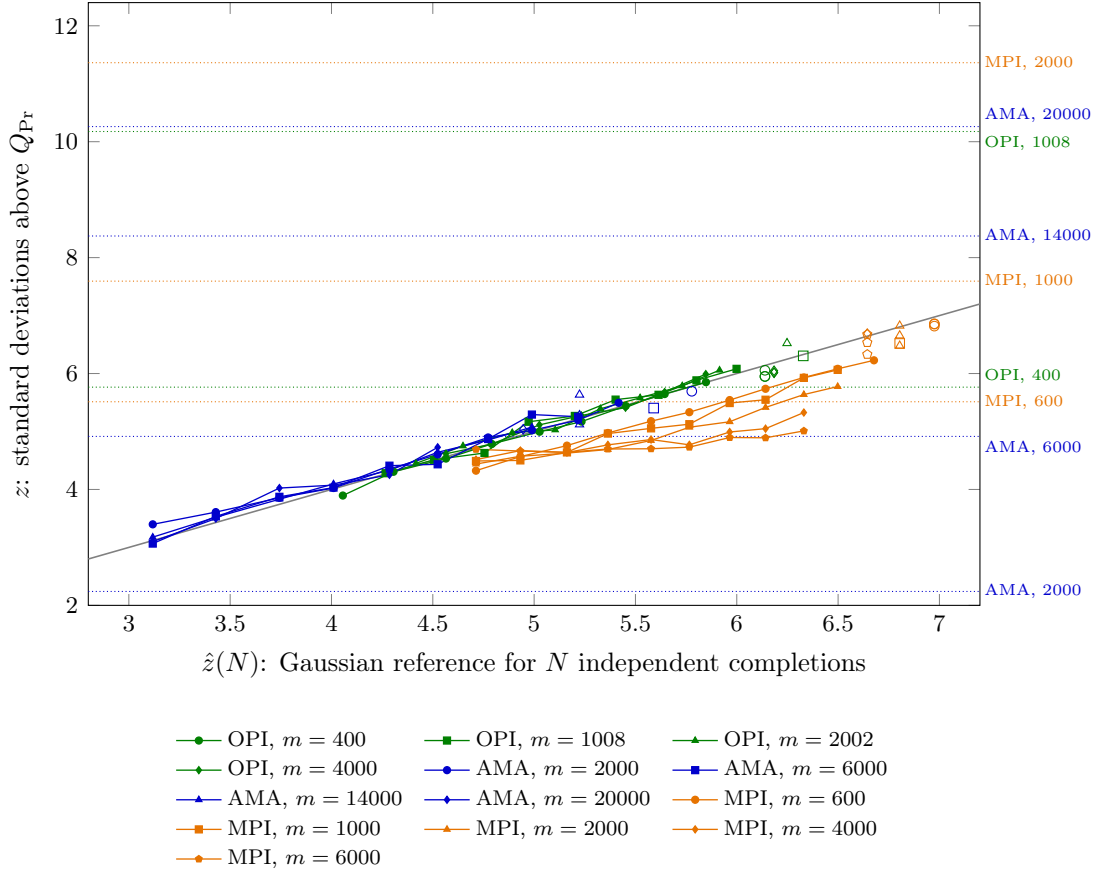

Figure~\ref{fig:scaling} and Tables~\ref{tab:scaling} and~\ref{tab:scalingz} show three things.

First, for OPI and AMA the searches perform close to the Gaussian reference. Over all run lengths of at least $10^3$ moves, the annealing means lie between
$-0.1$ and $+0.3$ standard deviations from $\hat z(N)$. Along the line-search runs, from
$10^3$ moves on, the mean over chains of the best quality so far stays between $-0.1$ and
$+0.4$ from $\hat z(qN_{\mathrm{mv}})$. The best chain of a seed fluctuates more. At the end
of the runs it lies between $-0.2$ and $+0.4$ from $\hat z(N)$. Earlier in the runs
it is up to $1.4$ above, where a single chain reaches its final quality early
(AMA at $m=2000$ and $m=14000$, after about $1.5\times10^{3}$ moves). The observed scores thus remain close to the reference, with a largest gain of
about $1.4$ standard deviations.

For MPI the line search is at the same level: along the runs
the mean over chains stays between $-0.3$ and $+0.1$ from
$\hat z(qN_{\mathrm{mv}})$, and at the end of the runs the best chains lie between
$-0.3$ and $0.0$ from $\hat z(N)$. The annealing means lie below $\hat z(N)$, by
$0.45$, $0.72$ and $1.32$ standard deviations at the longest run lengths for $m=600$, $2000$ and $6000$. At $m=6000$ the annealing means are nearly flat: $z$ rises only from $4.70$ to $5.01$ between $10^{4}$ and $10^{6}$ moves, while $\hat z(N)$ rises from $5.58$ to $6.33$.
Across these runs, the observed gains over $\Qpr$ are of order
$\sigma\,\hat z(N)/\sqrt m$ in quality, consistent with the scale of finite-size fluctuations.

Second, DQI's margin over $\Qpr$ tends to a positive constant, so $z(\Qdqi^{(m)})$ grows like
$\sqrt m$. If the searches continue to follow the reference curve, this scaling favors DQI as the length increases at fixed budget. Within our budgets the searches reach DQI's value for OPI at $m=400$, AMA at $m=2000$ and $6000$, and MPI at $m=600$, where
$z(\Qdqi^{(m)})$ is at most about $6$. They remain below it at the other lengths.

Third, the observed agreement with the reference curve suggests a conditional estimate of the effort needed for these searches to reach DQI's value: the number $N^*$ of independent completions with $\hat z(N^*)=z(\Qdqi^{(m)})$, that is,
$N^*=1/\overline\Phi(z(\Qdqi^{(m)}))$. Table~\ref{tab:scalingz} lists it together with the
corresponding number of field operations, $N^*$ times the cost of a move divided by $q$. $N^*$
estimates the effort of the tested searches if they continue to follow $\hat z(N)$ beyond the
tested range, which ends at $\hat z(N)=6.33$ for OPI and AMA and at
$6.98$ for MPI. Annealing on MPI lies below $\hat z(N)$ and
would need more. The estimate is sensitive to how far a search stays above the reference curve. A search that stays $\Delta$ standard
deviations above $\hat z(N)$ needs fewer codewords by a factor of about $e^{\Delta z}$, which is
$2^{14}$ per standard deviation at $z=10$. These values extrapolate an empirical regularity. They
are neither lower bounds on classical algorithms nor statements about algorithms that exploit the
structure of an instance.

\begin{table}[t]
\centering
\small
\begin{tabular}{lrrrrrrrl}
\toprule
 & & & & \multicolumn{3}{c}{annealing} & \multicolumn{2}{c}{line search}\\
\cmidrule(lr){5-7}\cmidrule(lr){8-9}
family & $m$ & seeds & $\Qdqi^{(m)}$ & $N_{\mathrm{mv}}$ & mean & best & moves & best chain\\
\midrule
OPI & 400 & 3 & 0.6856 & $10^{6}$ & 0.6876 & 0.7025 & $6\times10^{6}$ & 0.6900--0.6925\\
OPI & 1008 & 1 & 0.7017 & $10^{6}$ & 0.6405 & 0.6448 & $8\times10^{6}$ & 0.6438\\
OPI & 2002 & 1 & 0.7067 & $3\times10^{5}$ & 0.6139 & 0.6234 & $2.4\times10^{6}$ & 0.6189\\
OPI & 4000 & 3 & 0.7109 & $10^{5}$ & 0.5948 & 0.5990 & $8\times10^{5}$ & 0.5950--0.5953\\
\addlinespace
AMA & 2000 & 1 & 0.5226 & $3\times10^{6}$ & 0.5574 & 0.5650 & $2.4\times10^{7}$ & 0.5595\\
AMA & 6000 & 1 & 0.5285 & $10^{6}$ & 0.5306 & 0.5347 & $8\times10^{6}$ & 0.5315\\
AMA & 14000 & 3 & 0.5321 & $3\times10^{5}$ & 0.5187 & 0.5212 & $1.04\times10^{6}$ & 0.5189--0.5210\\
AMA & 20000 & 1 & 0.5329 & $3\times10^{4}$ & 0.5142 & 0.5165 & $2.4\times10^{5}$ & 0.5153\\
\addlinespace
MPI & 600 & 3 & 0.2235 & $10^{7}$ & 0.2434 & 0.2539 & $8\times10^{7}$ & 0.2598--0.2609\\
MPI & 1000 & 1 & 0.2338 & $3\times10^{6}$ & 0.2007 & 0.2043 & $2.4\times10^{7}$ & 0.2106\\
MPI & 2000 & 3 & 0.2433 & $3\times10^{6}$ & 0.1580 & 0.1639 & $2.4\times10^{7}$ & 0.1688--0.1741\\
MPI & 4000 & 1 & 0.2493 & $10^{6}$ & 0.1274 & 0.1298 & $8\times10^{6}$ & 0.1421\\
MPI & 6000 & 3 & 0.2518 & $10^{6}$ & 0.1141 & 0.1178 & $8\times10^{6}$ & 0.1257--0.1288\\
\bottomrule
\end{tabular}
\caption{Budget scaling of the generic searches: qualities at the longest completed run length.
Annealing: number $N_{\mathrm{mv}}$ of moves per chain, mean over chains and seeds of the best
quality visited, and best chain. Line search: total number of moves over the chains of one seed,
and best chain of each seed, as a range over seeds. OPI, AMA and MPI denote optimal polynomial intersection, alternant
max-agreement and multiplicative polynomial intersection.}
\label{tab:scaling}
\end{table}

\begin{table}[t]
\centering
\small
\begin{tabular}{lrrrrrrrr}
\toprule
 & & & \multicolumn{2}{c}{annealing} & \multicolumn{2}{c}{line search} & \multicolumn{2}{c}{to reach $\Qdqi^{(m)}$}\\
\cmidrule(lr){4-5}\cmidrule(lr){6-7}\cmidrule(lr){8-9}
family & $m$ & $z(\Qdqi^{(m)})$ & $z$ & $\hat z(N)$ & $z$ & $\hat z(N)$ & $\log_2N^*$ & $\log_2$ ops\\
\midrule
OPI & 400 & 5.76 & 5.85 & 5.85 & 5.99 & 6.14 & 28 & 36\\
OPI & 1008 & 10.18 & 6.08 & 6.00 & 6.31 & 6.33 & 79 & 89\\
OPI & 2002 & 14.81 & 6.06 & 5.91 & 6.52 & 6.25 & 163 & 174\\
OPI & 4000 & 21.47 & 5.99 & 5.85 & 6.03 & 6.18 & 338 & 349\\
\addlinespace
AMA & 2000 & 2.24 & 5.49 & 5.42 & 5.69 & 5.78 & 6 & 21\\
AMA & 6000 & 4.92 & 5.26 & 5.22 & 5.40 & 5.59 & 21 & 39\\
AMA & 14000 & 8.37 & 5.08 & 4.99 & 5.35 & 5.22 & 55 & 75\\
AMA & 20000 & 10.26 & 4.73 & 4.52 & 5.04 & 4.95 & 81 & 102\\
\addlinespace
MPI & 600 & 5.51 & 6.23 & 6.68 & 6.84 & 6.98 & 26 & 30\\
MPI & 1000 & 7.59 & 6.06 & 6.50 & 6.52 & 6.80 & 46 & 50\\
MPI & 2000 & 11.36 & 5.77 & 6.50 & 6.65 & 6.80 & 98 & 104\\
MPI & 4000 & 16.63 & 5.33 & 6.33 & 6.69 & 6.64 & 205 & 212\\
MPI & 6000 & 20.65 & 5.01 & 6.33 & 6.51 & 6.64 & 313 & 322\\
\bottomrule
\end{tabular}
\caption{The runs of Table~\ref{tab:scaling} on the scale of \eqref{eq:zscore}. For each search,
$z$ is the value of the tabulated quality (the mean for annealing, and for line search the mean over seeds
of the best chain) and $\hat z(N)$ the level of \eqref{eq:samplinglevel} for the codewords scored.
$N^*=1/\overline\Phi(z(\Qdqi^{(m)}))$ is the Gaussian estimate of the number of independent completions needed to reach
DQI's value. The last column gives the corresponding number of field operations under the move-cost model. Both are
extrapolations. OPI, AMA and MPI denote optimal polynomial intersection, alternant
max-agreement and multiplicative polynomial intersection.}
\label{tab:scalingz}
\end{table}

\subsection{Consequences for the candidates}\label{sec:scalingconsequences}

\paragraph{OPI.}
At $m=400$ the fixed-time portfolio of Table~\ref{tab:experiments} is below DQI. With
$10^{6}$ moves per chain the annealing mean is $0.6876$, with standard error $0.0013$, and
the best line-search chains reach 0.6900--0.6925 on the three seeds, against
$\Qdqi^{(m)}=0.6856$. The searches thus reach or exceed DQI's value, as
$z(\Qdqi^{(m)})=5.76$ against $\hat z(N)=5.85$ indicates. At $m=1008$, annealing with $10^{6}$ moves reaches a mean of $0.6405$, against $0.7017$. At $m=2002$, annealing with $3\times10^{5}$ moves reaches a mean of $0.6139$, against $0.7067$. At $m=4000$, annealing with $10^{5}$ moves reaches a mean of $0.5948$, against $0.7109$. The length at which DQI overtakes the classical searches therefore depends on the budget.

\paragraph{AMA.}
At $m=14000$, annealing with $3\times10^{5}$ moves reaches a mean quality of $0.5187$ over the three
seeds (best chain $0.5212$), and the best line-search chains of the three seeds reach 0.5189--0.5210 with $1.04\times10^{6}$ moves per seed, against $\Qdqi^{(m)}=0.5321$. Here
$z(\Qdqi^{(m)})=8.37$, and $N^*$ is about $2^{55}$ codewords, or $2^{75}$ field
operations. The quality $\Qdqi^{(m)}-0.01\simeq0.522$ has $z=5.90$, corresponding to about
$2^{29}$ independent completions under the Gaussian reference model. On seed $3101$ alone, annealing with $10^{6}$ moves gives a mean of $0.5193$ (best chain $0.5201$) and line search with $3\times10^{5}$ moves per chain a best chain of $0.5215$. At $m=2000$ and $6000$ the searches exceed DQI's value: annealing means $0.5574$ and $0.5306$ against $0.5226$ and $0.5285$.

Table~\ref{tab:scalinglengths} repeats this calculation for longer keys of the same family. Taken
at face value, the estimated cost of reaching $\Qdqi^{(m)}-0.01$ exceeds $2^{80}$ field
operations from $m=30000$, and the number of codewords alone exceeds $2^{80}$ from $m=40000$. If
a search keeps the largest lead over $\hat z(N)$ seen in our runs, $1.4$ standard deviations,
the estimate at $m=40000$ drops to $2^{69}$ codewords and $2^{93}$ field
operations, and the estimate at $m=30000$ to $2^{71}$ field operations. A field operation
costs at least one bit operation. For this reason the hardness assumptions of
Section~\ref{sec:mchardness}, whose budget is $2^{80}$ bit operations, are stated at $m=40000$
rather than at the lengths of our experiments. No instance of that length was generated: the
entries of Table~\ref{tab:scalinglengths} use the exact finite DQI value and the extrapolation
above.

\begin{table}[t]
\centering
\small
\begin{tabular}{rrrrrrrrr}
\toprule
 & & & & \multicolumn{2}{c}{reach $\Qdqi^{(m)}$} & \multicolumn{3}{c}{reach $\Qdqi^{(m)}-0.01$}\\
\cmidrule(lr){5-6}\cmidrule(lr){7-9}
$m$ & $r$ & $\ell$ & $\Qdqi^{(m)}$ & $z$ & $\log_2N^*$ & $z$ & $\log_2N^*$ & $\log_2$ ops\\
\midrule
14000 & 187 & 93 & 0.5321 & 8.37 & 55 & 5.90 & 29 & 49\\
20000 & 267 & 133 & 0.5329 & 10.26 & 81 & 7.30 & 43 & 64\\
30000 & 400 & 199 & 0.5335 & 12.81 & 123 & 9.18 & 65 & 88\\
40000 & 533 & 266 & 0.5340 & 15.02 & 168 & 10.83 & 89 & 113\\
50000 & 667 & 333 & 0.5343 & 16.92 & 212 & 12.24 & 113 & 137\\
\bottomrule
\end{tabular}
\caption{Alternant keys with $p=11$, $s=6$, $\alpha\simeq0.08$, $n=sr$ and accepted sets of
size $5$, decoded to radius $\ell$. For the exact finite DQI quality and for that quality minus
$0.01$: the value $z$ on the scale of \eqref{eq:zscore}, and the number $N^*=1/\overline\Phi(z)$
of independent completions estimated to reach it under the Gaussian reference model. The last column is the corresponding number of
field operations. The estimates extrapolate the regularity of Figure~\ref{fig:scaling} and assume
no offset from $\hat z(N)$.}
\label{tab:scalinglengths}
\end{table}

\paragraph{MPI.}
At $m=2000$, annealing with $3\times10^{6}$ moves reaches a mean cosine of $0.1580$, and the best line-search chains of the three seeds reach 0.1688--0.1741 with $2.4\times10^{7}$ moves per seed, against
$\Qdqi^{(m)}=0.2433$. At $m=6000$, annealing with $10^{6}$ moves reaches $0.1141$, and the best line-search chains of the three seeds reach 0.1257--0.1288 with $8\times10^{6}$ moves per seed,
against $0.2518$. The corresponding $N^*$ are about $2^{98}$ and $2^{313}$ codewords. The
target of Assumption~(C), $\Qdqi^{(m)}-0.01\simeq0.2418$ at $m=6000$, has $z=19.51$, corresponding to about $2^{280}$ independent completions under the Gaussian reference model. At $m=600$ the generic searches alone
exceed DQI's value.

\subsection{Searches with the alternant key}\label{sec:scalingkey}

The searches supplied with the alternant key in Section~\ref{sec:mckeyholder} made about $180$ moves per run. Script~53
runs them for $10000$ moves from the same starting assignments. For comparison, it also runs the
line search of Section~\ref{sec:scaling}, which uses only the public matrix and accepted sets, from the same assignments for the same
number of moves (Table~\ref{tab:scalingkey}). Neither search using the key improves the score of any of the three
starting assignments. The line search without the key improves two of the three.

The different moves help explain this outcome. A move using the key keeps
$r-1=186$ row values fixed and changes the others, whereas a move of the search without the key keeps
$n-1=1121$ fixed. The qualities scored by a move using the key are therefore centered near
$\rho+(1-\rho)(r-1)/m\simeq0.462$, about $9$ standard deviations below $\Qpr$.
When started from a random polynomial, the line search using the key stays below $\Qpr\simeq0.498$. In this experiment, moves using the key give the local searches no advantage. The experiment does not test other
uses of the key.

\begin{table}[t]
\centering
\small
\begin{tabular}{rrrrrr}
\toprule
 & & \multicolumn{3}{c}{with the key} & without the key\\
\cmidrule(lr){3-5}\cmidrule(lr){6-6}
seed & start & line search & annealing & line search, random start & line search\\
\midrule
3101 & 0.5144 & 0.5144 & 0.5144 & 0.4861 & 0.5162\\
3102 & 0.5151 & 0.5151 & 0.5151 & 0.4859 & 0.5156\\
3103 & 0.5171 & 0.5171 & 0.5171 & 0.4841 & 0.5171\\
\bottomrule
\end{tabular}
\caption{Alternant max-agreement at $m=14000$: best quality after $10000$ moves, one run
of each search per seed. The first three searches move along $\F_{p^s}$-lines of message
polynomials and score the current point and $4096$ random points of each line. The last moves along $\F_p$-lines of the
public code and scores their $11$ points. All but the third start from the best assignment of
the searches using only the public matrix and accepted sets. $\Qdqi^{(m)}=0.5321$.}
\label{tab:scalingkey}
\end{table}

\subsection{The lattice attack at larger block sizes}\label{sec:scalinglattice}

Section~\ref{sec:latticestudy} ran the row-subset attack within a time budget and reported the
best of the subsets examined. Script~51 fixes the number of subsets instead and reports the mean
over subsets, which estimates the score from a randomly chosen subset rather than the best subset found. It uses BKZ~2.0 with pruned enumeration~\cite{chen11bkz}, and floating-point types with
$64$ to $212$ bits of mantissa in dimensions above $160$, which avoid the numerical failures
reported in Section~\ref{sec:latticestudy}. Each subset is reduced progressively
through the block sizes of Table~\ref{tab:scalinglattice} (at $m=2000$, seeds $3102$ and $3103$ skip block size $30$), with at most eight tours
per block size and a time limit of $7200$~s per block size, tested after each tour. The
reductions at the larger block sizes have therefore not converged, and the entries marked
$\dagger$ were stopped by the time limit.

\begin{table}[t]
\centering
\footnotesize
\setlength{\tabcolsep}{4pt}
\begin{tabular}{rrrllllll}
\toprule
$m$ & $n$ & $\Qdqi^{(m)}$ & LLL & BKZ-20 & BKZ-30 & BKZ-40 & BKZ-50 & BKZ-60\\
\midrule
600 & 42 & 0.2235 & 0.203 (160) & 0.211 (130) & \textbf{0.242} (260) & \textbf{0.240} (160) & \textbf{0.249} (260) & \textbf{0.250} (260)\\
1000 & 70 & 0.2338 & 0.160 (340) & 0.173 (280) & 0.178 (180) & 0.186 (220) & 0.196 (280) & 0.204 (340)\\
2000 & 140 & 0.2433 & 0.109 (220) & 0.115 (260) & 0.126 (350) & 0.131 (500) & 0.141 (350) & 0.148$^\dagger$ (400)\\
6000 & 420 & 0.2518 & 0.087 (480) & 0.083 (480) & 0.081 (560) & 0.084 (560) & 0.086$^\dagger$ (560) & --\\
\bottomrule
\end{tabular}
\caption{Row-subset lattice attack on multiplicative polynomial intersection ($M=8191$,
$\alpha=0.07$): mean cosine on all rows, averaged over subsets, at the best of the subset sizes
$m'$ tested (in parentheses), for each block size. The sizes tested are six from $80$ to $260$ at $m=600$, six from $110$ to $340$ at $m=1000$, six from $220$ to $500$ at $m=2000$ and four from $480$ to $800$ at $m=6000$. Seed $3101$ at $m=600$, $1000$ and $6000$. At $m=2000$ the BKZ-20 and BKZ-60 entries pool seeds $3101$--$3103$, and the LLL, BKZ-30, BKZ-40 and BKZ-50 entries are from seed $3101$ alone.
Each mean is over $8$ to $28$ subsets, with standard errors between
$0.002$ and $0.007$. Taking the best of several sizes raises an entry by about one
standard error. Bold entries exceed DQI's exact finite value. $\dagger$: more than half of the
subsets reached the time limit at this block size. A dash marks a block size that was not
completed.}
\label{tab:scalinglattice}
\end{table}

At $m=600$ the mean rises from $0.203$ with LLL to $0.250$ with BKZ-$60$ and exceeds
DQI's $0.2235$ from block size $30$. At $m=1000$ it rises from $0.160$ to $0.204$ with BKZ-$60$, below DQI's $0.2338$. At $m=2000$ it rises from $0.109$ to $0.148$ with BKZ-$60$, against $0.2433$. At $m=6000$, where the time limit is reached from block size $50$, the means stay between $0.081$ and $0.087$, against $0.2518$.

\begin{figure}[t]
\centering
\begin{tikzpicture}
\begin{axis}[width=0.72\textwidth,height=0.44\textwidth,xmode=log,clip=false,
  xlabel={block size $\beta$},ylabel={mean cosine on all rows},
  xmin=1.6,xmax=2000,ymin=0.05,ymax=0.33,
  xtick={2,20,50,100,200,500,1000},xticklabels={LLL,20,50,100,200,500,1000},
  yticklabel style={/pgf/number format/fixed,/pgf/number format/precision=2},scaled y ticks=false,
  legend pos=south east,legend style={font=\scriptsize,draw=none,fill=none},
  tick label style={font=\footnotesize},label style={font=\small}]
\addplot[color=blue!80!black,only marks,mark=*,mark size=1.3pt,error bars/.cd,y dir=both,y explicit] coordinates {(2,0.2030) +- (0,0.0061) (20,0.2110) +- (0,0.0044) (30,0.2419) +- (0,0.0058) (40,0.2400) +- (0,0.0052) (50,0.2495) +- (0,0.0069) (60,0.2502) +- (0,0.0052)};
\addlegendentry{$m=600$}
\addplot[color=blue!80!black,line width=0.5pt,forget plot] coordinates {(50,0.2467) (57,0.2503) (65,0.2547) (73,0.2592) (84,0.2654) (95,0.2714) (108,0.2781) (123,0.2855) (140,0.2934) (159,0.3016) (180,0.3099) (205,0.3191) (233,0.3286)};
\addplot[color=blue!80!black,densely dotted,line width=0.4pt,forget plot] coordinates {(1.6,0.2235) (2000,0.2235)};
\node[anchor=west,font=\tiny,text=blue!80!black,inner sep=1.5pt] at (axis cs:2000,0.2235) {$600$};
\addplot[color=orange!90!black,only marks,mark=*,mark size=1.3pt,error bars/.cd,y dir=both,y explicit] coordinates {(2,0.1601) +- (0,0.0057) (20,0.1731) +- (0,0.0059) (30,0.1780) +- (0,0.0049) (40,0.1856) +- (0,0.0042) (50,0.1959) +- (0,0.0052) (60,0.2038) +- (0,0.0049)};
\addlegendentry{$m=1000$}
\addplot[color=orange!90!black,line width=0.5pt,forget plot] coordinates {(50,0.1993) (57,0.2024) (65,0.2062) (73,0.2100) (84,0.2153) (95,0.2205) (108,0.2264) (123,0.2328) (140,0.2397) (159,0.2468) (180,0.2543) (205,0.2624) (233,0.2709) (265,0.2798) (302,0.2892) (343,0.2988)};
\addplot[color=orange!90!black,densely dotted,line width=0.4pt,forget plot] coordinates {(1.6,0.2338) (2000,0.2338)};
\node[anchor=west,font=\tiny,text=orange!90!black,inner sep=1.5pt] at (axis cs:2000,0.2338) {$1000$};
\addplot[color=green!50!black,only marks,mark=*,mark size=1.3pt,error bars/.cd,y dir=both,y explicit] coordinates {(2,0.1088) +- (0,0.0038) (20,0.1152) +- (0,0.0023) (30,0.1265) +- (0,0.0026) (40,0.1306) +- (0,0.0035) (50,0.1412) +- (0,0.0024) (60,0.1485) +- (0,0.0022)};
\addlegendentry{$m=2000$}
\addplot[color=green!50!black,line width=0.5pt,forget plot] coordinates {(50,0.1470) (57,0.1494) (65,0.1522) (73,0.1553) (84,0.1595) (95,0.1635) (108,0.1681) (123,0.1732) (140,0.1787) (159,0.1844) (180,0.1903) (205,0.1969) (233,0.2038) (265,0.2111) (302,0.2189) (343,0.2268) (390,0.2352) (443,0.2439) (504,0.2530) (573,0.2625)};
\addplot[color=green!50!black,densely dotted,line width=0.4pt,forget plot] coordinates {(1.6,0.2433) (2000,0.2433)};
\node[anchor=west,font=\tiny,text=green!50!black,inner sep=1.5pt] at (axis cs:2000,0.2433) {$2000$};
\addplot[color=red!70!black,only marks,mark=*,mark size=1.3pt,error bars/.cd,y dir=both,y explicit] coordinates {(2,0.0870) +- (0,0.0027) (20,0.0832) +- (0,0.0023) (30,0.0811) +- (0,0.0029) (40,0.0836) +- (0,0.0028) (50,0.0864) +- (0,0.0022)};
\addlegendentry{$m=6000$}
\addplot[color=red!70!black,line width=0.5pt,forget plot] coordinates {(50,0.0960) (57,0.0971) (65,0.0982) (73,0.0996) (84,0.1014) (95,0.1032) (108,0.1054) (123,0.1079) (140,0.1107) (159,0.1136) (180,0.1171) (205,0.1209) (233,0.1252) (265,0.1300) (302,0.1352) (343,0.1405) (390,0.1462) (443,0.1521) (504,0.1585) (573,0.1651) (651,0.1719) (741,0.1793) (842,0.1868) (957,0.1947) (1089,0.2031) (1238,0.2117) (1407,0.2207)};
\addplot[color=red!70!black,densely dotted,line width=0.4pt,forget plot] coordinates {(1.6,0.2518) (2000,0.2518)};
\node[anchor=west,font=\tiny,text=red!70!black,inner sep=1.5pt] at (axis cs:2000,0.2518) {$6000$};
\end{axis}
\end{tikzpicture}
\caption{The row-subset lattice attack on multiplicative polynomial intersection: mean cosine on all rows at the best subset size. Marks:
measured means with their standard errors (Table~\ref{tab:scalinglattice}). Curves: the model
with $\delta_0(\beta)$ from the standard estimate, drawn while the block is shorter than the
decaying part of the profile. Dotted: DQI's exact finite value at the length given at the right.}
\label{fig:scalinglattice}
\end{figure}
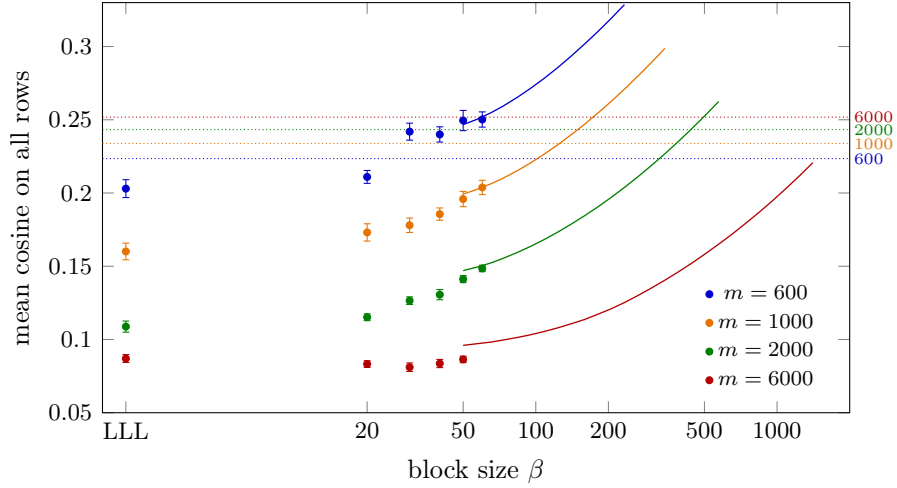

\paragraph{A model of the attack.}
Consider the lattice $\Lambda_J$ of a subset of $m'$ rows (Section~\ref{sec:latticestudy}), and
let $w_1^*,\dots,w_{m'}^*$ be the Gram--Schmidt vectors of its reduced basis. Their norms have the
shape usual for such lattices: a first group of norm $M$, then a geometric decay in which
$\log_2\lVert w_i^*\rVert$ decreases by $2\log_2\delta_0$ per index, where $\delta_0$ is the
root-Hermite factor of the reduction, and a last group of norm $1$. The product of the norms is
$M^{m'-n}$. We model the residual of Babai's nearest-plane algorithm as uniform in the box spanned
by the $w_i^*$. The directions with $\lVert w_i^*\rVert=M$ then contribute zero mean cosine, as do
the $m-m'$ rows outside the subset. Let $m''$ be the number of the other directions and
$\overline{w^2}$ the mean of $\lVert w_i^*\rVert^2$ over them. The predicted mean cosine on all
rows is
\[
  \frac{m''}{m}\,\exp\!\Bigl(-\frac{\pi^2\,\overline{w^2}}{6M^2}\Bigr),
\]
maximized over $m'$.

The model has no fitted parameter. We first evaluate it on the Gram--Schmidt norms of the
reduced bases themselves and compare with the measured means of all 128 configurations
(length, subset size, block size). At $m=600$ and $m=1000$ the root-mean-square difference is
$0.007$ and $0.006$, the size of the standard errors of the means. At $m=2000$ and
$m=6000$ it is $0.016$ and $0.021$, several standard errors. There the largest
differences, up to $0.044$, are overestimates by the model for large subsets whose reduced
bases still contain Gram--Schmidt norms above $M$, which the evaluation treats like norms equal to
$M$: LLL on subsets of $300$ rows or more at $m=2000$, and subsets of $660$ rows or more at every block size at $m=6000$. For $m'=660$ at $m=6000$ the measured mean falls from $0.065$ with LLL to $0.046$ with BKZ-$50$ while the model rises from $0.081$ to $0.087$. In those cases the model errs in favor of the attack. The measured root-Hermite
factors, from the slope of the decaying part, are $1.0222$ for LLL and $1.0186$, $1.0161$, $1.0143$, $1.0128$ and $1.0119$ for BKZ with
block sizes $20$, $30$, $40$, $50$ and $60$.

For larger block sizes we use the standard estimate
$\delta_0(\beta)=\bigl((\pi\beta)^{1/\beta}\beta/(2\pi e)\bigr)^{1/(2(\beta-1))}$~\cite{chen13thesis}.
It gives $1.0121$ and $1.0115$ at block sizes $50$ and $60$, slightly below the measured
factors, since the runs stop after eight tours. With this estimate and the profile shape above, the model predicts mean cosines of $0.247$ and $0.252$ at block sizes $50$ and $60$ for $m=600$, compared with measured values of $0.249$ and $0.250$. For $m=1000$, the predictions are $0.199$ and $0.204$, compared with $0.196$ and $0.204$. For $m=2000$, they are $0.147$ and $0.150$, compared with $0.141$ and $0.148$. At $m=6000$, the prediction at block size $50$ is $0.096$, compared with $0.086$. Thus the model agrees within about one standard error for $m\le1000$ and overestimates the measured means at $m=2000$ and $m=6000$ (Figure~\ref{fig:scalinglattice}).

The predicted mean reaches DQI's value at block sizes $71$, $125$, $278$, $439$ and $812$ for $m=800$, $1000$, $1500$, $2000$ and $3000$. At $m=600$
a block size of at most $50$ suffices. The block size grows faster than linearly with the
length. The estimate $\delta_0(\beta)$ presupposes a block shorter than the decaying part of the
profile, which has $207$, $265$, $410$, $555$ and $845$ vectors at these block sizes. The block is thus $34\%$, $47\%$, $68\%$, $79\%$ and $96\%$ of
the decaying part, against at most $33\%$ for the measured points of
Figure~\ref{fig:scalinglattice}, so the values for $m\ge1500$ are rough. For $m\ge4000$ no block shorter than the decaying part reaches DQI's
value: at the largest such block size, about $1100$ at $m=4000$ and $1500$ at
$m=6000$, the model gives $0.241$ and $0.227$, against $0.2493$ and $0.2518$, and we give
no block size. With the core-SVP estimate $2^{0.292\beta}$ for one call to the shortest-vector
oracle~\cite{becker16sieve}, the corresponding cost of one oracle call is about $2^{37}$ at $m=1000$ and,
with the caveat just stated, about $2^{128}$ at $m=2000$.

These are estimates for the mean over subsets of this attack, in this cost model. They do not
cover the selection of the best subset, which at $m=1000$ already gives single subsets above
DQI's value (up to $0.240$ against $0.2338$, at block sizes $50$ and $60$), nor attacks that combine subsets or replace Babai's algorithm by
enumeration.

A second model predicts the score of the best subset. A BKZ-$\beta$ reduced basis of the $m'$-dimensional $M$-ary lattice has Gram--Schmidt norms that fall geometrically with ratio $\delta_\beta^{-2}$ under the geometric series assumption, clamped between $1$ and $M$ (the $M$-ary profile), with product $M^{m'-n}$. We take $\delta_\beta$ from the experiments of Chen and Nguyen for $\beta<50$~\cite{chen11bkz} and from the closed form of Chen~\cite{chen13thesis} above. Babai's algorithm leaves a residual uniform along each Gram--Schmidt direction, as in the analysis of Lindner and Peikert~\cite{lindner11lwe}. The directions of norm $M$ are the axis vectors $Me_i$, so their rows receive a uniform residual and mean cosine $0$. The other rows share the residual of the sloped part, which we take as Gaussian per row, so each contributes $\exp(-2\pi^{2}s^{2}/M^{2})$ with $s^{2}$ the residual variance per row. The rows outside the subset contribute $0$ on average. Maximizing over $m'$ gives the mean cosine the attack reaches at block size $\beta$. Over all our runs the measured best is at most $1.4$ times this model value, and we multiply every prediction by $1.4$, which favors the attacker. The calibrated model puts DQI's value within reach of BKZ-$20$ at $m=600$, as observed. At $m=2000$ it needs $\beta\approx150$ on about $390$ rows, a core-SVP cost of about $2^{44}$, so this attack offers no hardness at $m=2000$. The first model concerns the mean over subsets and this one the best subset, which is why this one gives the lower cost. At $m=6000$ it needs $\beta\approx750$ on about $1170$ rows, about $2^{219}$, and at $h=17$, $m=20000$ ($n=1400$, $M=131071$) no block size up to $2000$ reaches DQI's value $0.256$. The model leaves out combinations of subsets, sieving beyond the core-SVP count and the structure of $B$, so it estimates this attack only.